\documentclass{article}
\usepackage{tablefootnote}
\usepackage{appendix}
\usepackage[utf8]{inputenc}
\usepackage{silence}
\usepackage{graphicx}
\usepackage{subcaption} 
\usepackage{amsmath,amssymb,amsthm} 
\usepackage{bm} 
\usepackage[a4paper, margin=2.5cm]{geometry} 

\usepackage[dvipsnames]{xcolor} 
\usepackage[colorlinks, linkcolor={blue}, citecolor={red}]{hyperref}

\usepackage[acronym]{glossaries} 
\glsdisablehyper 
\usepackage{notationsDMFT} 

\usepackage[symbol]{footmisc}

\title{A mathematical analysis of IPT-DMFT.}
\author{Éric Cancès\footnotemark[1], Alfred Kirsch\footnotemark[1], Solal Perrin-Roussel\footnotemark[1]}

\newacronym{dmft}{DMFT}{Dynamical Mean-Field Theory}
\newacronym{ipt}{IPT}{Iterated Perturbation Theory}
\newacronym{aim}{AIM}{Anderson Impurity Model}
\newacronym{gft}{GFT}{Generalized Fourier Transform}
\newacronym{car}{CAR}{Canonical Anti-commutation Relations}
\newacronym{gicar}{GICAR}{Gauge Invariant Canonical Anti-commutation Relations}
\newacronym{rhs}{r.h.s.}{right hand side}
\newacronym{lhs}{l.h.s.}{left hand side}
\newacronym{tddft}{TDDFT}{Time-Dependent Density Functionnal Theory}
\newacronym{kms}{KMS}{Kubo-Martin-Schwinger}
\newacronym{acp}{ACP}{Analytic Continuation Problem}
\newacronym{ba}{BU}{Bath Update}

\begin{document}

\maketitle
\abstract{We provide a mathematical analysis of the \acrfull{dmft}, a celebrated representative of a class of approximations in quantum mechanics known as \emph{embedding methods}. We start by a pedagogical and self-contained mathematical formulation of the \acrshort{dmft} equations for the finite Hubbard model. After recalling the definition and properties of one-body time-ordered Green's functions and self-energies, and the mathematical structure of the Hubbard and Anderson impurity models, we describe a specific \emph{impurity solver}, namely the \acrfull{ipt} solver, which can be conveniently formulated using Matsubara's Green's functions. Within this framework, we prove under certain assumptions that the \acrshort{dmft} equations admit a solution for any set of physical parameters. Moreover, we establish some properties of the solution(s).} 

\footnotetext[1]{CERMICS, École des Ponts, 6-8 avenue Blaise Pascal, 77455 Marne-la-Vallée, France, and Inria Paris, MATHERIALS. }

\tableofcontents

\section{Introduction}

The Dynamical Mean-Field Theory (DMFT) is an approximation method for the fermionic quantum many-body problem. It was introduced by Georges and Kotliar in 1992 and first applied to the case of the Hubbard model \cite{rozenberg_mott-hubbard_1994,georges_dynamical_1996}. It has since been extended to other settings \cite{sun_extended_2002} and coupled with Density Functional Theory (DFT) within the so-called DFT+DMFT method \cite{held_mott-hubbard_2001}. The latter is one of the reference methods for first-principle computations of electronic structures of strongly correlated materials. DMFT belongs to the class of quantum embedding methods, and has since been joined by many other methods such as Density-Matrix Embedding Theory (DMET) \cite{knizia_density_2012}, Rotationally-Invariant Slave Boson (RISB) method \cite{lechermann_rotationally_2007}, Energy-weighted DMET \cite{fertitta_energy-weighted_2019}, Quantum Embedding Theory \cite{ma_quantum_2021}, and related methods. 

At the time of writing, the mathematical analysis of quantum embedding methods is very limited. The rigorous results we are aware of are those on DMFT contained in Lindsey's PhD thesis \cite{lindsey_quantum_2019}, the ones on DMET recently obtained by the first two authors and their collaborators \cite{cances_mathematical_2023}, and a few others with a numerically oriented approach such as \cite{faulstich_pure_2022,wu_enhancing_2020}.

The purpose of this article is to establish a rigorous mathematical formulation of the \acrshort{dmft} equations and to prove, in particular, the existence of a solution to the \acrshort{dmft} equations for the Hubbard model within the \acrfull{ipt} approximation \cite{georges_dynamical_1996}. This relies on the extension of some results from \cite[Part VII]{lindsey_quantum_2019} and is based upon a reformulation of the \acrshort{dmft} equations as a fixed-point problem in the space of probability measures on the real line.

\medskip

In the language of linear algebra, solving the fermionic quantum many-body problem consists in computing some spectral properties of a Hermitian matrix $\Hamiltonian \in \ComplexNumbers^{M \times M}$, the Hamiltonian of the system, such as its ground-state energy (i.e. its lowest eigenvalue), or the partition function 
$$
\PartitionFunction:=\exp \left( -\beta (\Hamiltonian - \ChemicalPotential  \TotalNumberOperator) \right), 
$$
where $\beta=\frac{1}{k_{\rm B}T}$ is the inverse temperature, $\ChemicalPotential \in \RealNumbers$ is the chemical potential, and $\TotalNumberOperator$ is the number operator, as well as derivatives of $\PartitionFunction$ with respect to $\beta$, $\ChemicalPotential$ or parameters of $\Hamiltonian$. 

The difficulty is that the size $M$ of the matrix $\Hamiltonian$ can be huge (up to $10^{30}$ or more in some applications). Fortunately, the Hamiltonian $\Hamiltonian$ has specific properties, allowing one to use taylored methods. Indeed, in most applications, $\Hamiltonian$ is the matrix of a Hamiltonian operator acting on a fermionic Fock space $\FockSpace$, containing only one- and two-body terms, and satisfying symmetry properties (particle number conservation, spin, and possibly space, isospin, or time-reversal symmetries). Identifying the one-body state space $\OneParticleSpace[]$ with $\ComplexNumbers^{2L}$, it holds $M=2^{2L}$ and 
$$
\FockSpace = \bigoplus_{N=0}^L \OneParticleSpace[N], 
$$
where the $N$-particle sector $\OneParticleSpace[N]=\bigwedge^N \OneParticleSpace[]$ of the Fock space is of dimension $\begin{pmatrix} 2L \\ N \end{pmatrix}$.
In this decomposition, $\TotalNumberOperator$ is a block diagonal operator, the block corresponding to $\OneParticleSpace[N]$ being equal to $N$ times the identity matrix. If $\Hamiltonian$ is particle-number conserving, then it is also block-diagonal in this decomposition (equivalently $\Hamiltonian$ and $\TotalNumberOperator$ commute). If it only contains one- and two-body terms, then $\Hamiltonian$ has a compact representation in the second quantization formalism involving a Hermitian matrix $H^0 \in \ComplexNumbers^{2L \times 2L}$ and a fourth-order tensor $V \in \ComplexNumbers^{2L \times 2L \times 2L \times 2L}$. Spin, space, isospin, or time-reversal symmetries allow one to further reduce the complexity of the representation of $\Hamiltonian$ and refine its block diagonal structure. Still, solving the quantum many-body problem remains extremely challenging.

\medskip

Quantum embedding methods can be seen as domain decomposition methods in the Fock space, using a partition of the $L$ ``sites'' (also called orbitals) of the model into $P$ non-overlapping clusters $(\HubbardVertices_p)_{1 \le p \le P}$ of cardinalities $L_p:=|\HubbardVertices_p|$. Without loss of generality, we can assume that the first cluster consists of the first $L_1$ orbitals, the second cluster of the next $L_2$ orbitals, and so on. To each cluster is associated an impurity model, a quantum many-body problem set on the $L_p$ sites of the cluster, as well as on virtual sites called bath orbitals. In DMET, the number of bath orbitals is exactly equal to $L_p$ so that the impurity quantum many-body problem is of size $M_p:=2^{4L_p}$. In practice, DMET impurity problems are solved either by brute-force diagonalization (full CI) if $L_p$ is not too large, or by low-rank tensor methods (e.g. Density Matrix Renormalization Group, DMRG \cite{white_density_1992}). In \acrshort{dmft}, the impurity problem can be much larger, but the impurity Hamiltonian has the relatively simple form of an \acrfull{aim}: within each of the $P$ impurity models, bath orbitals do not contribute to two-body interactions, and only interact with cluster orbitals via one-body interactions. It can be shown that the \acrshort{aim} associated with the $p$-th cluster can be completely described by the restrition of $\Hamiltonian$ to the $p$-th cluster's orbitals and a hybridization function $\Hybridization_p: \ComplexNumbers \setminus \RealNumbers \to \ComplexNumbers^{L_p \times L_p}$. \acrshort{aim}s are usually solved in practice either by a quantum Monte Carlo method~\cite{rubtsov_continuous-time_2005}, or by an approximate solver such as the IPT (Iterative Perturbation Theory) solver \cite{georges_dynamical_1996} considered in this article. The IPT solver was introduced in the seminal paper \cite{zhang_mott_1993}, and is still used to study very challenging systems such as moiré heterobilayers \cite{tan_doping_2023}.

\medskip

Quantum embedding methods are self-consistent  theories: the $P$ impurity problems are coupled through a mean-field defined on the whole quantum system with $L$ orbitals. 

In DMET, the role of the mean-field is played by an approximation $D \in \ComplexNumbers^{2L \times 2L}$ of the ground-state one-body density matrix (1-RDM) of the system. The matrix $D$ allows one to define an impurity problem for each cluster, and the self-consistent condition is that for each cluster $p$, the diagonal block of~$D$ corresponding to this cluster agrees with the restriction of the exact ground-state 1-RDM of the $p$-th impurity problem to the cluster $p$. It is expected that at self-consistence the diagonal blocks of $D$ corresponding to the cluster decomposition are good approximations of the diagonal blocks of the exact ground-state 1-RDM of the whole system \cite{cances_mathematical_2023}.

In DMFT, the role of the mean-field is played by an approximation $G$ of the exact one-body Green's function \cite{martin_interacting_2016} associated with some equilibrium state, usually the ground-state of $\Hamiltonian$ in the $N$-particle sector, or a canonical or grand-canonical thermodynamical equilibrium state. One-body Green's functions can be represented by analytic functions $G : \ComplexNumbers \setminus \RealNumbers \to \ComplexNumbers^{2L \times 2L}$ and are thus computationally tractable objets for values of $L$ up to a few thousands. The function $G$ is a particular holomorphic extension of the Fourier transform of the time-ordered Green's function. Loosely speaking, the latter is an equilibrium time-correlation function obtained by creating (resp. annihilating) a particle at time $t_0=0$ (resp. at $t<0$), letting the system evolve from $t_0$ to $t$ (resp. from $t$ to $t_0$), and annihilating the extra particle (resp. restoring the missing particle) at time $t$ (resp. at time $t_0$). The exact one-body Green's function contains a lot of valuable information about the quantum system under investigation. In particular, the 1-RDM of the equilibrium state, hence the expectation value of any one-body observable, can be easily extracted from it. The same holds true for the average energy,  thanks to Galitski-Migdal's formula~\cite{gontier_contributions_2015,martin_interacting_2016}. Also, the poles of the analytic continuation of $G$ to the real-axis correspond to the one-particle excitation energies measured in photoemission and inverse-photoemission spectroscopies~\cite{zhou_unraveling_2020}. Remarkably, the exact Green's function $\GreensFunction^0$ of a non-interacting system, i.e. of a many-body Hamiltonian which is the second quantization of a one-body hamiltonian $\SecondQuantization{\NIHamiltonian}$ is simply the resolvent of $\NIHamiltonian$: $\GreensFunction^0(z)= (z-\NIHamiltonian)^{-1}$, whatever the reference equilibrium state. The self-energy of an interacting system with Hamiltonian $\Hamiltonian = \SecondQuantization{\NIHamiltonian} + \InteractingHamiltonian$, where $\InteractingHamiltonian$ accounts for the two-body interactions, is the function $\SelfEnergy : \ComplexNumbers \setminus \RealNumbers \to \ComplexNumbers^{2L \times 2L}$ defined by
\begin{equation}
\SelfEnergy(z) = \GreensFunction^0(z)^{-1} - \GreensFunction(z)^{-1} \quad \mbox{or equivalently} \quad 
\GreensFunction(z) = (z-H^0-\SelfEnergy(z))^{-1}.    
\end{equation}
\acrshort{dmft} consists in 
\begin{itemize}
    \item approximating the exact self-energy of the whole system by a block diagonal self-energy $\SelfEnergy = \mbox{block-diag}(\SelfEnergy_1, \dots, \SelfEnergy_\PartitionSize)$, with $\SelfEnergy_p : \ComplexNumbers \setminus \RealNumbers \to \ComplexNumbers^{2L_p \times 2L_p}$ compatible with the cluster decomposition. This condition is sometimes called the DMFT approximation;
    \item imposing the self-consistent conditions that for each cluster 
    \begin{itemize}
        \item the self-energy $\SelfEnergy_p$ agrees with the restriction to the cluster of the exact self-energy of the associated AIM. 
        \item the restriction to the cluster of the approximate Green's function of the whole system agrees with the restriction to the cluster of the exact Green's function of the AIM. This condition is often referred to as the self-consistent condition.
    \end{itemize}
\end{itemize}
In practice, the DMFT equations are solved by fixed-point iterations. The input of iteration $n$ is a collection of $P$ hybridization functions $(\Hybridization_p^{(n)})_{1 \le p \le P}$. At step 1, the $P$ AIM problems with hybridization functions $\Hybridization_p^{(n)}$ are solved in parallel, in order to compute $P$ cluster self-energies $\SelfEnergy_p^{(n)}$, yielding an approximation $\SelfEnergy^{(n)} = \mbox{block-diag}(\SelfEnergy_1^{(n)}, \cdots \SelfEnergy_P^{(n)})$  of the self-energy of the whole system. At step 2, the above two self-consistent conditions are combined yielding a new set $(\Hybridization_p^{(n+1)})_{1 \le p \le P}$ of hybridization functions. The DMFT iteration scheme can therefore be sketched as
$$
\Hybridization^{(n)}:=(\Hybridization_p^{(n)})_{1 \le p \le P} \mathop{\longrightarrow}^{f^{\rm AIM}} \SelfEnergy^{(n)}:=(\SelfEnergy_p^{(n)})_{1 \le p \le P} \mathop{\longrightarrow}^{f^{\rm SC}}
\Hybridization^{(n+1)}:=(\Hybridization_p^{(n+1)})_{1 \le p \le P},
$$
or written in the more compact form
\begin{equation} \label{eq:DMFT_loop_1}
\Hybridization^{(n+1)} = f^{\rm DMFT}\left(\Hybridization^{(n)}\right).
\end{equation}
Of course, this basic self-consistent loop can be stabilized and accelerated using e.g. damping and Anderson-Pulay extrapolation methods. In this article, we forego an in-depth analysis of the iterative scheme and its convergence, opting instead to direct our attention towards a fundamental inquiry: the existence of solutions within the \acrshort{dmft} equations. Specifically, we address the question of the existence of a fixed-point of the DMFT map $f^{\rm DMFT}$, a critical aspect which, to our knowledge remains unestablished in the current literature.

\medskip

This article is organized as follows. 
In Section~\ref{sec:MathematicalFramework}, we provide a mathematical introduction to \acrshort{dmft} for the Hubbard model aimed at being accessible to readers unfamiliar with this theory. The Hubbard model provides insights into the behavior of electrons in strongly correlated systems. Its integration within the \acrshort{dmft} framework offers a powerful tool for understanding the interplay between electron-electron interactions in finite structures (truncation of a lattice for instance), shedding light on phenomena such as metal-insulator transitions and high-temperature superconductivity. We recall the basics of second quantization formalism, the formulation of the Hubbard and Anderson impurity models, the definitions of one-body Green's functions, self-energies, and hybridization functions, and the precise formulation of the \acrshort{dmft} equations.  
In Section~\ref{sec:MainResults}, we state our main results. They are based on the observation that the key mathematical objects involved in \acrshort{dmft} (exact and approximate one-body Green's functions and self-energies, hybridization functions) are all negatives of Pick functions. Recall that scalar Pick functions are analytic functions from the open upper-half plane to the closed upper-half  plane~\cite{nevanlinna_uber_1919}, \cite{pick_uber_1915}. An interesting property of scalar Pick functions, which we use extensively in our analysis, is that any Pick function admits an integral reprentation involving a positive Borel measure on $\RealNumbers$, called its Nevanlinna-Riesz measure~\cite{nevanlinna_uber_1919}. Analogous properties hold true for matrix-valued Pick functions~\cite{gesztesy_matrixvalued_2000}. For the Hubbard model with a finite number of sites, the exact Green's function and self-energy can be extended to meromorphic functions on $\ComplexNumbers$ with finite numbers of poles, and are therefore represented by discrete Nevanlinna-Riesz measures with finite support. We then focus on the paramagnetic single-site translation invariant \acrshort{ipt}-\acrshort{dmft} approximation of the Hubbard model, for which $\PartitionSize=L$ and $L_1=\cdots = L_\PartitionSize=1$. We show that these equations have no solutions in the class of (negatives of) Pick functions with discrete Nevanlinna-Riesz measures of finite support, but do have solutions in the set of (negatives of) Pick functions. More precisely, equation \eqref{eq:DMFT_loop_1} has a translation invariant fixed point $(\Hybridization,\dots,\Hybridization)$, $\Hybridization$ being the negative of a Pick function whose Nevanlinna-Riesz measure has the form $c \nu$, where $c \in \RealNumbers_+$ is a fixed constant only depending on the matrix $\NIHamiltonian$, and $\nu$ a Borel probability measure on $\RealNumbers$. To obtain the latter result, we show that the \acrshort{ipt}-\acrshort{dmft} iteration map $f^{\rm DMFT}$ in \eqref{eq:DMFT_loop_1} can be rewritten as a map $\DMFTmap : \ProbabilityMeasures \to \ProbabilityMeasures$, which is continuous for the weak topology.  
We conclude by checking that the Schauder-Singbal's fixed-point theorem \cite{shapiro_fixed-point_2016} can be applied to this setting.



\section{DMFT of the Hubbard model}
\label{sec:MathematicalFramework}

We provide in this section a mathematical description of the models and quantities of interest involved in \acrfull{dmft} for the Hubbard model. We first recall the definitions of one-body Green's functions and the self-energy. We then introduce the Hubbard model and the Anderson Impurity Model (AIM). Next, we derive the \acrshort{dmft} equations and finally present the \acrfull{ipt} solver, which is the approximate impurity solver considered in this work.

\subsection{One-body Green's functions and the self-energy} \label{subsec:OneBodyTimeOrderedGreensFunctions}
One-body Green's functions are key objects in DMFT. To avoid technicalities, we will define Green's functions in a finite-dimensional setting and assume that the one-body state space is a finite-dimensional Hilbert space $(\OneParticleSpace,\HermitianProduct{\cdot}{\cdot})$, ${\rm dim}(\OneParticleSpace)=2L \in \Integers^*$. We refer to e.g. \cite{cances_mathematical_2016} for a mathematical introduction to Green's functions in an infinite-dimensional setting. The associated Fock space
\begin{equation*}
\FockSpace=\bigoplus_{n=0}^{2L}\bigwedge^n \OneParticleSpace,
\end{equation*}
where the $n$-particle sector $\displaystyle\bigwedge^n \OneParticleSpace$ is the anti-symmetrized tensor product of $n$ copies of $\OneParticleSpace$, is then of dimension $2^{2L}$. 
Given a one-particle state $\OneParticleState \in \OneParticleSpace$, we denote by $\AnnihilationOperator[\OneParticleState]$ (resp. $ \CreationOperator[\OneParticleState]$)  the usual annihilation (resp. creation) operator defined on $\FockSpace$ (see e.g. \cite{bratteli_operator_1997}), which satisfy the \acrfull{car}:
\begin{equation*}
\forall \OneParticleState, \OneParticleState' \in \OneParticleSpace, \quad \AntiCommutator{\AnnihilationOperator[\OneParticleState]}{\AnnihilationOperator[\OneParticleState']}=\AntiCommutator{\CreationOperator[\OneParticleState]}{\CreationOperator[\OneParticleState]}=0,\quad \AntiCommutator{\AnnihilationOperator[\OneParticleState]}{\CreationOperator[\OneParticleState']}=\HermitianProduct{\OneParticleState}{\OneParticleState'}
\end{equation*}
where $\AntiCommutator{\Operator}{\Operator'}=\Operator \Operator' + \Operator'\Operator$ is the anti-commutator of the two operators $\Operator,\Operator' \in \LinearOperator(\FockSpace)$.

\paragraph{Equilibrium states.} A \emph{state} $\EquState$ is a linear form on the set of operators $\LinearOperator(\FockSpace)$, which is positive ($\EquState(\Operator^\dagger \Operator) \geq 0$) and normalized (i.e. $\sup \{ \Modulus{\EquState(\Operator)}, \Norm{\Operator}=1\}=1$).
In the finite-dimensional case, any state $\EquState$ can be represented by a unique self-adjoint operator $\DensityOperator \in \SelfAdjointOperator(\FockSpace)$  such that for all $\Operator \in \LinearOperator(\FockSpace)$, $\EquState(\Operator)=\Trace(\DensityOperator \Operator)$. The operator $\DensityOperator$ is positive and satisfies ${\rm Tr}(\DensityOperator)=1$. It is called the \emph{density operator} associated to the state $\EquState$. For an isolated quantum system described by a time-independent  Hamiltonian  $\Hamiltonian \in \SelfAdjointOperator(\FockSpace)$, an equilibrium state corresponds to a stationary solution to the quantum Liouville equation
$$
i \frac{d\DensityOperator}{dt}(t) = [\Hamiltonian,\DensityOperator(t)],
$$
where $[\Operator,\Operator']=\Operator\Operator' - \Operator'\Operator$ is the commutator of $\Operator,\Operator'\in\LinearOperator(\FockSpace)$.
It follows that a state is an equilibrium state if and only if its density $\DensityOperator$ commutes with the Hamiltonian $\Hamiltonian$, namely $\Commutator{\Hamiltonian}{\DensityOperator}=0$. Important examples of equilibrium states are thermal and osmotic equilibrium states known as Gibbs states, as well as ground and excited states of $\Hamiltonian$ with a prescribed number of particles (for particle-number conserving Hamiltonians). 

\medskip

\paragraph{One-body Green's functions.} The \emph{one-body time-ordered Green's functions} are then defined as follows:

\begin{definition}[One-body time-ordered Green's function] Given a Hamiltonian $\Hamiltonian \in \SelfAdjointOperator(\FockSpace)$ and an associated equilibrium state $\EquState$, one defines the $\LinearOperator(\OneParticleSpace)$-valued function $\TimeGreensFunction: \RealNumbers \to \LinearOperator(\OneParticleSpace)$, known as \emph{one-body time-ordered Green's function}, so that  
$i\TimeGreensFunction(t)$ is the operator represented by the sesquilinear form 
\begin{equation}\label{eq:DefinitionTimeOrderedGreensFunctions}
\HermitianProduct{\OneParticleState}{(i\TimeGreensFunction(t))\OneParticleState'}= \CharacteristicFunction{\RealNumbers_{+}} (t) \, \EquState(\HeisenbergPicture{\AnnihilationOperator[\OneParticleState]}(t)\CreationOperator[\OneParticleState']) - \CharacteristicFunction{\RealNumbers_{-}^*} (t) \, \EquState(\CreationOperator[\OneParticleState'] \HeisenbergPicture{\AnnihilationOperator[\OneParticleState]}(t))
\end{equation}
where for all $\Operator \in \LinearOperator(\FockSpace), \HeisenbergPicture{\Operator}: \RealNumbers \ni t \mapsto e^{it\Hamiltonian} \Operator e^{-it\Hamiltonian}$ is the \emph{Heisenberg picture} of $\Operator$ and $\CharacteristicFunction{A}$ is the characteristic function of the set $A$.
\end{definition}

Let us comment on the terminology. First, the term ``body'' encompasses ``particle'' and ``hole'': the first term of the \acrlong{rhs} of \eqref{eq:DefinitionTimeOrderedGreensFunctions} can be interpreted as describing the propagation from $t_0=0$ to $t > 0$ of a particle added to the system at $t_0=0$, while the second term can be interpreted as the propagation from $t < 0$ to $t_0=0$ of a hole created at $t<0$. Second, it is ``time-ordered'': the \acrshort{rhs} of \eqref{eq:DefinitionTimeOrderedGreensFunctions} can be rewritten as
\begin{equation*}
\EquState\left(\mathcal{T}\left(\HeisenbergPicture{\AnnihilationOperator[\OneParticleState]},\HeisenbergPicture{\CreationOperator[\OneParticleState']}\right)(t,0)\right)
\end{equation*}
where for all operators-valued functions $\RealNumbers \ni t \mapsto \Operator(t) \in \LinearOperator(\FockSpace),\RealNumbers \ni t \mapsto \Operator'(t) \in \LinearOperator(\FockSpace)$, the fermionic \emph{time-ordered product} $\mathcal{T}(\Operator,\Operator')$ is the operator-valued function $\RealNumbers^2 \to \LinearOperator(\FockSpace)$ defined as
\begin{equation*}
\mathcal{T}(\Operator,\Operator')(t,t')=\left\lbrace 
\begin{matrix}
\Operator(t) \Operator'(t') \text{ if t $\geq$ t'} \\
-\Operator'(t')\Operator(t) \text{ otherwise,}
\end{matrix}\right.
\end{equation*}
where the minus sign is specific to the fermionic case. Up to a sign, it is the product of the operators applied in the order of increasing time.

The $i$ prefactor in the \acrlong{lhs} of \eqref{eq:DefinitionTimeOrderedGreensFunctions} is a convention that facilitates the expression of the results to come, especially Proposition \ref{prop:NonInteractingGreensFunction}. 

Finally, note that $\TimeGreensFunction$ is real-analytic on $(-\infty,0)\cup(0,+\infty)$ with a $-i\Identity$ jump at $t=0$ due to the \acrshort{car}.

As $\FockSpace$ is finite-dimensional, the Green's function can be expanded in a joint orthonormal eigenbasis $\FockSpaceBasis$ of $\DensityOperator$ and $\Hamiltonian$, leading to the Källén-Lehmann (KL) representation \cite{kallen_definition_1952,lehmann_uber_1954}  
\begin{equation} \label{eq:KL}
\forall \OneParticleState,\OneParticleState' \in \OneParticleSpace, \quad \HermitianProduct{\OneParticleState}{i\TimeGreensFunction(t)\OneParticleState'}=\sum_{\FockSpaceState,\FockSpaceState' \in \FockSpaceBasis} e^{it(E_{\FockSpaceState}-E_{\FockSpaceState'})}\HermitianProduct{\FockSpaceState}{\AnnihilationOperator[\OneParticleState]\FockSpaceState'} \HermitianProduct{\FockSpaceState'}{\CreationOperator[\OneParticleState']\FockSpaceState} \left( \rho_\FockSpaceState \CharacteristicFunction{\RealNumbers_+}(t) -  \rho_{\FockSpaceState'} \CharacteristicFunction{\RealNumbers_-^*}(t)\right),
\end{equation}
where $\forall \FockSpaceState \in \FockSpaceBasis, \Hamiltonian \FockSpaceState = E_{\FockSpaceState} \FockSpaceState$ (with $E_{\FockSpaceState} \in \RealNumbers$) and $\DensityOperator \FockSpaceState=\rho_{\FockSpaceState} \FockSpaceState$  (with $\rho_{\FockSpaceState} \in \RealNumbers_+$, $\sum_{\FockSpaceState \in \FockSpaceBasis} \rho_{\FockSpaceState}=1$).

Other types of Green's functions are encountered in the physics literature, notably retarded/advanced Green's functions. These objects encode the same information on the spectral properties of~$\Hamiltonian$ as the time-ordered Green's function, but this information is stored in a different way. A suitable way to highlight this information is to consider specific holomorphic extensions to the complex plane of the time-Fourier transform of these Green's functions \cite{martin_interacting_2016,cances_mathematical_2016}. In the case of the time-ordered one-body Green's function, the suitable holomorphic extension is provided by the generalized Fourier transform introduced by 
Titchmarsh~\cite{titchmarsh_introduction_1948}.

\begin{definition}[\acrfull{gft}]
\label{def:GreensFunctionGFT}
The \acrfull{gft} of the one-body time-ordered Green's function $\TimeGreensFunction$ is the analytic function on the upper-half plane $\GreensFunction:\UpperHalfPlane \to \LinearOperator(\OneParticleSpace)$, also called a (one-body) Green's function, defined by
\begin{align}
\forall z \in \UpperHalfPlane, \GreensFunction(z)&= \GreensFunction_+(z)+\GreensFunction_{-}(\overline{z})^\dagger
\label{eq:DefinitionGFTGreensFunction}\end{align}
with
\begin{align*}
\forall z \in \UpperHalfPlane:=\{ z \in \ComplexNumbers \ | \ \Im(z)>0 \}, \quad \GreensFunction_+(z)&=\int_{\RealNumbers_+}e^{izt} \TimeGreensFunction(t) dt, \\ 
\forall z \in \LowerHalfPlane:=\{ z \in \ComplexNumbers \ | \ \Im(z) < 0 \}, \quad\GreensFunction_-(z)&=\int_{\RealNumbers_-}e^{izt} \TimeGreensFunction(t) dt.
\end{align*}
\end{definition}
Note that the Green's function $\GreensFunction$ can be extended to $\ComplexNumbers\setminus\RealNumbers$ by reflection, namely by setting 
\begin{equation}
\forall z \in \LowerHalfPlane, \quad \GreensFunction(z)=\GreensFunction(\bar{z})^\dagger. \label{eq:extension_C-}
\end{equation}

By construction, $\GreensFunction(z)$ is analytic on $\UpperHalfPlane$. In addition, it follows from the KL representation~\eqref{eq:KL} that 
\begin{equation}
    \label{eq:KL_C+}
    \forall z \in \UpperHalfPlane, \quad \forall \OneParticleState,\OneParticleState' \in \OneParticleSpace, \quad \HermitianProduct{\OneParticleState}{G(z)\OneParticleState'}=\sum_{\FockSpaceState,\FockSpaceState' \in \FockSpaceBasis} \frac{\rho_{\FockSpaceState}+\rho_{\FockSpaceState'}}{z + (E_{\FockSpaceState}-E_{\FockSpaceState'})} \HermitianProduct{\FockSpaceState}{\AnnihilationOperator[\OneParticleState]\FockSpaceState'} \HermitianProduct{\FockSpaceState'}{\CreationOperator[\OneParticleState']\FockSpaceState}.
\end{equation}
An important observation is that 
\begin{equation} \label{eq:G_Pick}
   \forall z \in \UpperHalfPlane, \quad \forall \OneParticleState \in \OneParticleSpace \setminus \{0\}, \quad \Im\left(\HermitianProduct{\OneParticleState}{G(z)\OneParticleState}\right)=- \Im(z) \sum_{\FockSpaceState,\FockSpaceState' \in \FockSpaceBasis} \frac{\rho_{\FockSpaceState'}+\rho_\FockSpaceState}{|z+(E_{\FockSpaceState}-E_{\FockSpaceState'})|^2} |\HermitianProduct{\FockSpaceState}{\AnnihilationOperator[\OneParticleState]\FockSpaceState'}|^2 <0,
\end{equation}
which shows in particular that $\GreensFunction(z)$ is invertible for all $z \in \UpperHalfPlane$.

Up to now, we have not specified the Hamiltonian $\Hamiltonian$; in the sequel, we will assume that it is of the form
\begin{equation}
\label{eq:HamiltonianNIAndIDecomposition}
\Hamiltonian=\SecondQuantization{\NIHamiltonian} + \InteractingHamiltonian, \quad \NIHamiltonian \in \SelfAdjointOperator(\OneParticleSpace), \quad \InteractingHamiltonian \in \SelfAdjointOperator(\FockSpace)
\end{equation}
where $\SecondQuantization{\NIHamiltonian}$ is the second quantization of the one-particle Hamiltonian $\NIHamiltonian \in \SelfAdjointOperator(\OneParticleSpace)$ (see e.g. \cite{bratteli_operator_1997}) and  $\InteractingHamiltonian \in \SelfAdjointOperator(\FockSpace)$ some interaction Hamiltonian. We say that $\Hamiltonian$ is \emph{non-interacting} if $\InteractingHamiltonian=0$.

Depending on the formalism of interest, one can consider the grand canonical Hamiltonian $\Hamiltonian'=\Hamiltonian-\ChemicalPotential\TotalNumberOperator$ without loss of generality.

The following is an essential property of the Green's function, to which it owes its name: the Green's function of a non-interacting system in an equilibrium state is the resolvent of the one-particle Hamiltonian.

\begin{proposition}[Non-interacting Green's function]
\label{prop:NonInteractingGreensFunction}
Let $\NIHamiltonian\in\SelfAdjointOperator(\OneParticleSpace)$  be a one-particle Hamiltonian and $\EquState$ an equilibrium state  of the non-interacting Hamiltonian $\SecondQuantization{\NIHamiltonian}$. The Green's function $\GreensFunction^0$ of $\SecondQuantization{\NIHamiltonian}$ associated to $\EquState$ is the resolvent of $\NIHamiltonian$: 
\begin{equation} \label{eq:G0}
\forall z \in \UpperHalfPlane, \quad \GreensFunction^0(z)=(z-\NIHamiltonian)^{-1}.
\end{equation}
In particular, $\GreensFunction^0$ is independent of $\EquState$.
\end{proposition}

The proof of this classical result is recalled in Section \ref{sec:ProofNIGreensFunction}. It is a consequence of the fact that the one-body Green's function $\TimeGreensFunction^0$ of the non-interacting Hamiltonian $\SecondQuantization{\NIHamiltonian}$ is solution to the equation
\begin{equation}
\left(i\frac{d}{dt} -\NIHamiltonian \right) \TimeGreensFunction = \delta \Identity \quad \text{ in } \mathcal{D}'(\RealNumbers;\LinearOperator(\OneParticleSpace))
\label{eq:PDEGreensFunction}
\end{equation}
so that the time-ordered Green's function $\TimeGreensFunction^0$ is indeed a Green's function of the linear differential operator $i\frac{d}{dt} - \NIHamiltonian$. The various avatars of the Green's function (retarded/advanced) also satisfy this equation in $\mathcal{D}'(\RealNumbers^*;\LinearOperator(\OneParticleSpace))$, but with different boundary conditions at infinity and jumps at $t=0$.  The properties \eqref{eq:G0}-\eqref{eq:PDEGreensFunction} hold \emph{only} for non-interacting Hamiltonians.

\medskip

\paragraph{Self-energy.}
For interacting systems, the general relation between the Hamiltonian and the Green's function is more involved: the discrepancy with the non-interacting case is characterized by the \emph{self-energy}.

\begin{definition}[Self-energy]
\label{def:SelfEnergy}
The \emph{self-energy} $\SelfEnergy:\UpperHalfPlane \to \LinearOperator(\OneParticleSpace)$ of an Hamiltonian $\Hamiltonian$ of the form \eqref{eq:HamiltonianNIAndIDecomposition} associated to an equilibrium state $\EquState$ of $\Hamiltonian$ is defined by
\begin{equation} \label{eq:self-energy}
\forall z \in \UpperHalfPlane, \quad \SelfEnergy(z)=\GreensFunction^0(z)^{-1} - \GreensFunction(z)^{-1},
\end{equation}
where the \emph{non-interacting Green's function} $\GreensFunction^0$ is the Green's function of any equilibrium state of $\SecondQuantization{\NIHamiltonian}$.
\end{definition}

Recall that $G^0(z)^{-1} = z-H^0$ \eqref{eq:G0} and that $G(z)$ is invertible in view of \eqref{eq:G_Pick}.
Let us emphasize that the definition of the self-energy only depends on $\Hamiltonian$, $H^0$, and the considered equilibrium state $\EquState$ of $\Hamiltonian$, since $\GreensFunction^0$ is independent of the equilibrium state associated to $\SecondQuantization{\NIHamiltonian}$.

Using again \eqref{eq:G0} and the definition of the self-energy, the Green's function $G$ can be rewritten as
\begin{equation*}
\forall z \in \UpperHalfPlane, \quad\GreensFunction(z)=\left(z-(\NIHamiltonian+\SelfEnergy(z))\right)^{-1}
\end{equation*}
so that, for a given complex frequency $z$, $\NIHamiltonian + \SelfEnergy(z)$ can be considered as an \emph{effective one-particle Hamiltonian} (compare with~\eqref{eq:G0}): the self-energy is thus the extra term to be added to $\NIHamiltonian$ to obtain a representation of the interacting system of particles in terms of a system of non-interacting ones. The frequency dependence of $\Sigma$ comes from the fact that particles do interact in the original system.

\subsection{Hubbard model} \label{subsec:HubbardAndAndersonModel}

Originally introduced in quantum chemistry \cite{pariser_semiempirical_1953,pople_electron_1953}, the Hubbard model \cite{kanamori_electron_1963,hubbard_electron_1963,gutzwiller_effect_1963} is an idealized model that minimally describes \emph{interacting electrons} in a molecule or a crystalline material. From the mathematical physics' perspective, it is a prototypical example of short-range fermionic lattice system, and its mathematical study has been pioneered soon after its introduction in \cite{lieb_absence_1968,lieb_one-dimensional_2003}, triggering an extensive study of its ground-states properties \cite{lieb_hubbard_2004}. More recently, the discovery of cuprate-based high-temperature superconductors \cite{bednorz_possible_1986,yu_high-temperature_2019}, exhibiting a \emph{layered structure}, shed  new light on the \emph{square lattice} Hubbard model, for which much remains to be discovered \cite{lieb_hubbard_2004}. Since then, \emph{approximation methods} have been derived for this model such as \emph{generalized Hartree-Fock} \cite{bach_generalized_1994}, and \acrlong{dmft} \cite{georges_dynamical_1996}.

In this article, we restrict ourselves to \emph{finite-dimensional} Hubbard models, i.e. to Hubbard models defined on \emph{finite graphs}. The reason for this is threefold. First, we stick to the case of finite-dimensional state spaces, for which all the objects involved in the mathematical formulation of \acrshort{dmft} are well-defined. Second, \emph{graph theory} provides a suitable formalism to describe on the same footing different physical settings, ranging from molecular systems to truncated lattices, and with arbitrary hopping parameters (nearest neighbours, next nearest neighbors, ...). Third, this is precisely the language in which \acrshort{dmft} can be formulated concisely, as described in Section \ref{subsec:DMFTEquations}.

\begin{figure}[ht]
\centering
\begin{subfigure}{0.4\textwidth}
\centering
\includegraphics{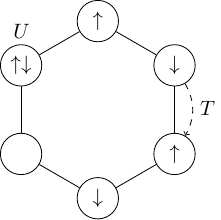}
\subcaption{The Pariser-Parr-Pople model of benzene C$_6$H$_6$: the cyclic graph $C_6$.}
\end{subfigure}
\begin{subfigure}{0.4\textwidth}
\centering
\includegraphics{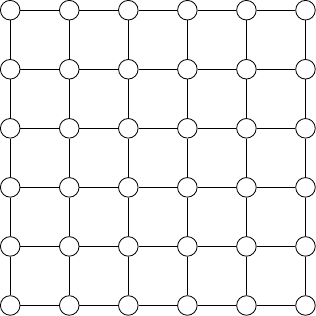}
\subcaption{A truncated lattice $\RelativeIntegers^2/6\RelativeIntegers$ with nearest neighbors : the square grid graph $P_6^{\square 2}$.}
\end{subfigure}
\caption{The Hubbard model for different graphs}
\label{fig:HubbardGraphs}
\end{figure}

Consider now a finite undirected graph $\HubbardGraph=(\HubbardVertices,\HubbardEdges)$ that describes the ``sites'' $\HubbardVertices$ hosting electrons, as depicted in Figure~\ref{fig:HubbardGraphs}. The state space of a site is the vector space spanned by the orthonormal vectors $\VacuumState$ (empty site), $\UpState$ (site occupied by one spin-up electron), $\DownState$ (site occupied by one spin-down electron), and $\UpDownState$ (doubly-occupied site). Note that $\UpDownState$ is a shorthand notation for the two-electron singlet state $2^{-1/2}(\UpState \otimes \DownState- \DownState \otimes \UpState)$. In chemistry language, to each site is associated a single orbital (and thus two spin-orbitals); in physics, this setting is referred to as the \emph{one-band} Hubbard model. Since the sites are distinguishable, the state space of the full system is the tensor product of the state space of each site. It is therefore of dimension $4^{\Cardinal{\HubbardVertices}}=2^{2L}$ where $L=\Cardinal{\HubbardVertices}$ is the number of sites.

Let us now specify the Hamiltonian. In the Hubbard model, electrons can jump from one site to any other neighboring site. This models the tunnel effect, whose intensity is described by the \emph{hopping matrix} $\HoppingMatrix[]$, as for tight-binding Hamiltonians. Electrons repel each other due to (screened) Coulomb interaction. The simplicity of the Hubbard model lies in the \emph{range} of this interaction, which is the shortest possible: it is only effective for two electrons occupying the same site, and if the site $i$ is doubly occupied, the energy of the system is increased by an \emph{on-site repulsion} $\OnSiteRepulsion[i]$.

More precisely, a finite-dimensional Hubbard model can be defined as follow.

\begin{definition}[Hubbard model of a finite graph] Given a finite graph $\HubbardGraph=(\HubbardVertices,\HubbardEdges)$, a hopping matrix $\HoppingMatrix[]:\HubbardEdges \to \RealNumbers$, and an on-site repulsion $\OnSiteRepulsion[]: \HubbardVertices \to \RealNumbers$, the Hubbard Fock Space $\HubbardFockSpace$ is defined as 
\begin{equation*}
\HubbardFockSpace=\bigotimes_{i \in \HubbardVertices} \HubbardOneSiteSpace, \quad  \HubbardOneSiteSpace=\Span(\VacuumState,\UpState,\DownState,\UpDownState)
\end{equation*}
and the Hubbard Hamiltonian $\HubbardHamiltonian \in \SelfAdjointOperator(\HubbardFockSpace)$ as 
\begin{equation*}
\HubbardHamiltonian=\sum_{\{i,j\} \in \HubbardEdges, \sigma \in \{\uparrow,\downarrow\}}\HoppingMatrix \left(\CreationOperator\AnnihilationOperator + \CreationOperator[j,\sigma]\AnnihilationOperator[i,\sigma]\right) + \sum_{i \in \HubbardVertices} \OnSiteRepulsion \NumberOperator[i,\uparrow] \NumberOperator[i,\downarrow]
\end{equation*}
where the usual annihilation and creation operators $\AnnihilationOperator[i,\sigma]$ and $\AnnihilationOperator[i,\sigma]^\dagger$ of an electron on site $j$ with spin $\sigma$ satisfy the \acrshort{car}

\begin{equation}\label{eq:CARHubbard}
\forall i,j \in \HubbardVertices, \quad \sigma \in \{\uparrow,\downarrow\}, \quad \AntiCommutator{\AnnihilationOperator[i,\sigma]}{\AnnihilationOperator[j,\sigma']}=\AntiCommutator{\CreationOperator[i,\sigma]}{\CreationOperator[j,\sigma']}=0,\quad \AntiCommutator{\AnnihilationOperator[i,\sigma]}{\CreationOperator[j,\sigma']}=\delta_{i,j}\delta_{\sigma,\sigma'},
\end{equation}
and the site number operators $\widehat n_{i,\sigma}$ are defined by $\widehat n_{i,\sigma}=\CreationOperator\AnnihilationOperator[i,\sigma]$.
\end{definition}

In this paper, we do not consider external magnetic field, so that we can assume without loss of generality that the hopping matrix $\HoppingMatrix[]$ is real-valued \cite{lieb_hubbard_2004}. For standard Coulomb interaction, the on-site repulsion $\OnSiteRepulsion[]$ is positive, but we do not need to make this assumption here.

\begin{remark} The Hubbard Hamiltonian $\HubbardHamiltonian$ is particle-number and spin conserving, i.e. it commutes with the number operator $\HubbardNumberOperator=\sum_{i \in \HubbardVertices} \left( \NumberOperator[i,\uparrow] + \NumberOperator[i,\downarrow]\right)$ and the spin operators. In particular, it commutes with the $z$-component $\HubbardSpinOperator^z=\frac 12\sum_{i \in \HubbardVertices} \left( \NumberOperator[i,\uparrow]-\NumberOperator[i,\downarrow]\right)$ of the spin operator. This property will be used later to make the IPT-DMFT model spin-independent.
\end{remark}

\subsection{\acrfull{aim}}\label{subsec:AndersonModel}

\paragraph{Impurity models}
Impurity models are models in which an ``impurity'' is coupled to a ``bath'' in such a way that particles in the bath do not interact, and the coupling Hamiltonian between the impurity and the bath only involves one-body terms. Otherwise stated, the one-particle state space and the Fock space of the total system can be decomposed as
\begin{equation} \label{eq:split_space_HAI}
    \OneParticleSpace[\rm IM] = {\mathcal H}_{\rm imp} \oplus {\mathcal H}_{\rm bath} \quad \mbox{and} \quad 
\FockSpace_{\rm IM} = \FockSpace_{\rm imp} \otimes \FockSpace_{\rm bath},
\end{equation}
and its Hamiltonian as
\begin{equation} \label{eq:split_Hamiltonian_HAI}
\Hamiltonian_{\rm IM} = \Hamiltonian_{\rm imp} \otimes 1_{\rm bath} + 1_{\rm imp} \otimes \Hamiltonian_{\rm bath} + 
\Hamiltonian_{\rm coupling},
\end{equation}
with 
\begin{equation} \label{eq:shape_Hamiltonian_HAI}
\Hamiltonian_{\rm imp} = d\Gamma(H^0_{\rm imp}) + \Hamiltonian_{\rm imp}^{\rm I}, \quad \Hamiltonian_{\rm bath}= d\Gamma(H^0_{\rm bath}), \quad \Hamiltonian_{\rm coupling}= d\Gamma(H^0_{\rm coupling}),
\end{equation}
and $H^0_{\rm coupling}$ can be decomposed according to \eqref{eq:split_space_HAI} as
\begin{equation*}
H^0_{\rm coupling} = \left( \begin{array}{cc} 0 & V \\ V^\dagger & 0 \end{array} \right).
\end{equation*}

Reshuffling the terms, we also have
\begin{equation} \label{eq:reshuffled_Hamiltonian_HAI}
\Hamiltonian_{\rm IM} = d\Gamma(H^0_{\rm IM}) + \Hamiltonian_{\rm imp}^{\rm I} \otimes 1_{\rm bath}, \quad \mbox{with} \quad H^0_{\rm IM} = \left( \begin{array}{cc} H^0_{\rm imp}  & V \\ V^\dagger & H^0_{\rm bath} \end{array} \right). 
\end{equation}

As will be seen below, a key step of the \acrshort{dmft} iteration loop is to compute the restriction of the Green's function $\GreensFunction_{\rm IM}$ of an impurity model to the impurity space $\OneParticleSpace[\mathrm{imp}]$. This is in general a difficult task. On the other hand, computing the restriction of the \emph{non-interacting} Green's function can easily be done using a Schur complement technique. This leads to the concept of hybridization function, which, as announced in the introduction, is of paramount importance in the mathematical formulation of \acrshort{dmft}. 

\begin{definition}[Hybridization function $\Hybridization$ of an impurity model] Given an impurity model of the form \eqref{eq:split_space_HAI}-\eqref{eq:reshuffled_Hamiltonian_HAI}, its \emph{hybridization function} $\Hybridization:\UpperHalfPlane \to \LinearOperator(\OneParticleSpace[\mathrm{imp}])$ is defined by 
\begin{equation}
\label{def:HybridizationAIM}
\forall z \in \UpperHalfPlane, \quad \Hybridization(z)=V \left(z-\NIHamiltonian_{\mathrm{bath}}\right)^{-1}V^\dagger.
\end{equation}
\end{definition}

Using this definition and Proposition~\ref{prop:NonInteractingGreensFunction}, the non-interacting Green's function of the impurity model is then given in the decomposition \eqref{eq:split_space_HAI} by
\begin{equation*}
    \GreensFunction^0(z) = (z-H^0_{\rm IM})^{-1} = \left( \begin{array}{cc} (z-H^0_{\rm imp}-\Delta(z))^{-1} & * \\ * & * \end{array} \right).
\end{equation*}
The hybridization function $\Hybridization$ thus plays a role analogous to the self-energy $\SelfEnergy$. The equation
$$
\left(\GreensFunction^0(z)\right)_\mathrm{imp}=\left(z-(\NIHamiltonian_{\mathrm{imp}}+\Hybridization(z))\right)^{-1}
$$
means that $\NIHamiltonian_{\mathrm{imp}}+\Hybridization(z)$ can be considered as an \emph{effective one-particle impurity Hamiltonian}: the hybridization function is the extra term to be added to $\NIHamiltonian_{\mathrm{imp}}$ so that, in the non-interacting case, a particle in the whole system can be regarded as a particle localized on the impurity. In the time domain, the hybridization function describes the phenomenon of a particle localized on the impurity, jumping into the bath, and jumping back to the impurity, hence the name \emph{hybridization}.

The most important result in this section is the following.
\begin{theorem}[$\SelfEnergy$ sparsity pattern of an impurity model]
\label{thm:SelfEnergyImpurityPattern}
Given an impurity model of the form \eqref{eq:split_space_HAI}-\eqref{eq:reshuffled_Hamiltonian_HAI}, the self-energy $\SelfEnergy_{\rm IM}: \UpperHalfPlane\to\OneParticleSpace[\rm IM]$ associated to any equilibrium state of $\Hamiltonian_{\rm IM}$ reads in the decomposition \eqref{eq:split_space_HAI}
\begin{equation}
\forall z \in \UpperHalfPlane, \quad \SelfEnergy_{\rm IM}(z)=\left(\begin{matrix}
\SelfEnergy_\mathrm{imp}(z) &0 \\
0 &0
\end{matrix} \right).\label{eq:sparsity_pattern_Sigma_AI}
\end{equation}
\end{theorem}

In practical \acrshort{dmft} computations and as mentioned already in \cite{feynman_theory_1963}, the self-energy $\SelfEnergy$ of an impurity problem depends solely on the hybridization function $\Hybridization$ and on the impurity Hamiltonian defined by $\NIHamiltonian_{\mathrm{imp}}$ and $\InteractingHamiltonian$, via a quantum action defined by path integrals \cite[eq. 9]{gull_continuous-time_2011}. In this article, we focus on the \acrshort{ipt} approximation (see section \ref{subsec:IPTSolver}), which makes this dependence almost explicit, and postpone the study of the exact impurity solver in a more general framework to future works. 

\paragraph{Anderson Impurity Model}
 The original Anderson impurity model (AIM) \cite{anderson_localized_1961} is a specific impurity model in which the impurity consists of a single-site Hubbard model. It was introduced by Anderson back in 1961 to explain the low-temperature behavior of the conductivity of metallic alloys with dilute magnetic impurities, later called the \emph{Kondo effect} \cite{kondo_resistance_1964,sarachik_resistivity_1968}. The AIM was then extended to multiple-site Hubbard impurities. Figure~\ref{fig:AIM} provides a graphical illustration of an AIM model with a 4-site Hubbard impurity and 5 bath orbitals (${\rm dim}({\mathcal H}_{\rm imp})=8$, ${\rm dim}({\mathcal H}_{\rm bath})=10$). Without loss of generality, we can indeed identify the bath space ${\mathcal H}_{\rm bath}$ with $\ComplexNumbers^{2B}$ and assume that $H^0_{\rm bath}={\rm diag}(\epsilon_1, \epsilon_1, \cdots, \epsilon_B, \epsilon_B)$. In this picture, the canonical basis of $\ComplexNumbers^{2B}$ corresponds to an orthonormal basis of bath spin-orbitals with energies $\epsilon_1, \epsilon_1, \dots, \epsilon_B, \epsilon_B$. The matrix $\BathCoupling[]$ models jumps from the bath to the impurity, while the matrix $\BathCoupling[]^\dagger$ models jumps from the impurity to the bath.
 The coupling terms $V_{i,j}$ are also assumed to be real-valued due to the absence of magnetic field.

\begin{figure}[ht]
\centering
\includegraphics{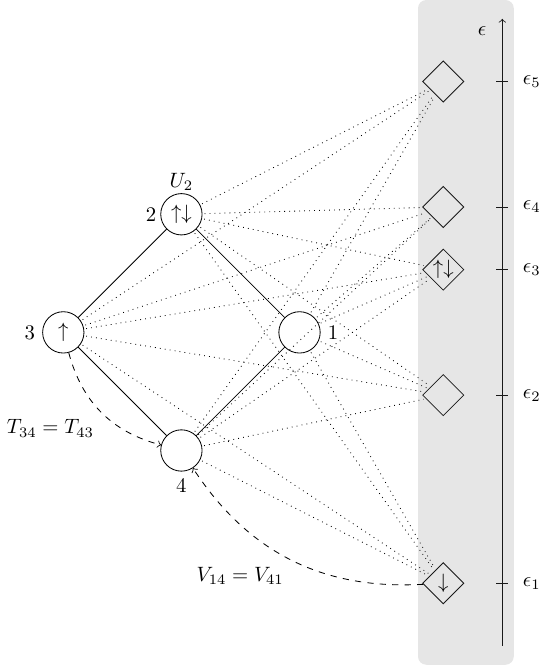}
\caption{The \acrlong{aim} of an impurity $(\HubbardGraph=C_4,\HoppingMatrix[],\OnSiteRepulsion[]$) and a \emph{bath} $(\BathDimension=5,\BathEnergy[],\BathCoupling[])$ in a 6 electrons state.}
\label{fig:AIM}
\end{figure}

In the seminal paper \cite{anderson_localized_1961}, the \emph{electronic bath} was thought as the set of \emph{conducting electrons}, for instance described by a tight-binding model on a (truncated) lattice.

\begin{remark} \label{rmk:AIMCommutationRelations} As for the Hubbard model, the z-component of the total spin operator $\AISpinOperator^z=\HubbardSpinOperator^z+\sum_{k=1}^\BathDimension \NumberOperator[k,\uparrow] - \NumberOperator[k,\downarrow]$, the total number operator $\AINumberOperator=\HubbardNumberOperator+\sum_{k=1}^\BathDimension \NumberOperator[k,\uparrow] + \NumberOperator[k,\downarrow]$ and the Anderson Impurity Hamiltonian $\AIHamiltonian$ pairwise commute. 
\end{remark}

\subsection{\acrfull{dmft}} \label{subsec:DMFTEquations}

Consider a Hubbard model defined by $(\HubbardGraph,\HoppingMatrix[],\OnSiteRepulsion[])$ in one of its equilibrium states $\EquState$. The purpose of \acrshort{dmft} is to provide an \emph{approximation} of the corresponding Green's function $\GreensFunction$, and more precisely on some blocks of this Green's function. 

To do so, DMFT uses a self-consistent mapping between the Hubbard model $(\HubbardGraph,\HoppingMatrix[],\OnSiteRepulsion[])$ and a collection of \acrlong{aim}s associated to a partition $\DMFTPartition$ of the vertices of the Hubbard graph $\HubbardGraph$. The process, illustrated in Figure~\ref{fig:DMFTPartition}, is the following.
\begin{enumerate}
\item Choose a partition $\DMFTPartition=\{\HubbardVertices_p, p \in \IntSubSet{1}{\PartitionSize} \}$ of the $L$ sites of the Hubbard model into $P$ impurities $\HubbardVertices_{1},\cdots \HubbardVertices_{P}$, ,  $\sqcup_{p=1}^{\PartitionSize}\HubbardVertices_{p}=\HubbardVertices$, and consider for all $p \in \IntSubSet{1}{\PartitionSize}$ the \emph{induced subgraphs} $\HubbardGraph[p]=(\HubbardVertices_{p},\HubbardEdges_{p})$ with $\HubbardEdges_{p}=\{\{i,j\} \in \HubbardEdges; \ i,j \in \HubbardVertices_{p}\}$. This partition leads to the canonical orthogonal decomposition 
\begin{equation} \label{eq:decomp_H1}
\OneParticleSpace[H]=\bigoplus_{p=1}^\PartitionSize \OneParticleSpace[p], \quad {\rm dim}(\OneParticleSpace[p])=2|\Lambda_p|=2L_p,
\end{equation}
from which follows the decomposition of the Fock space
\begin{equation}
    \label{eq:decompositionFockSpaceDMFT}
    \FockSpace_H = \bigotimes_{p=1}^\PartitionSize \FockSpace_p, \quad \dim \FockSpace_p = 4^{L_p}.
\end{equation}
The decomposition \eqref{eq:decomp_H1} of the one-particle state space of the original Hubbard model gives rise to block-operator representations of the exact Green's function and self-energy (for the state $\Omega$) of the original Hubbard model 
\begin{equation*}
G(z)=\left(\begin{matrix}
G_1(z) &*&\cdots&*\\
*& G_2(z)&\cdots&* \\
\vdots&\vdots&\ddots&\vdots \\
*&*&\cdots& G_P(z)
\end{matrix}\right), \qquad 
\Sigma(z)=\left(\begin{matrix}
\Sigma_1(z) &*&\cdots&*\\
*& \Sigma_2(z)&\cdots&* \\
\vdots&\vdots&\ddots&\vdots \\
*&*&\cdots& \Sigma_P(z)
\end{matrix}\right);
\end{equation*}
\item The decomposition \eqref{eq:decomp_H1} also leads to a block-operator representation of the one-body Hamiltonian corresponding to the non-interacting part of the Hubbard Hamiltonian
\begin{equation*}
H^0_{\rm H}=\left(\begin{matrix}
H^0_{1} &*&\cdots&*\\
*& H^0_{2} &\cdots&* \\
\vdots&\vdots&\ddots&\vdots \\
*&*&\cdots& H^0_\PartitionSize
\end{matrix}\right),
\end{equation*}
where the diagonal block $H^0_p$ is constructed from the \emph{induced hopping matrices} $\HoppingMatrix[p]=\HoppingMatrix[|\HubbardEdges_{p}]$, $1 \le p \le P$. Due to the locality of the interactions in the Hubbard model, the interaction term is ``block-diagonal'' in the decomposition  \eqref{eq:decompositionFockSpaceDMFT}. It is represented by the \emph{induced on-site repulsion} $\OnSiteRepulsion[p]=\OnSiteRepulsion[|\HubbardVertices_p]$, and it reads $\InteractingHamiltonian = \oplus_{p=1}^\PartitionSize \InteractingHamiltonian_p$, with $\InteractingHamiltonian_p = \sum_{i\in\HubbardVertices_p}U_i\NumberOperator[i,\uparrow]\NumberOperator[i,\downarrow]$;
\item The \acrshort{aim} associated to the $p$-th impurity is defined by (i) the impurity state space $\OneParticleSpace[{\rm imp},p]:=\OneParticleSpace[p]$, (ii) the impurity Hamiltonian defined by the induced Hubbard model $(\HubbardGraph[p],\HoppingMatrix[p],\OnSiteRepulsion[p])$, (iii) some bath state space $\OneParticleSpace[{\rm bath},p]$, bath one-body Hamiltonian $H^0_{{\rm bath},p}$ and coupling one-body Hamiltonian $H^0_{{\rm coupling},p}= \left( \begin{array}{cc} 0 & V_p^\dagger \\ V_p & 0 \end{array} \right)$ to be specified. From each of these $P$ AIMs, one can compute the Green's functions and self-energies 
\begin{equation*}
    G_{{\rm AI},p}(z) = \left( \begin{array}{cc} G_{{\rm imp},p}(z) & * \\ * & * \end{array} \right), \qquad 
     \Sigma_{{\rm AI},p}(z) = \left( \begin{array}{cc} \Sigma_{{\rm imp},p}(z) & 0 \\ 0 & 0 \end{array} \right),
\end{equation*}
for AIM equilibrium states $\Omega_p$ of the same nature as $\Omega$ (see Remark~\ref{rmk:AIMDMFTStates} below for a comment on this important point).
\end{enumerate}

\begin{figure}[ht]
\centering
\includegraphics{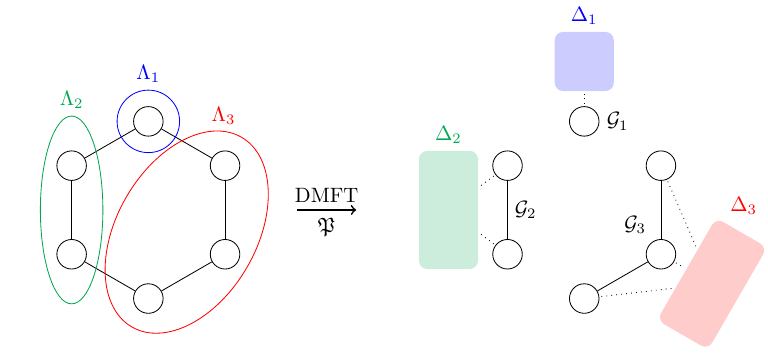}
\caption{A partition $\DMFTPartition=\HubbardVertices_1\sqcup \HubbardVertices_2\sqcup \HubbardVertices_3$ of the 6 vertices of the $C_6$ graph, and the 3 induced graphs $\HubbardGraph[1], \HubbardGraph[2], \HubbardGraph[3]$: the impurities of the corresponding \acrshort{aim}s are defined using the original Hubbard model, while the bath and the impurity-bath coupling are specified by the DMFT approximation.} \label{fig:DMFTPartition}
\end{figure}

DMFT aims at computing approximations of the diagonal-blocks $G_1(z), \cdots, G_P(z)$ of the exact Green's function $G(z)$. 
Ideally, the baths should be ajusted in such a way that $\GreensFunction_{\mathrm{imp},p} =\GreensFunction_p$, but of course this is not possible since the functions $\GreensFunction_p$ are unknown. To make DMFT a practical tool, the idea is to introduce approximate Green's functions and self-energies of the form
\begin{align}
G_{\rm DMFT}(z)&=\left(\begin{matrix}
G_{{\rm DMFT},1}(z) &*&\cdots&*\\
*& G_{{\rm DMFT},2}(z) &\cdots&* \\
\vdots&\vdots&\ddots&\vdots \\
*&*&\cdots& G_{{\rm DMFT},P}(z)
\end{matrix}\right), \label{eq:G_DMFT} \\ 
\Sigma_{\rm DMFT}(z)&=\left(\begin{matrix}
\Sigma_{{\rm DMFT},1}(z) & 0 &\cdots& 0\\
0& \Sigma_{{\rm DMFT},1}(z)&\cdots&0 \\
\vdots&\vdots&\ddots&\vdots \\
0&0&\cdots& \Sigma_{{\rm DMFT},P}(z)
\end{matrix}\right), \label{eq:Sigma_DMFT}
\end{align}
related by 
\begin{align*}
    \GreensFunction_{\mathrm{DMFT}}(z) &= ((\GreensFunction^0_\mathrm{DMFT})^{-1}(z)- \SelfEnergy_{\mathrm{DMFT}}(z))^{-1},\\
    \GreensFunction^0_{\mathrm{DMFT}}(z)&= \left(z-\NIHamiltonian_{\mathrm{H}} \right)^{-1},
\end{align*}
and to choose the baths in such a way that 
\begin{align}
    \forall 1 \le p \le P, \quad & G_{{\rm DMFT},p}(z)=G_{{\rm imp},p}(z), \label{eq:SC_G} \\ 
 & \Sigma_{{\rm DMFT},p}(z)=\Sigma_{{\rm imp},p}(z). \label{eq:SC_Sigma}
\end{align}
Of course, it is not clear whether a collection of baths that satisfies the above conditions exists, and this article partially answers this question. 

\medskip

Note that the DMFT Green's function $G_{\rm DMFT}(z)$ is \emph{not}, a priori, the Green's function of some interacting quantum many-body problem which could be defined as in Section~\ref{subsec:OneBodyTimeOrderedGreensFunctions}. Instead, it is defined from an approximate self-energy $\Sigma_{\rm DMFT}(z)$.
Equations \eqref{eq:Sigma_DMFT} and \eqref{eq:SC_Sigma} indicate that, in the DMFT approximation, each impurity interacts with its neighborhood as if the former were an impurity and the latter a bath, in accordance with Theorem \ref{thm:SelfEnergyImpurityPattern}, hence the name of \emph{mean-field}. This mean-field is \emph{dynamical} because it is frequency dependent, in contrast with static mean-field theory such as in Hartree-Fock or Density-Matrix Embedding Theory (DMET).

\begin{remark}\label{rmk:AIMDMFTStates} In the above sketch of the DMFT framework, we did not specify the states~$(\EquState_p)_{p \in \IntSubSet{1}{\PartitionSize}}$ of the associated \acrshort{aim}s from which the self-energies $\left(\SelfEnergy_{\mathrm{imp},p}\right)_{p\in \IntSubSet{1}{\PartitionSize}}$ are computed. In \cite{georges_dynamical_1996}, it is implicitly stated that, for a Hubbard model in the \emph{Gibbs} state $\EquState_{\mathrm{H},\StatisticalTemperature,\ChemicalPotential}$ as defined below in Section \ref{subsec:MatsubaraGF}, the \acrshort{aim}'s equilibrium states $\EquState_p$ are defined to be \emph{Gibbs} states as well $\EquState_p=\EquState_{\mathrm{AIM},\StatisticalTemperature,\ChemicalPotential'}$, with the same temperature as that of the \emph{Gibbs} state of the whole system, but with \emph{a priori} different chemical potential $\ChemicalPotential'$. The latter is to be chosen to satisfy appropriate \emph{filling} conditions. We will address this question in a future work and simply assume here that the $\ChemicalPotential'=\ChemicalPotential$. Moreover, when working with the \acrshort{ipt} solver, the chemical potential is fixed by the on-site interaction, as detailed in Section \ref{subsec:IPTSolver}.
\end{remark}

Our analysis is based on a reformulation of conditions \eqref{eq:SC_G}-\eqref{eq:SC_Sigma} using the hybridization functions $(\Delta_p(z))_{1 \le p \le P}$ of the $P$ impurity problems as the main variables. As mentionned previously, knowing $\Delta_p$ as well as $T_p$, $U_p$ and $\Omega_p$,
suffices in principle to compute $\Sigma_{{\rm imp},p}$. In practice, this is done by using an approximate impurity solver. A particular example of such a solver will be presented in Section~\ref{subsec:IPTSolver}. Conversely, knowing $(\Sigma_{{\rm imp},p})_{1 \le p \le P}$ and $T$ suffices to characterize the unique set $(\Delta_p(z))_{1 \le p \le P}$ for which \eqref{eq:SC_G}-\eqref{eq:SC_Sigma} hold true. Indeed, by denoting
$$
\OneParticleSpace[\overline{p}] := \bigoplus_{p \neq q=1}^P \OneParticleSpace[q],
$$
$$
\Sigma_{{\rm DMFT},\overline{p}}(z) =\left(\begin{matrix}
\Sigma_{{\rm DMFT},1}(z) & \cdots &0 &\cdots& 0 \\
\vdots & \ddots & \vdots & & & \\
0 & \cdots & \Sigma_{{\rm DMFT},p-1}(z) & 0 & & & \\ 
0 & \cdots & 0 & \Sigma_{{\rm DMFT},p+1}(z) &\cdots & 0 \\
\vdots&\ddots&\vdots&\vdots \\
0&\cdots & 0 & 0 & \cdots& \Sigma_{{\rm DMFT},P}(z)
\end{matrix}\right), 
$$
and using the Schur complement formula, we have on the one hand
\begin{align}
    G_{{\rm DMFT},p}(z)&=\left( \left(z-(H^0_{\rm H}+\Sigma_{{\rm DMFT}}(z))\right)^{-1}\right)_p \nonumber \\
    &=  \left( z- \left(H^0_p+\Sigma_{{\rm DMFT},p}(z)\right) - W_p \left(z- \left( H^0_{\overline{p}} +  \Sigma_{{\rm DMFT},\overline{p}}(z)\right) \right)^{-1} W_p^\dagger \right)^{-1}, 
    \label{eq:G_DMFT_p}
\end{align}
where 
\begin{equation}
\label{eq:BlockStructure}
\begin{pmatrix}
    \NIHamiltonian_p & \NIHamiltonianSchurBlock \\ \NIHamiltonianSchurBlock^\dagger & \NIHamiltonian_{\Bar{p}}
\end{pmatrix}
\text{ and }
\left( \begin{array}{cc} z-(H^0_p+\Sigma_{\rm DFMT,p}(z))  & -W_p \\
-W_p^\dagger & z-(H^0_{\overline{p}}+\Sigma_{\rm DFMT,\overline{p}}(z)) \end{array} \right)
\end{equation}
 are the block-representations of $\NIHamiltonian_{\rm H}$ and $(z-(H^0_{\rm H}+\Sigma_{\rm DMFT}(z)))$, respectively, in the decomposition $\OneParticleSpace=\OneParticleSpace[p] \oplus \OneParticleSpace[\overline{p}]$. Note that $H^0_p \in \LinearOperator(\OneParticleSpace[p])$, $W_p \in 
\LinearOperator\left(\OneParticleSpace[\overline{p}];\OneParticleSpace[p] \right)$, and $ H^0_{\overline{p}}, \Sigma_{\rm DFMT,\overline{p}}(z) \in \LinearOperator\left(\OneParticleSpace[\overline{p}]\right)$. 
On the other hand, using again the Schur complement formula, we have

\begin{align*}
    G_{{\rm AI},p}(z) &= \left( G^0_{{\rm AI},p}(z)^{-1} - \Sigma_{{\rm AI},p}(z) \right)^{-1} =  \left( z-H^0_{{\rm AI},p} - \Sigma_{{\rm AI},p}(z) \right)^{-1} \\
    &=
    \left(\begin{array}{cc}
z-H^0_{\rm p}-\Sigma_{{\rm imp},p}(z) & -V_{p} \\ -V_p^\dagger & z-H^0_{{\rm bath},p} 
\end{array}\right)^{-1}
= 
    \left(\begin{array}{cc}
(z-H^0_{\rm p}-\Sigma_{{\rm imp},p}(z)-\Delta_p(z))^{-1} & * \\ * & * 
\end{array}\right),
\end{align*}
so that
\begin{equation} \label{eq:G_imp_p}
    G_{{\rm imp},p}(z) = (z-H^0_{\rm p}-\Sigma_{{\rm imp},p}(z)-\Delta_p(z))^{-1}.
\end{equation}
Conditions \eqref{eq:SC_G}-\eqref{eq:SC_Sigma}, together with \eqref{eq:G_DMFT_p}-\eqref{eq:G_imp_p} yield the self-consistent condition 

\begin{equation}  \label{eq:SC_condition}
\forall 1 \le p \le P, \quad \Hybridization_p(z) = W_p \left(z-H^0_{\overline{p}} - \bigoplus_{q=1,q \neq p}^\PartitionSize \Sigma_{\mathrm{imp},q}(z) \right)^{-1} W_p^\dagger.
\end{equation}
Note that the matrices $W_p$, $\NIHamiltonian_{\bar{p}}$ depend only on the  hopping matrix $\HoppingMatrix[]$ through
\begin{align*}
[W_p]_{i\sigma,j\sigma'}&=\left\lbrace \begin{matrix}
\HoppingMatrix[i,j] \\ 0 \end{matrix} \begin{matrix} \text{ if }  i \in \HubbardVertices_p, j \in \HubbardVertices\setminus\HubbardVertices_p, \{i,j\} \in \HubbardEdges,  \; \sigma=\sigma', \\
 \text{ otherwise,}
\end{matrix} \right. \\ [\NIHamiltonian_{\bar{p}}]_{i\sigma,j\sigma'}&= \left\lbrace \begin{matrix} \HoppingMatrix[i,j] \\ 0 \end{matrix} \begin{matrix} \text{ if }  i,j \in \HubbardVertices\setminus\HubbardVertices_p, \{i,j\} \in \HubbardEdges, \; \sigma=\sigma', \\
 \text{ otherwise,} \end{matrix} \right.
\end{align*}
where $i,j \in \Lambda$ denote site indices and $\sigma,\sigma' \in \{\uparrow,\downarrow\}$ spin indices.

\begin{remark}
\label{rmk:TranslationInvariantDMFT}
This formulation of \acrshort{dmft} agrees with the original one \cite{georges_fermions_2018, georges_dynamical_1996} in the translation-invariant setting where
\begin{equation}
\forall p \in \IntSubSet{1}{\PartitionSize}, \quad \OneParticleSpace[p]=\OneParticleSpace[1], \quad H^0_p=H^0_1, \quad W_p=W_1, \quad H^0_{\overline{p}}=H^0_{{\overline{1}}}, \quad U_p=U_1. \label{eq:TranslationInvariantSetting}
\end{equation} 
In this setting, we can consider translation-invariant solutions to the DMFT equations, for which 
\begin{align*}
\forall p \in \IntSubSet{1}{\PartitionSize}, \quad \forall z \in \UpperHalfPlane, \quad \Sigma_{{\rm imp},p}(z)&=\Sigma_{{\rm imp}}(z), \quad \Hybridization_p(z)=\Hybridization(z),
\end{align*}
where $\Sigma_{{\rm imp}}(z)$ and $\Hybridization(z)$ are related by the translation-invariant self-consistent condition
\begin{equation}
\Hybridization(z)=W_1 \left(z-\NIHamiltonian_{\bar{1}} -\bigoplus_{q=2}^{\PartitionSize} \SelfEnergy_{\mathrm{imp}}(z)\right)^{-1}W_1^\dagger.
\label{eq:translation_invariant_SC}
\end{equation}
This setting is sketched in Figure~\ref{fig:TranslationInvariantDMFT}. 
When $\DMFTPartition$ is a partition into singletons (single-site \acrshort{dmft}), condition \eqref{eq:TranslationInvariantSetting} is equivalent to the fact that $\OnSiteRepulsion[]$ is constant and  $(\HubbardGraph,\HoppingMatrix[])$ is \emph{vertex-transitive}, namely that for all $\lambda_1,\lambda_2 \in \HubbardVertices$, there exists a graph isomorphism $\tau:\HubbardVertices\to\HubbardVertices$ such that 
\begin{equation*}
\tau(\lambda_1)=\lambda_2, \quad \forall \lambda,\lambda' \in \HubbardVertices, \quad\HoppingMatrix[\tau(\lambda),\tau(\lambda')]=\HoppingMatrix[\lambda,\lambda']
\end{equation*} In particular, the graph Cartesian product $C_{N}^{\square d}$ of $d$ copies of the $N$-cycle, which is the nearest neighbor graph of a truncation of the $d$-dimensional square lattice, with constant hopping and on-site repulsion, and artificial periodic boundary conditions (supercell method), is vertex-transitive due to the translation invariance of the corresponding lattice. This setting was the one considered in the first \acrshort{dmft} computations \cite{georges_dynamical_1996}.
 Due to translation invariance, a single impurity problem has to be solved at each iteration. Recall that the impurity solver is the computational bottleneck in DMFT.
\end{remark}

\begin{figure}[ht]
\centering
\includegraphics{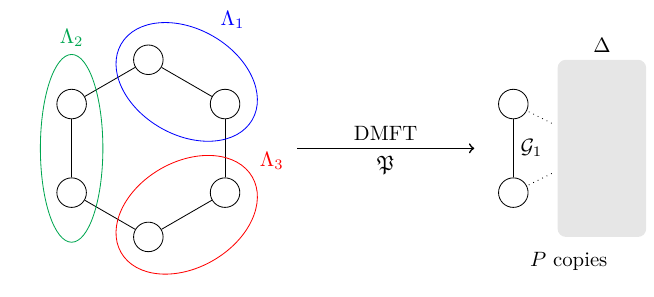}
\caption{Translation-invariant \acrshort{dmft} formalism.} \label{fig:TranslationInvariantDMFT}
\end{figure}

The DMFT self-consistent equations \eqref{eq:SC_condition}, combined with an exact impurity solver, give the exact Green's function of the original Hubbard model in the following trivial limits.

\begin{proposition}[Exactness of DMFT in some trivial limits]\label{prop:DMFTTrivialLimitsExactness} Consider a Hubbard model $(\HubbardGraph,\HoppingMatrix[],\OnSiteRepulsion[])$ in a Gibbs state $\GibbsState$.
The self-consistent DMFT equations \eqref{eq:SC_condition}, combined with an exact impurity solver, admit a unique solution in each of the following two settings:
\begin{enumerate}
    \item non-interacting particles. If the on-site repulsion term is equal to zero  ($\OnSiteRepulsion[i]=0$ for all $i \in \Lambda$), then this solution is given by
\begin{align}
\forall p \in \IntSubSet{1}{\PartitionSize}, \quad \forall z \in \UpperHalfPlane, \quad \SelfEnergy_{\mathrm{imp},p}(z)&=0, \nonumber \\
 \Hybridization_p(z)&= W_p \left(z-\NIHamiltonian_{\overline{p}}\right)^{-1} W_p^\dagger;
 \label{eq:hybridization_NI}
\end{align}
\item disconnected graphs and atomic limit. If the partition $\DMFTPartition$ of $\HubbardGraph$ is such that $\HubbardGraph=\bigoplus_{p=1}^\PartitionSize \HubbardGraph[p]$ (meaning $\HubbardEdges_H=\sqcup_{p=1}^\PartitionSize \HubbardEdges_p$), that is if the partition decomposes the original Hubbard model $(\HubbardGraph,\HoppingMatrix[],\OnSiteRepulsion[])$ into $P$ independent Hubbard models $(\HubbardGraph[p],\HoppingMatrix[p],\OnSiteRepulsion[p])$, or if the hopping matrix $\HoppingMatrix[]$ vanishes, then this solution is given by

\begin{align}
\forall p \in \IntSubSet{1}{\PartitionSize}, \quad \forall z \in \UpperHalfPlane, \quad \SelfEnergy_{\mathrm{imp},p}(z)&=\SelfEnergy_{p}(z), \nonumber \\
\Hybridization_p(z)&=0,
\label{eq:hybridization_disconnected}
\end{align}
where $\SelfEnergy_{p}$ is the exact self-energy associated to the $p$-th Hubbard model in the associated Gibbs State $\GibbsState^p$.
\end{enumerate} 
In both settings, DMFT is exact in the sense that
\begin{equation*}
\forall z \in \UpperHalfPlane, \quad \GreensFunction_{\mathrm{DMFT}}(z)=\GreensFunction(z).
\end{equation*}
\end{proposition}

 The purpose of \acrshort{dmft} \cite{georges_fermions_2018} is to build an appoximation that complies with these two limiting cases. Another limiting case in which \acrshort{dmft} is claimed to be exact is in the ``infinite dimension'' limit \cite{metzner_correlated_1989,georges_dynamical_1996}. We leave the mathematical investigation of this limit to future works.

\medskip

We deduce from \eqref{eq:hybridization_NI}-\eqref{eq:hybridization_disconnected} that in the trivial limits considered in Proposition~\ref{prop:DMFTTrivialLimitsExactness}, the hybridization functions $\Delta_p(z)$ are either identically zero, or have a finite number of poles, so that the corresponding baths can be chosen finite dimensional.

Anticipating on the following, we prove in Proposition \ref{prop:NonExistenceFiniteBathIPTDMFT} that, when coupled to the \acrfull{ipt} impurity solver, the translation-invariant self-consistent equation \eqref{eq:translation_invariant_SC} has no solution with a finite number of poles. This motivates the functional setting described in Section \ref{sec:MainResults}.

\subsection{A specific impurity solver : the \acrfull{ipt} solver.}\label{subsec:IPTSolver}

To define properly the \acrshort{ipt} solver, we need first to introduce \emph{Matsubara}'s formalism, and more precisely \emph{Matsubara's Green's function} $\MatsubarasGreensFunction$.

\subsubsection{Matsubara's Green's functions}\label{subsec:MatsubaraGF}

Matsubara's Green's functions are defined only for \emph{Gibbs states} $\GibbsState$ at a given inverse temperature $\StatisticalTemperature$ and chemical potential $\ChemicalPotential$. These functions have been more extensively studied mathematically than time-ordered Green's functions \cite{bratteli_operator_1987,bratteli_operator_1997}. In this section, we recall their definition and prove an analytic continuation result that will be useful for our analysis. The setup remains the same as the one introduced in Section \ref{subsec:OneBodyTimeOrderedGreensFunctions}.
\begin{definition}[Gibbs's state and Matsubara's time-ordered Green's function] Given a particle-number conserving Hamiltonian $\Hamiltonian \in \SelfAdjointOperator(\FockSpace)$, that is $\Commutator{\Hamiltonian}{\TotalNumberOperator}=0$, the \emph{Gibbs state} $\GibbsState$ at inverse temperature $\StatisticalTemperature \in \RealNumbers_+^*$ and chemical potential $\ChemicalPotential \in \RealNumbers$ is defined through its density operator by
\begin{equation*}
\DensityOperator=\frac{1}{\PartitionFunction} e^{-\StatisticalTemperature\left(\Hamiltonian-\ChemicalPotential \TotalNumberOperator\right)}, \quad \PartitionFunction= \Trace\left(e^{-\StatisticalTemperature\left(\Hamiltonian-\ChemicalPotential \TotalNumberOperator\right)}\right).
\end{equation*}
The \emph{Matsubara's Green's function} $\MatsubarasGreensFunction$ is defined as the $\LinearOperator(\OneParticleSpace)$-valued map $\MatsubarasGreensFunction:[-\StatisticalTemperature,\StatisticalTemperature) \to \LinearOperator(\OneParticleSpace)$ represented by the sesquilinear form verifying for all $\OneParticleState,\OneParticleState' \in \OneParticleSpace$,
\begin{equation}
\label{eq:DefinitionMatsubaraGreensFunction}
\forall \tau \in [-\StatisticalTemperature,\StatisticalTemperature), \quad 
\HermitianProduct{\OneParticleState}{-\MatsubarasGreensFunction(\tau)\OneParticleState'}= \CharacteristicFunction{\RealNumbers_+}(\tau) \,
\GibbsState\left(\MatsubarasPicture{\AnnihilationOperator[\OneParticleState]}(\tau) \CreationOperator[\OneParticleState']\right) - \CharacteristicFunction{\RealNumbers_-^*}(\tau) \, \GibbsState\left(\CreationOperator[\OneParticleState']\MatsubarasPicture{\AnnihilationOperator[\OneParticleState]}(\tau)\right),
\end{equation}
where for all $\Operator \in \LinearOperator(\FockSpace)$,
$$
\MatsubarasPicture{\Operator}:[-\StatisticalTemperature,\StatisticalTemperature] \ni \tau \mapsto e^{\tau(\Hamiltonian-\ChemicalPotential\TotalNumberOperator)}\Operator e^{-\tau(\Hamiltonian-\ChemicalPotential\TotalNumberOperator)}
$$
is the \emph{Matsubara picture} of $\Operator$.
\end{definition} 

Recall that in our setting, any operator is trace-class since $\FockSpace$ is finite dimensional, hence $\PartitionFunction$ is always finite.

As it is the case for the time-ordered Green's function $\TimeGreensFunction$, one can recast equation \eqref{eq:DefinitionMatsubaraGreensFunction} using the time-ordered product :
\begin{equation*}
\forall \OneParticleState,\OneParticleState' \in \OneParticleSpace, \ \forall \tau \in [-\StatisticalTemperature,\StatisticalTemperature), \quad \HermitianProduct{\OneParticleState}{-\MatsubarasGreensFunction(\tau) \OneParticleState'}=\GibbsState\left(\mathcal{T}\left(\MatsubarasPicture{\AnnihilationOperator[\OneParticleState]},\MatsubarasPicture{\CreationOperator[\OneParticleState']}\right)(\tau,0)\right).
\end{equation*}
Note that the negative sign in the definition of $\MatsubarasGreensFunction$ must be consistent with the $i$ prefactor in the definition of $\TimeGreensFunction$ for Theorem \ref{thm:MatsubaraFourierAnalyticExtension} below to hold. Considering the grand canonical Hamiltonian $\Hamiltonian'=\Hamiltonian-\ChemicalPotential\TotalNumberOperator$ as the Hamiltonian from which the Green's function $\TimeGreensFunction$ is defined, the Matsubara formalism involves the following formal connection (known as \emph{Wick rotation} \cite{wick_properties_1954}):
\begin{equation*}
\tau \leftrightarrow it
\end{equation*}
or, in other words, working with $t \leftrightarrow -i \tau$ where $\tau$ is real. Hence, the term \emph{imaginary-time} Green's function \cite{martin_interacting_2016}.

Contrary to the Heisenberg picture $\Operator \mapsto \HeisenbergPicture{\Operator}$, the Matsubara picture does not consist in a family of $C^*$-morphisms: one has for all $\tau \in [-\StatisticalTemperature,\StatisticalTemperature]$,
\begin{equation*}
\left(\MatsubarasPicture{\Operator}(\tau)\right)^\dagger = \MatsubarasPicture{\Operator^\dagger}(-\tau).
\end{equation*}
A consequence of this property is that the Matsubara's Green's function is Hermitian:
\begin{equation*}
\left(\MatsubarasGreensFunction(\tau)\right)^\dagger=\MatsubarasGreensFunction(\tau).
\end{equation*}

Note moreover that Gibbs states are \emph{\acrfull{kms}} states \cite{bratteli_operator_1997}, meaning that they satisfy the following property:
\begin{equation*}
\forall\ \Operator,\Operator' \in \LinearOperator(\FockSpace),\ \forall \tau \in [-\StatisticalTemperature,0], \quad \GibbsState(\MatsubarasPicture{\Operator}(\tau + \StatisticalTemperature)\Operator')=\GibbsState(\Operator'\MatsubarasPicture{\Operator}(\tau)).
\end{equation*}
This implies that $\MatsubarasGreensFunction$ is $\StatisticalTemperature$-\emph{anti-periodic}, i.e.
\begin{equation*}
\forall \tau \in [-\StatisticalTemperature,0), \quad \MatsubarasGreensFunction(\tau + \beta)=-\MatsubarasGreensFunction(\tau).
\end{equation*}

As for the time-ordered Green's function, the Matsubara's Green's function has a Källén-Lehmann's representation: given an orthonormal basis $\FockSpaceBasis$ of $\FockSpace$ which jointly diagonalizes $\Hamiltonian$ and $\TotalNumberOperator$ ($\forall \FockSpaceState \in \FockSpaceBasis,  \Hamiltonian \FockSpaceState = E_\FockSpaceState\FockSpaceState, \; \TotalNumberOperator\FockSpaceState= N_\FockSpaceState\FockSpaceState$), we have for all $\tau \in [0,\StatisticalTemperature)$
\begin{align*}
\HermitianProduct{\OneParticleState}{-\MatsubarasGreensFunction(\tau) \OneParticleState'} &= \sum_{\FockSpaceState,\FockSpaceState' \in \FockSpaceBasis}  e^{\tau((E_\FockSpaceState-\ChemicalPotential N_\FockSpaceState)-(E_{\FockSpaceState'} -  \ChemicalPotential N_{\FockSpaceState'}))} \HermitianProduct{\FockSpaceState}{\AnnihilationOperator[\OneParticleState]\FockSpaceState'} \HermitianProduct{\FockSpaceState'}{\CreationOperator[\OneParticleState']\FockSpaceState} e^{-\StatisticalTemperature(E_\FockSpaceState - \ChemicalPotential N_\FockSpaceState)} \\ & =\sum_{\FockSpaceState,\FockSpaceState' \in \FockSpaceBasis}  e^{\tau(E_\FockSpaceState-E_{\FockSpaceState'} + \ChemicalPotential)} \HermitianProduct{\FockSpaceState}{\AnnihilationOperator[\OneParticleState]\FockSpaceState'} \HermitianProduct{\FockSpaceState'}{\CreationOperator[\OneParticleState']\FockSpaceState} e^{-\StatisticalTemperature(E_\FockSpaceState - \ChemicalPotential N_\FockSpaceState)},
\end{align*}
as $N_{\FockSpaceState'}=N_\FockSpaceState+1$ whenever $\HermitianProduct{\FockSpaceState'}{\CreationOperator[\OneParticleState']\FockSpaceState} \neq 0$.

Similarly to the time-ordered Green's function $\TimeGreensFunction$, it is convenient to work with a Fourier representation of the Matsubara Green's function $\MatsubarasGreensFunction$.
Since the latter is defined only on $[-\StatisticalTemperature,\StatisticalTemperature)$, it is quite natural to extend $\MatsubarasGreensFunction$ to the real-line by periodicity and introduce the associated  Fourier series with coefficients defined by
\begin{equation*}
\forall n \in \RelativeIntegers, \quad \frac{1}{2}\int_{-\StatisticalTemperature}^{\StatisticalTemperature}\MatsubarasGreensFunction(\tau) e^{in\frac{\pi}{\StatisticalTemperature}\tau}d\tau.
\end{equation*}
Due to the $\StatisticalTemperature$-anti-periodicity of $\MatsubarasGreensFunction$, the even Fourier coefficients vanish and it holds
\begin{equation*}
\frac{1}{2}\int_{-\StatisticalTemperature}^{\StatisticalTemperature}\MatsubarasGreensFunction(\tau) e^{in\frac{\pi}{\StatisticalTemperature}\tau}d\tau= \left\lbrace \begin{matrix}
\int_{0}^\StatisticalTemperature \MatsubarasGreensFunction(\tau) e^{in\frac{\pi}{\StatisticalTemperature}\tau}d\tau & \text{if $n$ is odd,} \\
0 & \text{ otherwise.}
\end{matrix} \right.
\end{equation*}

\begin{definition}[Matsubara's Fourier series and frequencies] \label{def:MatsubarasFourierSeries} The \emph{Matsubara's Fourier coefficients} $\left(\MatsubaraFourierSeries_n\right)_{n \in \RelativeIntegers}$ are defined by
\begin{equation*}
\forall n \in \RelativeIntegers, \quad 
\MatsubaraFourierSeries_n=\int_0^\StatisticalTemperature \MatsubarasGreensFunction(\tau) e^{i\MatsubaraFrequency_n\tau} d\tau
\end{equation*}
where $\MatsubaraFrequency_n=(2n+1)\frac{\pi}{\StatisticalTemperature}$ is the $n$-th \emph{Matsubara's frequency}. 
\end{definition}

We thus have
\begin{equation*}
\forall \tau \in [-\StatisticalTemperature,\StatisticalTemperature), \quad
\MatsubarasGreensFunction(\tau)=\frac{1}{\StatisticalTemperature} \sum_{n \in \RelativeIntegers}e^{-i\MatsubaraFrequency_n\tau} \MatsubaraFourierSeries_n.
\end{equation*}

The very reason these coefficients are useful in Green's functions methods lies in the following theorem.
\begin{theorem}[Matsubara's Fourier coefficients analytic continuation]
\label{thm:MatsubaraFourierAnalyticExtension}
Let $\Hamiltonian \in \SelfAdjointOperator(\FockSpace)$ be a particle-number-conserving Hamiltonian, $\GibbsState$ the associated Gibbs state, and $\GreensFunction: \UpperHalfPlane\to\LinearOperator(\OneParticleSpace)$ the \acrlong{gft} of the associated time-ordered Green's function $\TimeGreensFunction$ defined from the grand canonical Hamiltonian $\Hamiltonian'=\Hamiltonian-\ChemicalPotential\TotalNumberOperator$. Then $\GreensFunction$ is the only analytic matrix-valued function such that
\begin{equation}
\label{eq:NegativeImaginaryPartG}
    \forall z\in\UpperHalfPlane, \ \Im(G(z)):=\frac{G(z)-G(z)^\dagger}{2i}\leq 0, \text{ and }
\end{equation}
\begin{equation}
\label{eq:MatsubaraInterpolation}
\forall n \in \Integers, \quad \GreensFunction(i\MatsubaraFrequency_n)= \MatsubaraFourierSeries_n.
\end{equation}
\end{theorem}

Note that, since $\MatsubaraFrequency_{-n}=-\MatsubaraFrequency_{n-1}$ and $\MatsubarasGreensFunction$ is Hermitian, it holds 
\begin{equation*}
\forall n \in \Integers^*, \quad \left( \MatsubaraFourierSeries_{n-1}\right)^\dagger = \MatsubaraFourierSeries_{-n},
\end{equation*}
so that \eqref{eq:MatsubaraInterpolation} actually holds for all $n \in \RelativeIntegers$, with the extension of $\GreensFunction$ to $\LowerHalfPlane$ defined as in \eqref{eq:extension_C-}.

\begin{remark}
The requirement that $\GreensFunction$ is analytic and verifies \eqref{eq:NegativeImaginaryPartG} is crucial for uniqueness: for instance, for each $m \in 2\RelativeIntegers+1$, the function
\begin{equation*}
\UpperHalfPlane \ni z \mapsto \frac{1-e^{m\StatisticalTemperature z}}{2} \GreensFunction(z) \in \LinearOperator(\OneParticleSpace)
\end{equation*}
also satisfies \eqref{eq:MatsubaraInterpolation} and is analytic, but its imaginary part is not negative semi-definite.
\end{remark}

Theorem~\ref{thm:MatsubaraFourierAnalyticExtension} is extensively used for practical computations: it is sufficient to run the computations for the Matsubara frequencies and then perform an analytic continuation \cite{fei_analytical_2021}.

Note that Theorem \ref{thm:MatsubaraFourierAnalyticExtension} works in particular for non-interacting Hamiltonians, which leads to the following definition.

\begin{definition}[Matsubara's self-energy]
The \emph{Matsubara's self-energy Fourier coefficients} $\left(\MatsubaraSelfEnergyFourierSeries_n\right)_{n\in \RelativeIntegers}$ are defined by 
\begin{equation*}
\forall n \in \RelativeIntegers, \quad \MatsubaraSelfEnergyFourierSeries_n=i\MatsubaraFrequency_n +\ChemicalPotential -\NIHamiltonian - \left(\MatsubaraFourierSeries_n\right)^{-1}.
\end{equation*}
\end{definition}

In fact, the self-energy $\SelfEnergy: \UpperHalfPlane\to\LinearOperator(\OneParticleSpace)$ defined in \eqref{eq:self-energy} is the only analytic function with negative imaginary part such that  
\begin{equation*}
\forall n \in \Integers, \quad \SelfEnergy(i\MatsubaraFrequency_n)=\MatsubaraSelfEnergyFourierSeries_n.
\end{equation*}
This follows from Proposition \ref{prop:PickFunctionsDMFT} and Theorem \ref{thm:UniquenessNevanlinnaPickInterpolation}.
One can define the \emph{Matsubara's self-energy} $\MatsubaraSelfEnergy(\tau)$ by Fourier summation formula as in Definition \ref{def:MatsubarasFourierSeries}, but this function will not play any role in what follows.

\medskip

With all these definitions in place, we can now introduce the final ingredient of the model under investigation, namely the paramagnetic single-site translation-invariant \acrfull{ipt}.

\subsubsection{IPT approximation}

In this article, we will not discuss the derivation nor the validity of the Iterative Perturbation Theory (IPT) approximation and refer the interested reader to \cite{arsenault_benchmark_2012}. As in the review paper \cite{georges_dynamical_1996}, we only consider \emph{single-site} and \emph{translation-invariant} \acrshort{dmft}. This seems to be a setting in which the usual IPT approximation is generally considered as valid in the physics literature, while constructing an IPT-like approximation for multi-site \acrshort{dmft} is still an active field of research \cite{arsenault_benchmark_2012}. 
In the former case, $\DMFTPartition$ is a partition of the $L$ sites into $P=L$ singletons and the on-site repulsion $\OnSiteRepulsion[]$ is constant as exposed in remark \ref{rmk:TranslationInvariantDMFT}. Also, we will focus on the paramagnetic case~\cite{rozenberg_mott-hubbard_1994}. In other words, we will assume that there is no spin-symmetry breaking, so that the spin components can be factored out as detailed in Appendix \ref{appendix:SpinIndependenceIPTDMFT}. In the single-site translation invariant paramagnetic IPT-DMFT approximation considered in the sequel, we thus have $P=L$ and for $z\in\UpperHalfPlane$,

\begin{equation*}
        \NIHamiltonian_{\rm H},\ \GreensFunction^{\acrshort{dmft}}(z),\ \SelfEnergy^{\acrshort{dmft}}(z) \in \ComplexNumbers^{P \times P} .
\end{equation*}
Recall that in translation invariant \acrshort{dmft}, we restrict ourselves to translation invariant solutions to the \acrshort{dmft} equations, so that we consider only one hybridization  function $\Hybridization$ and self-energy $\SelfEnergy$, with for all $z\in\UpperHalfPlane$,
\begin{equation*}
     \Delta(z),\ \SelfEnergy(z) \in \ComplexNumbers,  
\end{equation*}
and we use the following notations
\begin{equation*}
      W^\dagger := W_p^\dagger \in \ComplexNumbers^{P-1}, \quad \NIHamiltonian_\perp := \NIHamiltonian_{\overline p} \in \ComplexNumbers^{(P-1) \times (P-1)}. 
\end{equation*}

Moreover, we stick to the \emph{half-filled} setting \cite{georges_dynamical_1996}, for which the chemical potential $\ChemicalPotential$ of the \acrlong{aim} is set to $\OnSiteRepulsion[]/2$. The Hamiltonians of interest are then on the one-hand the grand canonical Hubbard Hamiltonian $\Hamiltonian_{H}-\ChemicalPotential\TotalNumberOperator_{H}$, and on the other hand the "impurity" grand canonical Hamiltonian $\AIHamiltonian-\ChemicalPotential\TotalNumberOperator_\mathrm{imp}$ where $\TotalNumberOperator_\mathrm{imp}$ is the impurity number operator, complying with \cite[eq. 14]{georges_dynamical_1996}. 

The \acrshort{ipt} solver is based on a second order \emph{perturbation expansion} of the Matsubara's self-energy Fourier coefficients $\MatsubaraSelfEnergyFourierSeries_{\mathrm{imp},n}$ of the single-site impurity problem in the parameter $\OnSiteRepulsion[]$ : the Matsubara's self-energy Fourier coefficients of the impurity problem is approximated by 
\begin{equation*}
\forall n \in \Integers, \quad \MatsubaraSelfEnergyFourierSeries_{\mathrm{imp},n} \approx \Sigma_{\mathrm{imp},n}^{M,{\acrshort{ipt}}} :=\frac{\OnSiteRepulsion[]}{2}+\OnSiteRepulsion[]^2\int_0^\StatisticalTemperature e^{i\MatsubaraFrequency_n\tau} \left( \NIMatsubarasGreensFunction_{\mathrm{imp}}(\tau)\right)^3 d\tau,
\end{equation*}
where $\NIMatsubarasGreensFunction_{\mathrm{imp}}$ is the restriction to $\OneParticleSpace[\mathrm{imp}]$ of the Matsubara's Green's function of the non-interacting Hamiltonian $H^0_{{\rm AI}}$. The Fourier coefficients of $\NIMatsubarasGreensFunction_{\mathrm{imp}}$ are given by 
\begin{align*}
\NIMatsubaraFourierSeries_{\mathrm{imp},n}= \left(\left(\GreensFunction_{\mathrm{imp}}^0(i\MatsubaraFrequency_n)\right)^{-1} - \ChemicalPotential \right)^{-1} = \left( i\MatsubaraFrequency_n - \NIHamiltonian_{\mathrm{imp}} - \Hybridization(i\MatsubaraFrequency_n) \right)^{-1} =\left( i\MatsubaraFrequency_n - \Hybridization(i\MatsubaraFrequency_n) \right)^{-1},
\end{align*}
since $H^0_{\mathrm{imp}}=0$ and where the first equality is a consequence of a shift to enforce particle-hole symmetry \cite[p.50]{georges_dynamical_1996}. 
Finally, noticing that 
\begin{equation}
    \Hybridization(z)=W\left(z+\ChemicalPotential-\NIHamiltonian_{\perp} + \SelfEnergy(z) \right)^{-1}W^\dagger = W \left(z-\NIHamiltonian_{\perp} -(\SelfEnergy(z)-\ChemicalPotential)\right)^{-1}W^\dagger,
\end{equation}
we make the change of variable $\SelfEnergy \leftarrow \SelfEnergy - \ChemicalPotential$ and we thus have, due to the filling condition,
\begin{equation}
\SelfEnergy_{\mathrm{imp},n}^{M,\acrshort{ipt}}=\OnSiteRepulsion[]^2\int_{0}^\StatisticalTemperature e^{i\MatsubaraFrequency_n \tau} \left(\frac{1}{\StatisticalTemperature}\sum_{n' \in \RelativeIntegers}\left(i\MatsubaraFrequency_{n'} - \Hybridization (i\MatsubaraFrequency_{n'}) \right)^{-1} e^{-i\MatsubaraFrequency_{n'} \tau}  \right)^3 d \tau . \label{eq:IPT_2}
\end{equation}
The IPT approximation therefore provides an explicit expression of the Matsubara Fourier coefficients of the impurity self-energy $\SelfEnergy_{\mathrm{imp}}$ as a function of $\OnSiteRepulsion[]$ and $\Hybridization$. To reconstruct $\SelfEnergy^{\acrshort{ipt}}_{{\rm imp}}(z)$ from these Fourier coefficients, we have to solve an \acrlong{acp}. The following result shows that, in the case of finite-dimensional baths, this problem has a unique solution, which can be computed (almost) explicitly. Our results are similar to computations already obtained in \cite[eq. 4]{kajueter_new_1996}.

\begin{proposition}[\acrshort{ipt} solver for finite-dimensional baths] \label{prop:IPTFiniteDimension}
Let $U \in \RealNumbers$, $\Delta : \UpperHalfPlane \to \ComplexNumbers$ of the form
\begin{equation} \label{eq:finite_bath_Delta}
\forall z \in \UpperHalfPlane, \quad \Hybridization(z) =  
\sum_{k=1}^K \frac{a_k}{z-\eps_k}, \quad \mbox{with} \quad \eps_1 < \eps_2 < \cdots < \eps_K \mbox{ and }  a_k > 0 \mbox{ for all } 1 \le k \le K,
\end{equation}
and $(\SelfEnergy_{{\rm imp},n})_{n \in \Integers}$ defined by 
\begin{equation*}
\forall n \in \Integers, \quad \SelfEnergy_{{\rm imp},n}=\OnSiteRepulsion[]^2\int_{0}^\StatisticalTemperature e^{i\MatsubaraFrequency_n\tau}\left(\frac{1}{\StatisticalTemperature} \sum_{n' \in \RelativeIntegers}\left(i\MatsubaraFrequency_{n'} - \Hybridization(i\MatsubaraFrequency_{n'})\right)^{-1} e^{-i\MatsubaraFrequency_{n'}\tau}\right)^3 d\tau.
\end{equation*}
Then, the \acrfull{acp}
\begin{equation}
\label{eq:IPTAnalyticalContinuationProblem}
\begin{cases}
    \SelfEnergy_{\rm imp} : \UpperHalfPlane\to \overline{\UpperHalfPlane} \\ 
    \SelfEnergy_{\rm imp} \text{ is analytic} \\
    \forall n\in\Integers, \ \SelfEnergy_{\rm imp}(i\MatsubaraFrequency_n) = \SelfEnergy_{{\rm imp},n}
\end{cases}
\end{equation}
has a unique solution, which is given by 
\begin{equation} \label{eq:IPT_finite_bath}
\forall z \in \UpperHalfPlane, \quad \SelfEnergy_{\rm imp}^{\acrshort{ipt}}(z)= \OnSiteRepulsion[]^2\sum_{k_1,k_2,k_3=1}^{K+1}\frac{a'_{k_1,k_2,k_3}}{z-\left(\eps_{k_1}'+\eps_{k_2}'+\eps_{k_3}'\right)},
\end{equation}
where $\eps_1' < \eps_2' < \cdots < \eps'_{K+1}$ are the $(K+1)$ real roots of the equation
\begin{equation*}
\eps-\Hybridization(\eps)=0,
\end{equation*}
which satisfy the interlacing relation 
$$
\eps_1' < \eps_1 < \eps_2' < \eps_2 < \cdots < \eps_{K} < \eps'_{K+1},
$$
and for all $k_1,k_2,k_3 \in \IntSubSet{1}{K+1}$, $a_{k_1,k_2,k_3}$ is defined by
\begin{equation}
\label{eq:FormulaResiduesSelfEnergy}
a'_{k_1,k_2,k_3}=\frac{1+e^{-\StatisticalTemperature(\eps_{k_1}'+\eps_{k_2}'+\eps_{k_3}')}}{\left(1+e^{-\StatisticalTemperature \eps_{k_1}'}\right)\left(1+e^{-\StatisticalTemperature \eps_{k_2}'}\right)\left(1+e^{-\StatisticalTemperature \eps_{k_3}'}\right)} \prod_{i=1}^3\left(1-\Hybridization'(\eps_{k_i}')\right)^{-1} >0.
\end{equation}
\end{proposition}

We denote by 
\begin{equation*}
\SelfEnergy_{\rm imp}^{\acrshort{ipt}}=\IPTSolver_\beta(\OnSiteRepulsion[],\Hybridization)
\end{equation*}
the output of this solver. At this time, this solver is defined only for finite baths, that is for hybridization functions that are rational functions of the form \eqref{eq:finite_bath_Delta}. We will see later (in Proposition~\ref{prop:DefinitionIPTmap}) that this map can be extended by (weak) continuity to the space of all physically admissible hybridization functions. We thus finally obtain the system of translation invariant paramagnetic single-site \acrshort{ipt}-\acrshort{dmft} equations 
\begin{align}
 \forall z  \in \UpperHalfPlane, \quad \Hybridization(z)&=W\left(z-\NIHamiltonian_{\perp} - \SelfEnergy(z) \right)^{-1}W^\dagger \label{eq:IPT-DMFT_eq1} \\
\SelfEnergy&=\IPTSolver_\StatisticalTemperature(\OnSiteRepulsion[],\Hybridization) \label{eq:IPT-DMFT_eq2},
\end{align}
where the on-site interaction energies $U\in \RealNumbers$, the inverse temperature $\StatisticalTemperature > 0$, the vector $W^\dagger \in \ComplexNumbers^{P-1}$ and the matrix $H^0_{\perp} \in \ComplexNumbers^{(P-1)\times(P-1)}$ obtained from the hopping matrix $T$ are the parameters of the model, and where $\Delta:\UpperHalfPlane \to \ComplexNumbers$ and $\Sigma:\UpperHalfPlane \to \ComplexNumbers$ are the unknowns.

In the remainder of this article, our main focus will be the existence of solutions to the above equations.
\section{Main results}
\label{sec:MainResults}

Let us introduce and recall some useful notation. We denote by 
$$
\UpperHalfPlane:= \{z\in\ComplexNumbers \ | \ \Im(z)>0\}
$$
the complex open upper-half plane, that is the set of complex numbers with positive imaginary part. For a matrix $M\in\ComplexNumbers^{n\times n}$, $n\geq 1$, the imaginary part of $M$ is defined by
\begin{equation*}
	\Im(M) = \frac{M-M^\dagger}{2i}.
\end{equation*}
The set of Hermitian matrices of size $n$ is denoted by $\SelfAdjointOperator_n(\ComplexNumbers)$, and the set of positive-semidefinite matrices by $\SelfAdjointOperator_n^+(\ComplexNumbers)$. The notation $M \ge 0$ (resp. $M > 0$) means that the matrix $M$ is positive semidefinite (resp. positive definite). In the following, we will also deal with measure and probability theory. The set of finite signed Borel measures on $\RealNumbers$ is denoted by $\FiniteMeasures$. The subset $\PositiveFiniteMeasures\subset\FiniteMeasures$ is the set of finite positive Borel measures on $\RealNumbers$, and finally, the subset $\ProbabilityMeasures\subset\PositiveFiniteMeasures$ denotes the set of Borel probability measures on $\RealNumbers$. 

For a positive Borel measure $\mu$ on $\RealNumbers$, we say that $\mu$ has a finite moment of order $k\in\Integers$ if $\int_\RealNumbers|\eps|^k d\mu(\eps)<\infty$. In this case, we denote by $m_k(\mu)=\int_\RealNumbers \eps^k d\mu(\eps)$ its $k$-th moment. In particular, $\mu$ is finite if and only if it has a finite moment of order 0. In this case, $\mu\in\PositiveFiniteMeasures$ and its 0-th moment is called the \emph{mass} of $\mu$, denoted by $\mu(\RealNumbers)=m_0(\mu)=\int_\RealNumbers d\mu$. These notation and considerations extend to the case of matrix-valued measures, as will be discussed in the next section.

\subsection{Pick functions}
\label{sec:PickFunctions}

Our mathematical framework intersects with the realm of complex analysis pioneered by Pick \cite{pick_uber_1915} and Nevanlinna \cite{nevanlinna_uber_1919}, focusing on the study of so-called Pick functions.	A \emph{Pick function} is a map $f : \UpperHalfPlane\to\overline{\UpperHalfPlane}$ which is analytic. In this article, we use the term Pick functions, but several terminologies coexist in the literature: Nevanlinna functions, Herglotz functions, Riesz functions, or $R$-functions A \emph{Pick matrix} is an analytic map $f:\UpperHalfPlane\to\ComplexNumbers^{n\times n}$, $n\geq 1$, such that for all $z\in\UpperHalfPlane, \Im(f(z))\geq 0$ (i.e. $\Im(f(z))\in\SelfAdjointOperator_n^+(\ComplexNumbers)$).
Sometimes, it is convenient to extend Pick matrices to the lower-half-plane $\LowerHalfPlane$. In this case, the usual convention is to set for all $z \in \LowerHalfPlane, f(z)=f(\overline{z})^\dagger$ \cite{gesztesy_matrixvalued_2000}.
One of the most important results about Pick functions is that they have an integral representation.

\begin{theorem}[Nevanlinna-Riesz representation theorem \cite{gesztesy_matrixvalued_2000}]
\label{thm:NevanlinnaRepresentation}
	Let $f : \UpperHalfPlane\to \ComplexNumbers^{n \times n}$ be a Pick matrix. 
	There exist $a\in\SelfAdjointOperator_n^+(\ComplexNumbers)$, $b\in\SelfAdjointOperator_n(\ComplexNumbers)$ and $\mu$ a Borel $\SelfAdjointOperator_n^+(\ComplexNumbers)$-valued measure on $\RealNumbers$ such that $(1+|\eps|)^{-1}$ is $\mu$-integrable and
	\begin{equation}
	\label{eq:NevanlinnaRepresentation}
	\forall z\in\UpperHalfPlane, \quad 	f(z) = az+b+\int_\RealNumbers \left(\dfrac{1}{\eps-z}-\dfrac{\eps}{1+\eps^2}\right)d\mu(\eps).
	\end{equation}
	The measure $\mu$ is called the Nevanlinna-Riesz measure of $f$, and $a = \underset{y\to +\infty}{\lim} \dfrac{1}{iy} f(iy)$, $b = \Re(f(i)):=(f(i)+f(i)^\dagger)/2$.	
	In the particular case of Pick functions, i.e. $n=1$, we have $a\geq 0$, $b\in\RealNumbers$, and $\mu$ is a positive Borel measure on $\RealNumbers$, with the same integrability condition.
\end{theorem}

The following theorem extends to Pick matrices a result on Pick functions which can be found in~\cite{eckhardt_continued_2022} and~\cite[Theorem 3.2.1]{akhiezer_classical_2020}. It states that the moments of the Nevanlinna-Riesz measure of a Pick function or matrix are related to its expansion on the imaginary axis at $+\infty$.

\begin{theorem}
\label{thm:MomentsAsymptoticExpansion}
	Let $f:\UpperHalfPlane\to\ComplexNumbers^{n\times n}$ be a Pick matrix and $\mu$ its Nevanlinna-Riesz measure. Let $n\in\Integers$. The function $f$ satisfies:
	\begin{equation}
	\label{eq:AsymptoticExpansion}
		-f(iy)  = m_{-2} (iy) + m_{-1} + \frac{m_0}{iy} + \frac{m_1}{(iy)^2} +\frac{m_2}{(iy)^3}+\dots+\frac{m_{2n}}{(iy)^{2n+1}}+o_{y \to +\infty} \left(\frac{1}{y^{2n+1}}\right)
	\end{equation}
	if and only if $\mu$ has finite moments of order less than or equal to $2n$, i.e. for all $x\in\ComplexNumbers^n$, $\HermitianProduct{x}{\int_\RealNumbers |\eps|^kd\mu(\eps)x}<\infty$ for $0\leq k\leq 2n$. For $0\leq k \leq 2n$, the coefficient $m_k$ is then the $k$-th moment of $\mu$, i.e. $m_k = \int_\RealNumbers \eps^kd\mu(\eps)\in\SelfAdjointOperator_n(\ComplexNumbers)$.
\end{theorem}

\begin{proof}
	The result for Pick functions can be found in \cite{eckhardt_continued_2022} and~\cite[Theorem 3.2.1]{akhiezer_classical_2020}. To extend the result to Pick matrices, it suffices to notice  $f$ is a Pick matrix if and only if for all $x\in\ComplexNumbers^n$, the map $f_x:z\in\UpperHalfPlane\mapsto \HermitianProduct{x}{f(z)x}$ is a Pick function.
\end{proof}

As mentioned in the previous section, Pick matrices are related to the study of Green's functions methods in general, and to \acrshort{dmft} in particular because of the following result.

\begin{proposition}[$-\GreensFunction,-\Hybridization,-\SelfEnergy$ are Pick]
\label{prop:PickFunctionsDMFT} $\;$
\begin{enumerate}
\item Let $\Hamiltonian$ be a Hamiltonian on ${\rm Fock}(\ComplexNumbers^n)$, the Fock space associated to $\OneParticleSpace\simeq \ComplexNumbers^n$, and $\GreensFunction : \UpperHalfPlane \to \ComplexNumbers^{n \times n}$  the one-body Green's function of $\Hamiltonian$ in an equilibrium state. Then $-\GreensFunction$ is a Pick matrix.
\item Let $\Hamiltonian$ be a Hamiltonian  on ${\rm Fock}(\ComplexNumbers^n)$, with non-interacting Hamiltonian $\SecondQuantization{\NIHamiltonian}$, $\GreensFunction^0$ the one-body Green's function of $\SecondQuantization{\NIHamiltonian}$, $\GreensFunction$ the one-body Green's function of $\Hamiltonian$ in an equilibrium state of $\Hamiltonian$, and $\SelfEnergy : \UpperHalfPlane \to \ComplexNumbers^{n \times n}$ the self-energy defined by
 \begin{equation*}
     \SelfEnergy(z) := \GreensFunction^0(z)^{-1} - \GreensFunction(z)^{-1}.
 \end{equation*}
 Then $-\Sigma$ is a Pick matrix.
	\item Let $\Hybridization :\UpperHalfPlane \to \ComplexNumbers^{n \times n}$ be the hybridization function of some \acrfull{aim} with impurity one-particle state space $\ComplexNumbers^n$. Then $-\Hybridization$ is a Pick matrix.
\end{enumerate}
\end{proposition}
\begin{remark}
In condensed matter physics, the Nevanlinna-Riesz measure of $-G$ is the so-called spectral function \cite{martin_interacting_2016}.
\end{remark}

The following proposition highlights the fact that, the single-site translation-invariant IPT-DMFT equations have no solution with hybridization functions of a finite-dimensional \acrshort{aim}. 

\begin{proposition}[Non-existence of solutions to the finite-dimensional bath single-site paramagnetic translation-invariant \acrshort{ipt}-\acrshort{dmft} equations]
\label{prop:NonExistenceFiniteBathIPTDMFT}
    Apart from the limit cases described in Proposition \ref{prop:DMFTTrivialLimitsExactness}, the single-site paramagnetic translation-invariant \acrshort{ipt}-\acrshort{dmft} equations \eqref{eq:IPT-DMFT_eq1}-\eqref{eq:IPT-DMFT_eq2} have no finite-dimensional bath solution, that is no solution $\Hybridization$ of the form
	\begin{equation*}
		\Hybridization(z) = \sum_{k=1}^K \frac{a_k}{z-\eps_k},\quad  K\geq 1, \quad a_k > 0, \quad \eps_k \in \RealNumbers.
	\end{equation*}
\end{proposition}

This implies that finding a solution to the DMFT equations requires considering  infinite-dimensional bath hybridization functions. The appropriate function space can be characterized in terms of Nevanlinna-Riesz measures, as will be shown in the subsequent section.

\subsection{Functional setting: the \acrlong{ba} and \acrshort{ipt} maps}

The \acrshort{dmft} map is the composition of the \acrshort{ipt}  map and the \acrfull{ba} map, which we will study separately.
Before focusing on the paramagnetic single-site translation-invariant case, let us get back to the general case presented in Section \ref{subsec:DMFTEquations}.
First, we formalize in our setting the definition of the \acrshort{ba} map $\SCmap$.
It has been proved by Lindsey, Lin and Schneider \cite{lindsey_quantum_2019} that the \acrshort{ba} map is well defined when the Nevanlinna-Riesz measure of each self-energy is a finite sum of Dirac measures, which means in particular that the self-energy is a rational function. The following proposition extends this result to the case of finite $\SelfAdjointOperator_n^+(\ComplexNumbers)$-valued measures, by using a different approach. In the following, we will denote by $\FragmentSize:=|\HubbardVertices_p|$ the size of the $p$-th fragment, that is the cardinality of the subgraph $\HubbardVertices_p\subset\HubbardVertices$, and identify $\OneParticleSpace[p]$ with $\ComplexNumbers^{2\FragmentSize}$ for convenience. Recall that the cardinality of the graph $\HubbardVertices$ is denoted by $L=|\HubbardVertices|$. The spaces to which  the self-energies $\SelfEnergy_p$ and the hybridization functions $\Hybridization_p$ belong are respectively given by
\begin{equation}
\label{eq:GeneralSelfEnergySpace}
	\SelfEnergySpace_p = \left\{z\in\UpperHalfPlane\mapsto C+\PickFunction{\mu} \ ; \ C \in\HermitianMatricesFragment, \mu \in \FinitePositiveMatrixMeasures{2\FragmentSize} \right\},
\end{equation}
where $\FinitePositiveMatrixMeasures{n}$ is the set of finite $\SelfAdjointOperator_n^+(\ComplexNumbers)$-valued Borel measures on $\RealNumbers$, and 
\begin{equation*}
	\HybridizationSpace_p =\left\{z\in\UpperHalfPlane\mapsto \PickFunction{\nu}\ ; \ \nu\in\FinitePositiveMatrixMeasures{2\FragmentSize}, \ \nu(\RealNumbers)=\NIHamiltonianSchurBlock^\dagger\NIHamiltonianSchurBlock \right\},
\end{equation*}
where $W_p\in\ComplexNumbers^{2(L-\FragmentSize)\times 2\FragmentSize}$ is defined in \eqref{eq:G_DMFT_p}-\eqref{eq:BlockStructure}.
These definitions are motivated by the consequences of Proposition \ref{prop:NonExistenceFiniteBathIPTDMFT} and the statements of Propositions \ref{prop:WellPosednessSelfEnergyToHybridization} and \ref{prop:DefinitionIPTmap}.

\begin{proposition}[\acrlong{ba} map: $\SelfEnergy\mapsto\Hybridization$]
\label{prop:WellPosednessSelfEnergyToHybridization}
For $1\leq p \leq \PartitionSize$, let $\SelfEnergy_p\in\SelfEnergySpace_p$, and let $\mu_p$ be the Nevanlinna-Riesz measure associated to $\SelfEnergy_p$. For $1\leq p\leq \PartitionSize$, the  hybridizations functions $\Hybridization_p$ given by 
\begin{equation}
\label{eq:HybridizationFormulaSCmap}
\Hybridization_p(z)=\NIHamiltonianSchurBlock \left( z - \NIHamiltonianSchurCentralBlock - \bigoplus_{q\neq p} \SelfEnergy_{q}(z) \right)^{-1} \NIHamiltonianSchurBlock^\dagger
\end{equation}
for $z\in\UpperHalfPlane$ are well-defined. With this definition, $-\Hybridization_p$ is a Pick matrix and there exists  a finite measure $\nu_p\in\FinitePositiveMatrixMeasures{2\FragmentSize}$ such that 
\begin{equation}
\label{eq:NevanlinnaMeasureHybridizationSCmap}
\forall z \in \UpperHalfPlane, \quad 
	\Hybridization_p(z) = \PickFunction{\nu_p} \quad \text{ and } \quad \nu_p(\RealNumbers) = \NIHamiltonianSchurBlock\NIHamiltonianSchurBlock^\dagger,
\end{equation}
namely $\Hybridization_p \in \HybridizationSpace_p$.
\end{proposition}

In the particular case of the single-site paramagnetic translation-invariant \acrshort{ipt}-\acrshort{dmft} framework, we have $\FragmentSize=1$ for all $1\leq p \leq \PartitionSize$ and the Nevanlinna-Riesz measures of the self-energies and the hybridization functions are finite positive Borel measures on $\RealNumbers$. The self-energy space \eqref{eq:GeneralSelfEnergySpace} takes the simpler form:
\begin{equation}
\label{eq:SelfEnergySpaceIPT}
	\SelfEnergySpace^{\acrshort{ipt}} = \left\{ \UpperHalfPlane \ni z\mapsto U^2\PickFunction{\mu} \in \ComplexNumbers \ ; \ \mu\in\PositiveFiniteMeasures\right\},
\end{equation}
and is therefore in one-to-one correspondence with the set $\PositiveFiniteMeasures$ of finite positive Borel measures on $\RealNumbers$. The hybridization space becomes
\begin{equation}
\label{eq:HybridizationSpace}
	\HybridizationSpace^{\acrshort{ipt}} = \left\{ \UpperHalfPlane \ni z \mapsto |\NIHamiltonianSchurBlock[]|^2  \PickFunction{\nu} \in \ComplexNumbers\ ; \ \nu\in\ProbabilityMeasures\right\},
\end{equation}
and is thus in one-to-one correspondence with the set $\ProbabilityMeasures$ of Borel probability measures on $\RealNumbers$.
These one-to-one correspondences allow us to focus on the measure spaces $\PositiveFiniteMeasures$ and $\ProbabilityMeasures$, and we will study the \acrshort{dmft} loop in terms of measures from now on. The \acrshort{ba} map can then be defined as a function $\SCmap$ between measure spaces as follows.

\begin{corollary}[\acrshort{ba} map in the \acrshort{ipt}-\acrshort{dmft} framework: $\SCmap: \SelfEnergy\mapsto\Hybridization$]
The \acrfull{ba} map in the paramagnetic single-site translation-invariant \acrshort{ipt}-\acrshort{dmft} framework  is defined as the function $\SCmap : \PositiveFiniteMeasures\to\ProbabilityMeasures$ such that 
\begin{equation*}
\SCmap(\mu) = \nu,
\end{equation*}
with
\begin{equation}
\label{eq:BathUpdateIPT}
	\HybridizationMass \PickFunction{\nu} = \NIHamiltonianSchurBlock[] \left( z - \NIHamiltonian_\perp - U^2\PickFunction{\mu} \right)^{-1} \NIHamiltonianSchurBlock[]^\dagger.
\end{equation}
\end{corollary}

In the \acrshort{dmft} loop, the impurity solver is the focus of the second stage. Within our model, we define the \acrshort{ipt} map $\IPTmap$, which transforms a given hybridization function $\Hybridization$ into a self-energy $\SelfEnergy$. As well as the \acrlong{ba} map, this mapping operates across measure spaces and maps Borel probability measures to finite positive Borel measures.

\begin{proposition}[Definition of the \acrshort{ipt} map ]
\label{prop:DefinitionIPTmap}
	Let $\nu \in \ProbabilityMeasures$ and $\Hybridization \in \HybridizationSpace^{\acrshort{ipt}}$ the hybridization function associated with $\nu$: for all $z \in \UpperHalfPlane$,
	\begin{equation*}
		\Hybridization(z)= |\NIHamiltonianSchurBlock[]|^2 \int_\RealNumbers \frac{d\nu(\eps)}{z-\eps}.
	\end{equation*}
	There exists $\xi \in \ProbabilityMeasures$, such that for all $z\in\UpperHalfPlane$,
	\begin{equation}
	\label{eq:IPTintermediate}
		\int_\RealNumbers \frac{d\xi(\eps)}{z-\eps} = \frac{1}{z-\Hybridization(z)}.
	\end{equation}
 Then define
 \begin{equation}
     \label{eq:IPTintermediate2}
     \tilde{\xi}(d\eps) := \dfrac{\xi(d\eps)}{1+e^{-\beta\eps}},
 \end{equation}
 \begin{equation}
     \label{eq:IPTintermediate3}
     \tilde{\mu} := \tilde{\xi}*\tilde{\xi}*\tilde{\xi},
 \end{equation}
 where $*$ is the  convolution product, and
 \begin{equation}
 \label{eq:IPTintermediate4}
     \mu(d\eps) := (1+e^{-\beta\eps})\tilde{\mu}(d\eps).
    \end{equation}
    Finally, define the self-energy associated to the measure $\mu$ by
	\begin{equation}
 \label{eq:SelfEnergyIPT}
		\SelfEnergy(z) = U^2\int_\RealNumbers \frac{d\mu(\eps)}{z-\eps}.
	\end{equation}
	The measure $\mu$ is a positive finite measure: $\mu \in \PositiveFiniteMeasures$, hence $\SelfEnergy \in \SelfEnergySpace^{\acrshort{ipt}}$, where $\SelfEnergySpace^{\acrshort{ipt}}$ is defined in \eqref{eq:SelfEnergySpaceIPT}.
	The IPT map $\IPTmap : \ProbabilityMeasures \to \PositiveFiniteMeasures$ is defined by 
\begin{equation*}
\IPTmap(\nu) = \mu.
\end{equation*}
Moreover, the map $\HybridizationSpace^{\acrshort{ipt}}\ni \Hybridization \mapsto \SelfEnergy\in\SelfEnergySpace^{\acrshort{ipt}}$ defined by \eqref{eq:IPTintermediate}-\eqref{eq:SelfEnergyIPT} is continuous with respect to the weak topology of measures, and coincides with the IPT solver defined in Proposition \ref{prop:IPTFiniteDimension} on finite-dimensional bath hybridization functions, hence it is its unique continuous extension to the set $\HybridizationSpace[]^{\acrshort{ipt}} $.
\end{proposition}

Now that we have defined the maps $\IPTmap$ and $\SCmap$, we define the paramagnetic translation-invariant single-site \acrshort{ipt}-\acrshort{dmft} map  as
$$
\DMFTmap := \SCmap\circ\IPTmap : \ProbabilityMeasures \to \ProbabilityMeasures.
$$
For the sake of brevity, we will sometimes refer to $\DMFTmap$ as the \acrshort{ipt}-\acrshort{dmft} map rather than restating all the assumptions: paramagnetic framework, translation-invariant and single-site. In the same way, we refer to fixed points of $\DMFTmap$ as \acrshort{ipt}-\acrshort{dmft} solutions. Such fixed points (more precisely the associated hybridization functions and self-energies) are indeed solutions of the single-site paramagnetic translation-invariant \acrshort{ipt}-\acrshort{dmft} equations \eqref{eq:IPT-DMFT_eq1}-\eqref{eq:IPT-DMFT_eq2}.

\subsection{Existence and properties of \acrshort{ipt}-\acrshort{dmft} solutions}

The main result of this paper is the existence of a solution to the \acrshort{dmft} equations \eqref{eq:IPT-DMFT_eq1}-\eqref{eq:IPT-DMFT_eq2} in the paramagnetic translation-invariant single-site  framework, using the \acrshort{ipt} impurity solver.

\begin{theorem}[Existence of a fixed point]
\label{thm:Existence}
The IPT-DMFT map $\DMFTmap$ has a fixed point $\nu \in \ProbabilityMeasures$.
\end{theorem}

\medskip

Moreover, \acrshort{ipt}-\acrshort{dmft} solutions have finite moments of all orders.

\medskip

\begin{proposition}
\label{prop:RegularizationDMFTLoop}
Let $\nu^0 \in \ProbabilityMeasures$, $\nu=\DMFTmap(\nu^0)$, and $k\in 2\Integers$. If $\nu^0\in\ProbabilityMeasures$ has finite $k$-th moment, then  $\nu$ has finite $(k+4)$-th moment. In particular, any fixed point of the IPT-DMFT map has finite moments of all orders.
\end{proposition}

\section{Proofs}

In this section, we give the proofs of the results stated in Section \ref{sec:MathematicalFramework} and Section \ref{sec:MainResults}. Among other things, we will make use of the results stated in Section \ref{sec:PickFunctions} about Pick functions and of some results from measure theory, which will be recalled when needed. As we will discuss continuity of functions defined on measure spaces, we must specify the topology we are considering. Recall that a sequence $(\mu_n)_{n \in \Integers}$ of finite Borel measures on $\RealNumbers$ is said to converge weakly to $\mu$ if for all bounded continuous function $f\in\BoundedContinuousFunctions$, \begin{equation} \label{eq:weak_CV}
  \int_\RealNumbers f \, d\mu_n \to \int_\RealNumbers f \, d\mu.   
\end{equation}
It  converges vaguely to $\mu$ if \eqref{eq:weak_CV} holds for all $f \in \ContinuousGoingToZeroFunctions$, the space of continuous functions from $\RealNumbers$ to~$\RealNumbers$ vanishing at infinity. Weak convergence clearly implies  vague convergence, and the converse is also true if all the $\mu_n$'s are probability measures, since $\RealNumbers$ is locally compact \cite{villani_optimal_2009}. We will also make use of the notions of Wasserstein distance and optimal transportation on $\RealNumbers$. The Wasserstein 2-distance between two Borel probability measures $\mu$ and $\nu$ on $\RealNumbers$ is defined by
\begin{equation}
\label{eq:WassersteinDistance}
	W_2(\mu,\nu) := \left(\underset{\pi\in\Pi(\mu,\nu)}{\inf} \int_{\RealNumbers^2}|x-y|^2d\pi(x,y)\right)^{\frac{1}{2}},
\end{equation}
where $\Pi(\mu,\nu)$ is the set of all couplings of $\mu$ and $\nu$, i.e. of Borel probability measures on $\RealNumbers \times \RealNumbers$ whose marginals with respect to the first and second variables are respectively $\mu$ and $\nu$. The infimum in the definition \eqref{eq:WassersteinDistance} is actually a minimum, and there exists in fact a unique $\pi_{\mu,\nu} \in\Pi(\mu,\nu)$ such that $W_2(\mu,\nu)^2=\int |x-y|^2d\pi_{\mu,\nu}(x,y)$ \cite{santambrogio_optimal_2015}.

\subsection{Proofs of the results in Section \ref{sec:MathematicalFramework}}

Most of the results presented in Section \ref{sec:MathematicalFramework} are known in other settings similar to ours. However, the proofs of Propositions \ref{prop:NonInteractingGreensFunction} and \ref{thm:SelfEnergyImpurityPattern} found in the literature are limited to specific states (such as ground states or Gibbs states) and do not emphasize the importance of the notion of \emph{equilibrium} state. Our proofs overcome this artificial distinction. Additionally, Proposition \ref{prop:DMFTTrivialLimitsExactness} is often regarded as obvious in the physics literature, but it is typically proven only for translation-invariant settings. Our proof allows for the computation of the Green's function of a strictly interacting Hubbard model and facilitates understanding of the DMFT equations, which we hope will aid the reader in grasping the machinery introduced in this section.
Furthermore, Theorem \ref{thm:MatsubaraFourierAnalyticExtension} has long been considered proven in \cite{baym_determination_1961} within the physics community. However, the proof given in \cite{baym_determination_1961} does not utilize classical analytic continuation techniques known at the time. Our proof is entirely different and relies precisely on analytic continuation techniques. We hope it sheds light on certain aspects of \cite{baym_determination_1961} in this specific \emph{finite}-dimensional Hilbert space setting. In particular, our proof of uniqueness of the analytic continuation of the Matsubara's Fourier coefficients does not rely on the asymptotic behavior of the Green's function, but on the properties of Pick functions.

\subsubsection{Proof of Proposition \ref{prop:NonInteractingGreensFunction}}
\label{sec:ProofNIGreensFunction}

Our proof is based on the time evolution of annihilation/creation propagators for ideal Fermi gases. More precisely, one has, as detailed in \cite[p.46]{bratteli_operator_1997},
\begin{equation}\label{eq:TimeEvolutionIdealFermiGaz}
\HeisenbergPicture{\AnnihilationOperator[\OneParticleState]}(t)=\AnnihilationOperator[e^{it\NIHamiltonian}\OneParticleState].
\end{equation}

For all $z \in \UpperHalfPlane$, we have to show that $(z-\NIHamiltonian)\GreensFunction(z)=(z-\NIHamiltonian)(\GreensFunction_{+}(z)+\GreensFunction_{-}(\bar{z})^\dagger)=\Identity$. Integrating by parts, one has

\begin{equation}
(z-\NIHamiltonian)\GreensFunction_{+}(z)=i\TimeGreensFunction(0^+)+\int_{\RealNumbers_+^*}e^{izt}\left(i\frac{d}{dt}\TimeGreensFunction(t)-\NIHamiltonian \TimeGreensFunction(t)\right) dt, \nonumber
\end{equation}

and for $t > 0$, $i\TimeGreensFunction(t)$ represents the sesquilinear form defined by

\begin{equation*}
\forall \OneParticleState,\OneParticleState' \in \OneParticleSpace, \quad\HermitianProduct{\OneParticleState}{(i\TimeGreensFunction)(t)\OneParticleState'}=\EquState(\HeisenbergPicture{\AnnihilationOperator[\OneParticleState]}(t)\CreationOperator[\OneParticleState']),
\end{equation*}

so that

\begin{equation}
\HermitianProduct{\OneParticleState}{(i\frac{d}{dt}\TimeGreensFunction)(t)\OneParticleState'}=\EquState(\frac{d}{dt}(\HeisenbergPicture{\AnnihilationOperator[\OneParticleState]})(t)\CreationOperator[\OneParticleState'])=-i\EquState(\HeisenbergPicture{\AnnihilationOperator[\NIHamiltonian\OneParticleState]}(t)\CreationOperator[\OneParticleState'])=\HermitianProduct{\NIHamiltonian\OneParticleState}{\TimeGreensFunction(t)\OneParticleState'}= \HermitianProduct{\OneParticleState}{\NIHamiltonian\TimeGreensFunction(t)\OneParticleState'}. \nonumber
\end{equation}

Similarly, another integration by parts leads to

\begin{equation}
(z-\NIHamiltonian)\GreensFunction_{-}(\bar{z})^\dagger=i\TimeGreensFunction(0^-)^\dagger + \int_{\RealNumbers_{-}}e^{-izt}(i\frac{d}{dt}\TimeGreensFunction(t)-\TimeGreensFunction(t) \NIHamiltonian)^\dagger dt. \nonumber
\end{equation}

For $t<0$, $i\TimeGreensFunction(t)$ represents the sesquilinear form defined by

\begin{equation*}
\forall \OneParticleState,\OneParticleState' \in \OneParticleSpace, \HermitianProduct{\OneParticleState}{(i\TimeGreensFunction)(t)\OneParticleState'}=-\EquState(\CreationOperator[\OneParticleState']\HeisenbergPicture{\AnnihilationOperator[\OneParticleState]}(t)).
\end{equation*}

Note that, because $\EquState$ is an equilibrium state and by the cyclic property of the trace, one has

\begin{equation}
\label{eq:TimeTranslationInvariantPropagators}
\HermitianProduct{\OneParticleState}{(i\TimeGreensFunction)(t)\OneParticleState'}=-\EquState(\HeisenbergPicture{\CreationOperator[\OneParticleState']}(-t)\AnnihilationOperator[\OneParticleState])
\end{equation}

(equilibrium propagators are time-translation-invariant), so that 

\begin{equation*}
\HermitianProduct{\OneParticleState}{(i\frac{d}{dt}\TimeGreensFunction)(t)\OneParticleState'}= i \EquState(\HeisenbergPicture{\CreationOperator[\NIHamiltonian\OneParticleState']}(-t)\AnnihilationOperator[\OneParticleState])=\HermitianProduct{\OneParticleState}{\TimeGreensFunction(t)\NIHamiltonian\OneParticleState'}.
\end{equation*}

Finally, one has for $ \OneParticleState,\OneParticleState' \in \OneParticleSpace$,

\begin{equation*}
\HermitianProduct{\OneParticleState}{(z-\NIHamiltonian)\GreensFunction(z)\OneParticleState'}=\HermitianProduct{\OneParticleState}{i\TimeGreensFunction(0^+) \OneParticleState'}+ \HermitianProduct{\OneParticleState}{i\TimeGreensFunction(0^{-})^\dagger \OneParticleState'}=\EquState(\AnnihilationOperator[\OneParticleState]\CreationOperator[\OneParticleState'] +\CreationOperator[\OneParticleState']\AnnihilationOperator[\OneParticleState]) = \HermitianProduct{\OneParticleState}{\OneParticleState'},
\end{equation*}

because $\EquState$ is a state and the annihilation/creation operators satisfy the \acrshort{car}.

\subsubsection{Proof of Theorem \ref{thm:SelfEnergyImpurityPattern}}
For the simplicity of the proof, note first that an equivalent definition of an impurity problem is that there exists an \emph{impurity space} $\OneParticleSpace[\mathrm{imp}]\subset \OneParticleSpace$ such that the interacting part $\InteractingHamiltonian$ of the Hamiltonian $\Hamiltonian$ as introduced in \eqref{eq:HamiltonianNIAndIDecomposition} belongs to the following subalgebra
\begin{equation}
    \InteractingHamiltonian \in \mathcal{A}\{\CreationOperator[\OneParticleState]\AnnihilationOperator[\OneParticleState'], \OneParticleState,\OneParticleState' \in \OneParticleSpace[\mathrm{imp}]\}.
\end{equation}
In other words, the interacting part of the Hamiltonian is an element of the \acrfull{gicar} algebra generated by  $\OneParticleSpace[\mathrm{imp}]$ (see \cite{sato_gicar_2024} for a concise introduction to \acrshort{gicar} algebras).

From that, the proof is a generalization of \cite{lin_sparsity_2020}. To prove the sparsity pattern of $\SelfEnergy$, we have to prove that for all $z \in \UpperHalfPlane$, for all $\OneParticleState \in \OneParticleSpace, \OneParticleState' \in \OneParticleSpace[\mathrm{imp}]^\perp$,
\begin{equation}\label{eq:GSigmaSelfEnergyImpurityPatternEquivalence}
\HermitianProduct{\OneParticleState}{\GreensFunction(z)\SelfEnergy(z)\OneParticleState'}=0 \text{ and } \HermitianProduct{\OneParticleState'}{\SelfEnergy(z)\GreensFunction(z)\OneParticleState}=0.
\end{equation}
The first equality is equivalent to
\begin{equation*}
\HermitianProduct{\OneParticleState}{\GreensFunction(z)(z-\NIHamiltonian)\OneParticleState'}=\HermitianProduct{\OneParticleState}{\OneParticleState'},
\end{equation*}
and, similarly as in the previous proof, we have
\begin{align}
\HermitianProduct{\OneParticleState}{\GreensFunction(z)(z-\NIHamiltonian)\OneParticleState'}=\HermitianProduct{\OneParticleState}{\OneParticleState'}&+\int_{\RealNumbers_+}e^{izt}\left(\frac{d}{dt}\HermitianProduct{\OneParticleState}{i\TimeGreensFunction(t)\OneParticleState'} - \HermitianProduct{\OneParticleState}{\TimeGreensFunction(t)\NIHamiltonian\OneParticleState'} \right)dt \nonumber \\
&+\int_{\RealNumbers_-^*} e^{-izt}\left(\frac{d}{dt}\HermitianProduct{i\TimeGreensFunction(t)\OneParticleState}{\OneParticleState'}-\HermitianProduct{\TimeGreensFunction(t)\NIHamiltonian\OneParticleState}{\OneParticleState'}\right)dt. \nonumber
\end{align}
Now, for all $t \in \RealNumbers_+$, we have using the cyclicity of the trace, similarly as in \eqref{eq:TimeTranslationInvariantPropagators},
\begin{equation*}
\frac{d}{dt}\HermitianProduct{\OneParticleState}{i\TimeGreensFunction(t)\OneParticleState'}=-i \EquState\left(\AnnihilationOperator[\OneParticleState]\HeisenbergPicture{\Commutator{\Hamiltonian}{\CreationOperator[\OneParticleState']}}(-t)\right).
\end{equation*}
To compute the commutator, note first that for all $\OneParticleState_1,\OneParticleState_2 \in \OneParticleSpace[\mathrm{imp}]$, we have using the \acrshort{car}
\begin{equation*}
\Commutator{\CreationOperator[\OneParticleState_1]\AnnihilationOperator[\OneParticleState_2]}{\CreationOperator[\OneParticleState']}=\HermitianProduct{\OneParticleState_2}{\OneParticleState'}\CreationOperator[\OneParticleState_1]=0
\end{equation*}
since $\OneParticleState' \in \OneParticleSpace[\mathrm{imp}]^\perp$, so that $\CreationOperator[\OneParticleState']$ commutes with the generators of the algebra to which $\InteractingHamiltonian$ belongs, hence with $\InteractingHamiltonian$. Moreover, we have using \eqref{eq:TimeEvolutionIdealFermiGaz},
\begin{equation*}
\Commutator{\SecondQuantization{\NIHamiltonian}}{\CreationOperator[\OneParticleState']}=-i\frac{d}{dt}\left(t\mapsto \HeisenbergPicture{\CreationOperator[\OneParticleState']}_{\NIHamiltonian}\right)(0)=\CreationOperator[\NIHamiltonian \OneParticleState'],
\end{equation*}
where $\HeisenbergPicture{\cdot}_{\NIHamiltonian}$ denotes the Heisenberg picture associated to the non-interacting Hamiltonian $\SecondQuantization{\NIHamiltonian}$, so that we end up with
\begin{equation*}
\frac{d}{dt}\HermitianProduct{\OneParticleState}{i\TimeGreensFunction(t)\OneParticleState'}=-i\EquState\left(\AnnihilationOperator[\OneParticleState] \HeisenbergPicture{\CreationOperator[\NIHamiltonian\OneParticleState']}(-t)\right)=\HermitianProduct{\OneParticleState}{\TimeGreensFunction(t)\NIHamiltonian\OneParticleState'}.
\end{equation*}

One shows, using the same techniques, that for all $t \in \RealNumbers_-^*$,
\begin{equation*}
\frac{d}{dt}\HermitianProduct{i\TimeGreensFunction(t)\OneParticleState}{\OneParticleState'}=\HermitianProduct{\TimeGreensFunction(t)\NIHamiltonian\OneParticleState}{\OneParticleState'}
\end{equation*}
hence the first equality of \eqref{eq:GSigmaSelfEnergyImpurityPatternEquivalence} is proven. The second equality can be proved similarly.

\subsubsection{Proof of Proposition \ref{prop:DMFTTrivialLimitsExactness}}

We stick to the case in Remark \ref{rmk:AIMDMFTStates}, where the \acrshort{aim} states are Gibbs states at inverse temperature $\StatisticalTemperature$ and chemical potential $\ChemicalPotential$. Now for the first case, if $\OnSiteRepulsion[]=0$, the \acrshort{aim} are non-interacting, hence the Green's functions are the non-interacting Green's functions, the self-energies $\SelfEnergy_{\mathrm{imp},p}$ are identically zero, and the hybridization functions are given by, for all $z \in \UpperHalfPlane$,
\begin{equation*}
\Hybridization_p(z)=W_p\left(z-\NIHamiltonian_{\overline{p}} \right)^{-1} W_p^\dagger.
\end{equation*}

Now if the partition $\DMFTPartition$ is such that $\HubbardGraph=\bigoplus_{p=1}^\PartitionSize \HubbardGraph[p]$, or if $\HoppingMatrix[]=0$, we have
\begin{equation*}
\NIHamiltonian_{\mathrm{H},p,\overline{p}}=0,
\end{equation*}
so that the hybridization functions $\Hybridization_p$ are identically zero. This is equivalent to zero-dimensional baths, and all the \acrshort{aim}s reduce to Hubbard models defined by $(\HubbardGraph[p],\HoppingMatrix[p],\OnSiteRepulsion[p])$. Hence the self-energies $\SelfEnergy_p$ are the self-energies associated to the corresponding Hubbard models.

\subsubsection{Proof of Theorem \ref{thm:MatsubaraFourierAnalyticExtension}}

We start by proving the equality. On the one hand, the Källén-Lehmann representation \eqref{eq:KL_C+} of $\TimeGreensFunction$ associated to the grand canonical Hamiltonian $\Hamiltonian'=\Hamiltonian-\ChemicalPotential \TotalNumberOperator$ reads for $\OneParticleState,\OneParticleState' \in \OneParticleSpace$ and  $z \in \UpperHalfPlane$,
\begin{equation*}
 \HermitianProduct{\OneParticleState}{\GreensFunction(z)\OneParticleState'}=\sum_{\FockSpaceState,\FockSpaceState' \in \FockSpaceBasis}(\rho_\FockSpaceState+\rho_{\FockSpaceState'})\HermitianProduct{\FockSpaceState}{\AnnihilationOperator[\OneParticleState]\FockSpaceState'}\HermitianProduct{\FockSpaceState'}{\CreationOperator[\OneParticleState']\FockSpaceState}\frac{1}{z+\ChemicalPotential +(E_\FockSpaceState-E_{\FockSpaceState'})}.
\end{equation*}
On the other hand, the Källén-Lehmann representation of $\MatsubarasGreensFunction$ reads for $\OneParticleState,\OneParticleState' \in \OneParticleSpace$ and $ n \in \Integers$,
\begin{equation*}
 \HermitianProduct{\OneParticleState}{\MatsubaraFourierSeries_n\OneParticleState'}=\sum_{\FockSpaceState,\FockSpaceState' \in \FockSpaceBasis} \left(\rho_\FockSpaceState + \rho_\FockSpaceState e^{\StatisticalTemperature(E_\FockSpaceState-E_{\FockSpaceState'}+\ChemicalPotential)} \right)\HermitianProduct{\FockSpaceState}{\AnnihilationOperator[\OneParticleState]\FockSpaceState'}\HermitianProduct{\FockSpaceState'}{\CreationOperator[\OneParticleState']\FockSpaceState}\frac{1}{i\MatsubaraFrequency_n + \ChemicalPotential +(E_\FockSpaceState-E_{\FockSpaceState'})}.
\end{equation*}
Note now that whenever $\HermitianProduct{\FockSpaceState}{\AnnihilationOperator[\OneParticleState]\FockSpaceState'}\neq 0$, we have $N_{\FockSpaceState'}=N_{\FockSpaceState}+1$ and then
\begin{equation*}
\rho_\FockSpaceState e^{\StatisticalTemperature(E_\FockSpaceState-E_{\FockSpaceState'}+\ChemicalPotential)} = e^{-\StatisticalTemperature(E_{\FockSpaceState'}-\ChemicalPotential(N_\FockSpaceState+1))}=\rho_{\FockSpaceState'},
\end{equation*}
yielding
\begin{equation*}
\forall n \in \Integers,\ \GreensFunction(i\MatsubaraFrequency_n)=\MatsubaraFourierSeries_{n}.
\end{equation*}
To prove uniqueness, we use the fact that $-G$ is a Pick function (see Proposition \ref{prop:PickFunctionsDMFT}), and its Nevanlinna-Riesz measure is a weighted sum of finitely many Dirac measures. It follows that the \acrlong{acp} defined by $\left(i\MatsubaraFrequency_n,-\MatsubaraFourierSeries_n\right)_{n \in \Integers}$ has no other solution thanks to Theorem~\ref{thm:UniquenessNevanlinnaPickInterpolation}, which concludes.

\subsubsection{Proof of Proposition \ref{prop:IPTFiniteDimension}}

Proposition \ref{prop:IPTFiniteDimension} can actually be seen as a corollary of Proposition \ref{prop:DefinitionIPTmap}, but we give at this stage a pedestrian proof, which enlights the way the hybridization function "encapsulates" the  energy of the bath orbitals and their coupling to the impurity. We have for all $z \in \UpperHalfPlane$, 
\begin{equation*}
(z-\Hybridization(z))^{-1}=\left(z-\NIHamiltonian_{\mathrm{AIM}}\right)^{-1}_{1,1}, \quad \text{ where } \NIHamiltonian_\mathrm{AIM}=\left(\begin{matrix}
    0& \sqrt{a_1} & \sqrt{a_2} &\cdots & \sqrt{a_K} \\
    \sqrt{a_1}&\eps_1 & 0 &\cdots & 0 \\
    \sqrt{a_2}& 0 & \eps_2 & \ddots &\vdots \\
    \vdots & \vdots & \ddots& \ddots & 0 \\
    \sqrt{a}_K & 0 & \ldots & 0 & \eps_K
\end{matrix}\right),
\end{equation*}
which holds true in particular for $z=i\MatsubaraFrequency_n$. Note that $\NIHamiltonian_{\mathrm{AIM}}$ is self-adjoint and that for all $z \in \UpperHalfPlane$, using functional calculus,
\begin{equation*}
    \int_{0}^\StatisticalTemperature e^{i\MatsubaraFrequency_n\tau} \frac{-e^{-\tau\NIHamiltonian_{\mathrm{AIM}}}}{1+e^{-\StatisticalTemperature\NIHamiltonian_{\mathrm{AIM}}}}d\tau = \left( i\MatsubaraFrequency_n - \NIHamiltonian_{\mathrm{AIM}}\right)^{-1},
\end{equation*}
so that we can perform explicitly the Fourier summation :
\begin{equation}
\label{eq:SchurIPT}
    \frac{1}{\StatisticalTemperature}\sum_{n' \in \RelativeIntegers}e^{-i\MatsubaraFrequency_n\tau}(i\MatsubaraFrequency_{n'} - \Hybridization(i\MatsubaraFrequency_{n'}))^{-1}=\left(\frac{-e^{-\tau\NIHamiltonian_{\mathrm{AIM}}}}{1+e^{-\StatisticalTemperature\NIHamiltonian_\mathrm{AIM}}}\right)_{1,1}=-\sum_{k=1}^{K+1}\frac{\vert P_{1,k}\vert^2 e^{-\tau \eps'_k}}{1+e^{-\StatisticalTemperature\eps'_k}},
\end{equation}

where $P \in \ComplexNumbers^{(K+1)\times(K+1)}$ is a unitary matrix such that $\NIHamiltonian_{\mathrm{AIM}}=P\diag(\eps'_1,\cdots,\eps'_{K+1}) P^\dagger$. The right-hand side of \eqref{eq:SchurIPT} is a continuous function on $[0,\StatisticalTemperature)$ and the following integral is well-defined and reads
\begin{align*}
\int_{0}^\StatisticalTemperature e^{i\MatsubaraFrequency_n\tau}\left( \frac{1}{\StatisticalTemperature}\sum_{n'\in \RelativeIntegers}\left(i\MatsubaraFrequency_{n'} - \Hybridization(i\MatsubaraFrequency_{n'})\right)^{-1}e^{-i\MatsubaraFrequency_{n'}\tau}\right)^{3} d\tau =&\\
\sum_{k_1,k_2,k_3=1}^{K+1}\frac{1+e^{-\StatisticalTemperature(\eps'_{k_1}+\eps'_{k_2}+\eps_{k_3})}}{(1+e^{-\StatisticalTemperature\eps'_{k_1}})(1+e^{-\StatisticalTemperature\eps'_{k_2}})(1+e^{-\StatisticalTemperature\eps'_{k_3}})}& \frac{\vert P_{1,k_1}\vert^2 \vert P_{1,k_2}\vert^2 \vert P_{1,k_3}\vert^2 }{i\MatsubaraFrequency_{n}-(\eps'_{k_1}+\eps'_{k_2}+\eps'_{k_3})}.
\end{align*}
Let us now compute the spectrum of $\NIHamiltonian_{\mathrm{AIM}}$: a simple calculation shows that its characteristic polynomial $\chi_{\NIHamiltonian_\mathrm{AIM}}$ reads
\begin{equation*}
    \chi_{\NIHamiltonian_{\mathrm{AIM}}}(\eps)=\left(\prod_{k=1}^K(\eps-\eps_k)\right)\eps - \sum_{k=1}^K a_k \prod_{l=1,l\neq k}^K (\eps-\eps_l).
\end{equation*}
By assumption on the $a_k$'s and $\eps_k$'s, we have $\chi_{\NIHamiltonian_{\mathrm{AIM}}}(\eps_k)\neq 0$, so that 
\begin{equation*}
    \chi_{\NIHamiltonian_{\mathrm{AIM}}}(\eps)=0 \iff \eps-\Hybridization(\eps)=0.
\end{equation*}
This  in fact straightforwardly follows from the Schur complement approach.
Moreover, one can compute explicitly $\vert P_{1,k}\vert^2$: by definition, we have for all $k \in \IntSubSet{1}{K+1}$ and $l \in \IntSubSet{1}{K}$,
\begin{equation*}
    \sqrt{a_l}P_{1,k} + \eps_l P_{l+1,k}=\eps'_{k} P_{l+1,k} \implies \frac{a_l}{(\eps'_k-\eps_l)^2}\vert P_{1,k} \vert^2 = \vert P_{l+1,k} \vert^2,
\end{equation*}
which gives after summation on $l$ and using the fact that $P P^\dagger =1$,
\begin{equation*}
    \vert P_{1,k} \vert^2 =(1-\Hybridization'(\eps'_k))^{-1}.
\end{equation*}

This shows that $\SelfEnergy^{\mathrm{\acrshort{ipt}}}(i\MatsubaraFrequency_n)=\SelfEnergy_{\mathrm{imp},n}$, hence $\SelfEnergy^{\mathrm{\acrshort{ipt}}}$ is a solution to the \acrfull{acp} defined in \eqref{eq:IPTAnalyticalContinuationProblem}. Then, Theorem \ref{thm:UniquenessNevanlinnaPickInterpolation} ensures that there is no other solution, which concludes the proof.

\subsection{Proof of Proposition \ref{prop:PickFunctionsDMFT} ($-\GreensFunction, -\SelfEnergy, -\Hybridization$ are Pick matrices)}

The fact that the Green's function $\GreensFunction$ is a Pick matrix readily follows from the KL representation~\eqref{eq:KL_C+} and inequality \eqref{eq:G_Pick}.

\medskip

Combining \eqref{eq:G0} and \eqref{eq:self-energy}, the self-energy can be written as
$$
\forall z \in \UpperHalfPlane, \quad 
\Sigma(z)= z-H^0 - G(z)^{-1}.
$$
Recall that we know from \eqref{eq:G_Pick} that $G(z)$ is invertible for all $z \in \UpperHalfPlane$. Since $-G$ is Pick, $G^{-1}$ is Pick. This readily implies that $\Sigma$ is analytic. Let us now prove that $-\Sigma$ is Pick. First, we infer from the KL representation \eqref{eq:KL_C+} that for all $k \in \Integers$, there exists $m_0,\dots,m_{2k}\in\ComplexNumbers^{n\times n}$ such that it holds
$$
G(iy) =\frac{m_0}{iy} + \frac{m_1}{(iy)^2} +\frac{m_2}{(iy)^3}+\dots+\frac{m_{2k}}{(iy)^{2k+1}}+o_{y \to +\infty} \left(\frac{1}{y^{2k+1}}\right).
$$
Using  the anti-commutation relation $\AnnihilationOperator[\phi]\CreationOperator[\phi']+\CreationOperator[\phi']\AnnihilationOperator[\phi]=\HermitianProduct{\phi}{\phi'}$ and the normalization condition $\sum_{\psi \in \FockSpaceBasis} \rho_\psi = 1$, we obtain that
$m_0=I_n$. 
As a consequence, we have
$$
G(iy)^{-1} = (iy) I_n -m_1 + \frac{1}{iy}(m_1^2-m_2) + o_{y \to +\infty}\left( \frac 1 y \right).
$$
 In view of Theorem~\ref{thm:MomentsAsymptoticExpansion},  the Pick matrix $G^{-1}$ has a Nevanlinna-Riesz representation of the form
$$
G(z)^{-1} = z - \Sigma_\infty +\int_{\RealNumbers} \frac{d\mu(\epsilon)}{\epsilon-z},
$$
with $\Sigma_\infty \in \SelfAdjointOperator_n(\ComplexNumbers)$  and  $\mu$ a finite Borel $\SelfAdjointOperator_n^+(\ComplexNumbers)$-valued measure on $\RealNumbers$. We thus obtain that
$$
\forall z \in \UpperHalfPlane, \quad \SelfEnergy(z) = 
\SelfEnergy_\infty +\int_{\RealNumbers} \frac{d\mu(\epsilon)}{z-\epsilon},
$$
from which we deduce that $-\SelfEnergy$ is Pick. As a matter of fact, $\SelfEnergy$ is a matrix-valued rational function, that is $\mu$ is a weighted sum of finitely many Dirac measures.

  \medskip
  
  Let $\Hybridization$ be a hybridization function of an \acrshort{aim} defined as \eqref{def:HybridizationAIM}. Since $H^0_{\rm bath}$ is self-adjoint, its spectrum is real and \eqref{def:HybridizationAIM} thus defines an analytic function on $\UpperHalfPlane$. In addition, for all $z\in\UpperHalfPlane$, $\Im(z-H^0_{\rm bath})=\Im(z)>0$ hence  $\Im((z-H^0_{\rm bath})^{-1}) <0$. 
  Since the congruence preserves the sign of the imaginary part we have $\Im(\Hybridization(z))\leq 0$. This shows that $-\Hybridization$ is a Pick matrix.

\subsection{Proof of Proposition \ref{prop:NonExistenceFiniteBathIPTDMFT} (no finite-dimensional bath solution)}

In this section, we prove the statement of Proposition \ref{prop:NonExistenceFiniteBathIPTDMFT}, which states that there is no solution to the \acrshort{ipt}-\acrshort{dmft} equations considering only hybridization functions with a finite-dimensional bath.

\begin{lemma}
\label{lem:IncreaseNumberOfPoles}
	Let $f$ and $g$ be rational matrix-valued functions of size $n\geq 1$ given by 
	\begin{equation*}
		f(z)=\sum_{k=1}^K \frac{A_k}{z-\eps_k} \quad\text{and} \quad g(z) = (z-C-f(z))^{-1},
	\end{equation*}
	where $A_1,\dots,A_K\in\SelfAdjointOperator_n^+(\ComplexNumbers)\setminus\{0\}$ are positive semi-definite matrices, $\eps_1<\dots<\eps_K$ are real numbers and $C\in\SelfAdjointOperator_n(\ComplexNumbers)$. Assume that the matrices $A_1,\dots,A_K$ and $C$ commute. Then $g$ writes 
	\begin{equation}
	\label{eq:Expression_g}
		g(z) = \sum_{k=1}^{K'} \frac{A_k'}{z-\eps_k'},
	\end{equation}
	where $K'\geq K+1$,  $A_k'\in\SelfAdjointOperator_n^+(\ComplexNumbers)\setminus\{0\}$, and $\eps_k'\in\RealNumbers$, with $\eps_1'<\dots<\eps_{K'}'$.
\end{lemma}
\begin{proof}
	The fact that $\UpperHalfPlane\ni z\mapsto -g(z)$ is a Pick matrix follows from the fact that $-f$ is also a Pick matrix and that $\Im(M)>0$ implies that $M$ is invertible and that $\Im(M^{-1})<0$. Indeed, for all $z\in\UpperHalfPlane$, we have
	\begin{equation*}
		\Im \left( z - C - f(z)\right)=\Im(z)-\Im(f(z))\geq \Im(z) > 0.
	\end{equation*}
	Theorem \ref{thm:NevanlinnaRepresentation} gives a Nevanlinna-Riesz representation for $-g$, but since $g$ is a rational matrix-valued function, the Nevanlinna-Riesz measure of $-g$ is just a finite sum of Dirac measures. We thus have 
	\begin{equation*}
		g(z) = -\tilde{A}z-\tilde{B}+\sum_{k=1}^{K'} \frac{A_k'}{z-\eps_k'},
	\end{equation*}
	with the stated properties of $A_k'$ and $\eps_k'$, and $\tilde{A}\geq 0$, $\tilde{B}\in\SelfAdjointOperator_L(\ComplexNumbers)$.
	Moreover, since $g(iy)\underset{y\to +\infty}{\longrightarrow}0$ due to the definition of $g$, the affine part of $-g$ is zero. This ensures that $g$ is of the  form \eqref{eq:Expression_g}. It remains to show that the number of poles of $g$ is at least $K+1$. Because of the assumption that the matrices $A_1,\dots,A_K$ and $C$ commute, they can be codiagonalized in an orthonormal basis. Let $P$ be the unitary matrix, such that $PA_kP^\dagger=\diag(\lambda_k^1,\dots,\lambda_k^L)$ and $PCP^\dagger=\diag(c^1,\dots,c^L)$. We have
	\begin{equation*}
		g(z) = P^\dagger\diag\left(\frac{1}{z-c^1-\sum_{k=1}^K\frac{\lambda_k^1}{z-\eps_k}},\dots,\frac{1}{z-c^L-\sum_{k=1}^K\frac{\lambda_k^L}{z-\eps_k}}\right)P.
	\end{equation*}
	The set of poles of $g$ contains the union of the sets of zeros of the rational functions $u_l(z) = z-c^l-\sum_{k=1}^K\frac{\lambda_k^l}{z-\eps_k}$, for $1\leq l\leq L$. The zeros of $u_l$ are on the real axis, because $\Im(u_l(z))>0$ if $\Im(z)>0$ and $\Im(u_l(z))<0$ if $\Im(z)<0$. For $\eps\in\RealNumbers\setminus\{\eps_1,\dots,\eps_K\}$, we have $u_l'(\eps) = 1+\sum_{k=1}^K\frac{\lambda_k^l}{(\eps-\eps_k)^2}>0$ so that $u_l$ is increasing on $(-\infty,\eps_1)\cup(\eps_1,\eps2)\cup\dots\cup(\eps_K,+\infty)$. As $u_l(\eps)\underset{\eps\to -\infty}{\longrightarrow}-\infty$ and $u_l(\eps)\underset{\eps\to\eps_1,\eps<\eps_1}{\longrightarrow}+\infty$, $u_l$ has exactly 1 zero in $(-\infty,\eps_1)$ by the intermediate value theorem. The same argument shows that there is exactly one zero in each interval $(\eps_k,\eps_{k+1})$ and in the interval $(\eps_K,+\infty)$. So $u_l$ exactly $K+1$ zeros. Therefore, $g$ has more than $K+1$ poles, which concludes the proof of the lemma.
\end{proof}

Suppose $\Hybridization$ is a hybridization function associated to a bath of finite dimension and which is solution to the \acrshort{ipt}-\acrshort{dmft} equations. That is, there exist $K\geq 1$, $a_1,\dots,a_K>0$ and $\BathEnergy[1]<\dots<\BathEnergy[K]$ such that 
\begin{equation*}
	\Hybridization(z) = \sum_{k=1}^K \frac{a_k}{z-\BathEnergy}.
\end{equation*}
Let $\SelfEnergy$ be the self-energy given by the \acrshort{ipt} impurity solver, see Proposition \ref{prop:IPTFiniteDimension}.
We have
\begin{equation*}
	\SelfEnergy(z) =U^2\sum_{1\leq k_1,k_2,k_3\leq K'} \frac{a_{k_1,k_2,k_3}'}{z-\eps_{k_1}'-\eps_{k_2}'-\eps_{k_3}'} \text{ and }
\end{equation*}

\begin{equation*}
a_{k_1,k_2,k_3}'=\left(1+e^{-\StatisticalTemperature(\eps_{k_1}'+\eps_{k_2}'+\eps_{k_3}')}\right)\prod_{i=1}^3\frac{1-\Hybridization'(\eps_{k_i}')}{1+e^{-\StatisticalTemperature \eps_{k_i}'}} > 0,
\end{equation*}
where $\eps_1'<\dots<\eps_{K'}'$ are the poles of the rational function $(z-\Hybridization(z))^{-1}$. The number of poles is exactly $K+1$ and the latter are real numbers, see the proof of Lemma \ref{lem:IncreaseNumberOfPoles}, so that $K'=K+1$. Since $\OnSiteRepulsion[]>0$ by assumption, $\SelfEnergy$ has more than $K+1$ poles (we do not need a better estimation of the number of poles). So we can write
\begin{equation*}
	\SelfEnergy(z) = U^2\sum_{k=1}^{K''} \frac{a_k''}{z-\eps_k''},
\end{equation*}
with $K''\geq K+1$, $a_k''>0$ and $\eps_1<\dots<\eps_{K''}$. As $\Hybridization$ is assumed to be a solution to the \acrshort{ipt}-\acrshort{dmft} equations, it reads
\begin{equation*}
	\Hybridization(z) = W \left(z-\NIHamiltonian_\perp-\SelfEnergy(z)\right)^{-1}W^\dagger.
\end{equation*}
Applying Lemma \ref{lem:IncreaseNumberOfPoles} to the matrix-valued rational function $\left(z-\NIHamiltonian_\perp-\SelfEnergy(z)\right)^{-1}$, we know that this matrix-valued rational function has more than $K''+1$ poles, hence more than $K+2$ poles. Now, $W\neq 0$ since we have eliminated the limit cases described in Proposition \ref{prop:DMFTTrivialLimitsExactness}, hence $\Hybridization$ also has more than $K+2$ poles. As $K$ is by definition the number of poles of $\Hybridization$, it is impossible and $\Hybridization$ cannot be a solution to the \acrshort{ipt}-\acrshort{dmft} equations.

\subsection{Proof of Proposition \ref{prop:WellPosednessSelfEnergyToHybridization} (\acrlong{ba} map)}
    Let $\SelfEnergy_p\in\SelfEnergySpace_p$, for $1\leq p\leq\PartitionSize$, and let $C_p\in\HermitianMatricesFragment$ and $\mu_p\in\FinitePositiveMatrixMeasures{2\FragmentSize}$ be such that for all $z\in\UpperHalfPlane$, $\SelfEnergy_p(z) = C_p +\PickFunction{\mu_p}$.
	Since for all $1\leq p \leq \PartitionSize$, $-\SelfEnergy_p: \UpperHalfPlane \to \ComplexNumbers^{2\FragmentSize\times 2\FragmentSize}$ is a Pick matrix, 
 $$
 - \bigoplus_{q \neq p} \SelfEnergy_{q} : \UpperHalfPlane \to \ComplexNumbers^{2(L-\FragmentSize) \times 2(L-\FragmentSize)}
 $$
 is a Pick matrix. As $\NIHamiltonianSchurCentralBlock$ is Hermitian, we have for all $z\in\UpperHalfPlane$,
	\begin{equation*}
		\Im \left( z - \NIHamiltonianSchurCentralBlock - \bigoplus_{q\neq p} \SelfEnergy_{q}(z)\right)=\Im(z)-\bigoplus_{q \neq p} \Im(\SelfEnergy_{q}(z)) \geq \Im(z) > 0,
	\end{equation*}
	where  $M_1\geq M_2$ means that $M_1-M_2$ is a positive semidefinite matrix. Moreover, if $\Im(M) > 0$, then $M$ is  invertible. Thus $z-\NIHamiltonianSchurCentralBlock-\bigoplus_{q\neq j}  \SelfEnergy_{q}(z)$ is invertible and $\Hybridization_p(z)$ is well defined.
	As $\Im(M)>0$ if and only if $\Im(M^{-1})<0$, and as the congruence preserves the sign of the imaginary part, we have $\Im(\Hybridization_p(z))\leq 0$ for all $z\in\UpperHalfPlane$, so that $-\Hybridization_p$ is a Pick matrix. To show formula \eqref{eq:NevanlinnaMeasureHybridizationSCmap}, we will make use of the results on Pick functions stated in Section \ref{sec:PickFunctions}. 
	As $\mu_q$, the Nevanlinna-Riesz measure of $\SelfEnergy_q$, is finite for $1\leq q \leq \PartitionSize$, we have for $x\in\ComplexNumbers^{2\FragmentSize[q]}$, that the positive Borel measure $\mu_q^x$ defined as the Nevanlinna-Riesz measure of the Pick function $z\mapsto -\HermitianProduct{x}{\SelfEnergy_q(z)x}$, is also a finite positive measure.
	Then, for $x\in\ComplexNumbers^{2\FragmentSize[q]}$ and $y\geq 1$,
	\begin{equation*}
		\left|\left\langle x,\int_\RealNumbers \frac{d\mu_q(\eps)}{iy-\eps}x\right\rangle\right| = \left|\int_\RealNumbers \frac{d\mu_q^x(\eps)}{iy-\eps}\right|  \leq \int_\RealNumbers \frac{d\mu_q^x(\eps)}{|iy-\eps|} \leq \int_\RealNumbers d\mu_q^x < \infty.
	\end{equation*}
	This coarse upper bound is enough to ensure that 
	\begin{eqnarray}
		iy\Hybridization_p(iy) &=& iy\NIHamiltonianSchurBlock\left(iy-\NIHamiltonianSchurCentralBlock-\bigoplus_{q\neq p} \SelfEnergy_q(iy)\right)^{-1} \NIHamiltonianSchurBlock^\dagger \nonumber \\
		&=& \NIHamiltonianSchurBlock\left(1-\frac{1}{iy}\left(\NIHamiltonianSchurCentralBlock +\bigoplus_{q\neq p} \left(C_q+\int_\RealNumbers\frac{d\mu_q(\eps)}{iy-\eps}\right)\right)\right)^{-1} \NIHamiltonianSchurBlock^\dagger \nonumber \\
		 & \underset{y\to +\infty}{\longrightarrow} & \NIHamiltonianSchurBlock \NIHamiltonianSchurBlock^\dagger . \nonumber
	\end{eqnarray}
		
	This gives the expansion $\displaystyle\Hybridization_p(iy) = \frac{\NIHamiltonianSchurBlock \NIHamiltonianSchurBlock^\dagger}{iy} + o\left(\frac{1}{y}\right)$, as $y$ goes to $+\infty$. By Theorem \ref{thm:MomentsAsymptoticExpansion}, it follows that the Nevanlinna-Riesz measure of $-\Hybridization_p$, denoted by $\nu_p$, is finite, and its mass is precisely the quantity $\NIHamiltonianSchurBlock^\dagger \NIHamiltonianSchurBlock$. Thus the Nevanlinna-Riesz representation of $-\Hybridization_p$ reads
	\begin{equation*}
		-\Hybridization_p(z) = az+b-\PickFunction{\nu_p},
	\end{equation*}
	for some $a\in\PositiveMatricesFragment$ and $b\in\HermitianMatricesFragment$. Now, because of the aforementioned expansion, we must have $a=b=0$, which concludes the proof.

\subsection{Proof of Proposition \ref{prop:DefinitionIPTmap} (\acrshort{ipt} map)}
\label{sec:ProofIPTmap}

	Let  $\nu\in\ProbabilityMeasures$, and define the associated hybridization function $\Hybridization(z) = \HybridizationMass\PickFunction{\nu}$. We first establish the following lemma, which we will use several times in our analysis.
	\begin{lemma}
\label{lem:RegularizationMoments}
	Let $c \in \RealNumbers$, $\mu_0\in\PositiveFiniteMeasures$. For $z\in\UpperHalfPlane$, we set
 $$
 f(z) = \PickFunction{\mu_0} \quad \mbox{and} \quad g(z) = \displaystyle \frac{1}{z-c-f(z)}.
 $$
 Then $-g$ is a Pick function and its Nevanlinna-Riesz representation reads $g(z)=\PickFunction{\mu}$ with $\mu\in\ProbabilityMeasures$. In particular, $\mu$ has finite moments of order less than or equal to 2, given by
 $$
 m_0(\mu)=1 \mbox{ ($\mu$ is a probability measure)}, \quad m_1(\mu)=c \quad \mbox{and} \quad m_2(\mu)=\mu_0(\RealNumbers)+c^2.
 $$
\end{lemma}

\begin{proof}
The result follows from the expansion of the function $g(z)$ when $z=iy$, $y\to+\infty$ and Theorem~\ref{thm:MomentsAsymptoticExpansion}. One has
\begin{equation*}
g(iy) = \frac{1}{iy}\frac{1}{1-\frac{c}{iy}-\frac{f(iy)}{iy}}.
\end{equation*}
Since $\displaystyle f(iy)=\frac{\mu_0(\RealNumbers)}{iy}+o\left(\frac{1}{y}\right)$, we have
\begin{eqnarray}
g(iy) = \frac{1}{iy}\frac{1}{1-\frac{c}{iy}-\frac{\mu_0(\RealNumbers)}{(iy)^2}+o\left(\frac{1}{y^2}\right)} &=& \frac{1}{iy}\left(1+\frac{c}{iy}+\frac{\mu_0(\RealNumbers)+c^2}{(iy)^2}+o\left(\frac{1}{y^2}\right)\right) \nonumber \\
&=& \frac{1}{iy}+\frac{c}{(iy)^2}+\frac{\mu_0(\RealNumbers)+c^2}{(iy)^3}+o\left(\frac{1}{y^3}\right).
\nonumber \end{eqnarray}

Using again Theorem \ref{thm:MomentsAsymptoticExpansion}, we obtain the desired result.
\end{proof}
We apply Lemma \ref{lem:RegularizationMoments} with $c=0$ and $f=\Hybridization$. It follows that there exists $\xi\in\ProbabilityMeasures$ satisfying \eqref{eq:IPTintermediate}.  Then, following equations \eqref{eq:IPTintermediate2}-\eqref{eq:IPTintermediate4}, set $\tilde{\xi}(d\eps) = \dfrac{\xi(d\eps)}{1+e^{-\StatisticalTemperature\eps}}$, and $\tilde{\mu}=\tilde{\xi}*\tilde{\xi}*\tilde{\xi}$. We must verify that, setting $\mu(d\eps) = (1+e^{-\StatisticalTemperature\eps})\tilde{\mu}(d\eps)$, $\mu$ is indeed a positive finite measure, so that the self-energy $\SelfEnergy$ associated to the measure $\mu$ belongs to $\SelfEnergySpace^{\rm IPT}$. 

The multiplication by positive functions and the convolution preserves  positivity, so $\mu$ is a positive measure. Moreover, we have
\begin{equation}
\label{eq:BoundSelfEnergyMass}
	\mu(\RealNumbers) = \int_{\RealNumbers^3} \frac{\left(1+e^{-\StatisticalTemperature(\eps_1+\eps_2+\eps_3)}\right)d\xi(\eps_1)d\xi(\eps_2)d\xi(\eps_3)}{(1+e^{-\StatisticalTemperature\eps_1})(1+e^{-\StatisticalTemperature\eps_2})(1+e^{-\StatisticalTemperature\eps_3})}\leq \int_{\RealNumbers^3} d\xi(\eps_1)d\xi(\eps_2)d\xi(\eps_3) = 1.
\end{equation}
Thus $\mu$ is finite, and the map $\IPTmap$ is well defined. To prove that this map is actually continuous with respect to the weak topology of measures,  we need to establish Lemma \ref{lem:ContinuityCentralMap}. This result states the continuity of the map $\PositiveFiniteMeasures \ni \mu_0 \mapsto \mu \in \ProbabilityMeasures$ defined in Lemma~\ref{lem:RegularizationMoments}, which is central in the IPT-DMFT equations. Note that it holds 
\begin{equation}
\label{eq:CentralMap}
\forall z\in\UpperHalfPlane, \quad  \PickFunction{\mu} = \frac{1}{z-c-\PickFunction{\mu_0}}.
\end{equation}

\begin{lemma}
\label{lem:ContinuityCentralMap}
	The map $\Phi:\PositiveFiniteMeasures \ni \mu_0 \mapsto \mu\in\ProbabilityMeasures$ defined in Lemma~\ref{lem:RegularizationMoments} is weakly continuous. More precisely, the following stronger result holds true: if $(\mu_0^n)_{n \in \Integers}$ converges weakly to $\mu_0$ in $\PositiveFiniteMeasures$, then  $W_2(\Phi(\mu_0^n),\Phi(\mu))\underset{n\to\infty}{\longrightarrow}0$, where $W_2$ is the Wasserstein distance of order 2.
\end{lemma}

\begin{proof}
Let $(\mu_0^n)_{n \in \Integers} \subset \PositiveFiniteMeasures$  converging weakly  to $\mu_0$ in $\PositiveFiniteMeasures$, $\mu^n:=\Phi(\mu^n_0)$ and $\mu:=\Phi(\mu_0)$. In view of Lemma~\ref{lem:RegularizationMoments}, $\mu^n$ has finite moments of orders 1 and 2, given by $m_1(\mu^n)=c$ and $m_2(\mu^n)=\mu_0^n(\RealNumbers)+c^2$.

As $\mu_0^n$  converges weakly to $\mu_0$, we have $\int_\RealNumbers f d\mu_0^n \to \int_\RealNumbers f d\mu_0$ for all $f \in \BoundedContinuousFunctions$. Taking $f \equiv 1$, we get $\mu_0^n(\RealNumbers) = \int_\RealNumbers d\mu_0^n \underset{n\to\infty}{\longrightarrow} \int_\RealNumbers d\mu_0 = \mu_0(\RealNumbers)$. Hence, $m_2(\mu^n)=\mu_0^n(\RealNumbers)+c^2\underset{n\to\infty}{\longrightarrow} \mu_0(\RealNumbers)+c^2=m_2(\mu)$.

Now, for all $z\in\UpperHalfPlane$, the function $\RealNumbers \ni \eps\mapsto\dfrac{1}{z-\eps} \in \ComplexNumbers$ is also bounded and continuous. Thus, we can pass to the limit in formula \eqref{eq:CentralMap}:
\begin{equation}
\label{eq:LimitCentralMap}
\int_\RealNumbers \frac{d\mu^n(\eps)}{z-\eps} = \frac{1}{z-c-\int_\RealNumbers\frac{d\mu_0^n(\eps)}{z-\eps}}\underset{n\to\infty}{\longrightarrow} \frac{1}{z-c-\int_\RealNumbers\frac{d\mu_0(\eps)}{z-\eps}} = \int_\RealNumbers \frac{d\mu(\eps)}{z-\eps}.
\end{equation}

This can be extended to complex numbers $z\in\ComplexNumbers\setminus\RealNumbers$ by taking the complex conjugate of the limit~\eqref{eq:LimitCentralMap}. Let $\mathcal{A}$ be the algebra generated by the functions $\displaystyle \RealNumbers \ni \eps\mapsto\frac{1}{z-\eps} \in \ComplexNumbers$ for $z\in\ComplexNumbers\setminus\RealNumbers$. It is known that $\mathcal{A}$ is dense in $\ContinuousGoingToZeroFunctions$ (this can be shown using Helffer-Sjöstrand formula \cite{helffer_equation_2005}). Together with~\eqref{eq:LimitCentralMap}, this implies that for all $f\in\ContinuousGoingToZeroFunctions$, $\int_\RealNumbers fd\mu^n \underset{n\to\infty}{\longrightarrow}\int_\RealNumbers fd\mu$, i.e. that $(\mu^n)_{n \in \Integers}$ vaguely converges to $\mu$. Since $\RealNumbers$ is locally compact and the $\mu^n$'s are probability measures, the vague convergence is equivalent in this case to the weak convergence. The map $\Phi$ is therefore  weakly continuous. Since on the space of Borel probability measures on $\RealNumbers$, $W_2$-convergence is equivalent to weak convergence and the convergence of the second moment \cite[Section 7.1]{ambrosio_gradient_2005}, the proof is complete.
\end{proof}

Let us now prove that the impurity solver $\IPTmap : \ProbabilityMeasures\to\PositiveFiniteMeasures$ is weakly continuous.
Let $(\nu^n)_{n \in \Integers} \subset\ProbabilityMeasures$  converging weakly to $\nu\in\ProbabilityMeasures$. Define $\xi^n\in\ProbabilityMeasures$ and $\mu^n:=\IPTmap(\nu^n)$ as in Proposition \ref{prop:DefinitionIPTmap}, see equation \eqref{eq:IPTintermediate}, as well as $\xi\in\ProbabilityMeasures$ and $\mu:=\IPTmap(\nu)$. We want to show that $\mu^n$  converges weakly to $\mu$. First, because of the definition of $\xi$ through equation \eqref{eq:IPTintermediate}, and thanks to Lemma \ref{lem:ContinuityCentralMap}, we know that $W_2(\xi^n,\xi)\underset{n\to\infty}{\longrightarrow}0$. Moreover, using the same density argument as in the proof of Lemma \ref{lem:ContinuityCentralMap}, it is sufficient to show that for all $z\in\UpperHalfPlane$, the following convergence holds:
	\begin{equation*}
		\PickFunction{\mu^n}\underset{n\to\infty}{\longrightarrow}\PickFunction{\mu}.
	\end{equation*}
	We have for $z\in\UpperHalfPlane$,
	\begin{equation*}
		\PickFunction{\mu^n} = \int_{\RealNumbers^3} \frac{1}{z-(\eps_1+\eps_2+\eps_3)}\frac{1+e^{-\StatisticalTemperature(\eps_1+\eps_2+\eps_3)}}{(1+e^{-\StatisticalTemperature\eps_1})(1+e^{-\StatisticalTemperature\eps_2})(1+e^{-\StatisticalTemperature\eps_3})}d\xi^n(\eps_1)d\xi^n(\eps_2)d\xi^n(\eps_3).
	\end{equation*}
	Let $\displaystyle\psi(\eps_1,\eps_2,\eps_3) = \frac{1}{z-(\eps_1+\eps_2+\eps_3)}\frac{1+e^{-\StatisticalTemperature(\eps_1+\eps_2+\eps_3)}}{(1+e^{-\StatisticalTemperature\eps_1})(1+e^{-\StatisticalTemperature\eps_2})(1+e^{-\StatisticalTemperature\eps_3})}$, for $(\eps_1,\eps_2,\eps_3)\in\RealNumbers^3$. Then,
	\begin{eqnarray}
		\left|\PickFunction{\mu^n}-\PickFunction{\mu}\right| &=& \left|\int_{\RealNumbers^3} \psi d\xi^nd\xi^nd\xi^n - \int_{\RealNumbers^3} \psi d\xi d\xi d\xi\right| \nonumber \\
		&\leq & \left|\int_{\RealNumbers^3} \psi d\xi^nd\xi^nd\xi^n - \int_{\RealNumbers^3} \psi d\xi^n d\xi^n d\xi\right| \label{eq:SimilarTerm1} \\
		&  +& \left|\int_{\RealNumbers^3} \psi d\xi^nd\xi^nd\xi - \int_{\RealNumbers^3} \psi d\xi^n d\xi d\xi\right| \label{eq:SimilarTerm2} \\
		&  +& \left|\int_{\RealNumbers^3} \psi d\xi^nd\xi d\xi - \int_{\RealNumbers^3} \psi d\xi d\xi d\xi\right|. \label{eq:TermToBound}
	\end{eqnarray}
	We now prove that the last term \eqref{eq:TermToBound} goes to zero when $n$ goes to $\infty$. The same arguments apply to the other two terms, \eqref{eq:SimilarTerm1} and \eqref{eq:SimilarTerm2}.  The function $\psi$ is  smooth and a simple calculation shows that its partial derivative with respect to $\eps_1$ is bounded by
	\begin{equation}
 \label{eq:BoundDerivativePsi}
		|\partial_1\psi(\eps_1,\eps_2,\eps_3)| \leq \frac{1}{|\Im(z)|^2}+\frac{2\StatisticalTemperature}{|\Im(z)|}=:\kappa_{z,\StatisticalTemperature}.
	\end{equation}
	Let $\eps_2,\eps_3\in\RealNumbers$. Using the $W_2$ convergence of $\xi^n$ towards $\xi$, set $\pi_n$ the optimal coupling between these two measures \cite{santambrogio_optimal_2015}. We have that
	\begin{equation*}
		\int_\RealNumbers \psi(\eps_1,\eps_2,\eps_3)d\xi^n(\eps_1) = \int_{\RealNumbers^2} \psi(\eps_1,\eps_2,\eps_3) d\pi_n(\eps_1,\eps_1').
	\end{equation*}
	By Taylor's Theorem and since $\psi$ is continuously derivable, we have
	\begin{equation*}
		\int_\RealNumbers \psi(\eps_1,\eps_2,\eps_3)d\xi^n(\eps_1) =  \int_{\RealNumbers^2} \psi(\eps_1',\eps_2,\eps_3) d\pi_n(\eps_1,\eps_1') + \int_{\RealNumbers^2} \int_{\eps_1'}^{\eps_1}(\eps_1-t) \partial_1\psi(t,\eps_2,\eps_3) dt d\pi_n(\eps_1,\eps_1').
	\end{equation*}
	On the one hand, the first term is exactly $\int_\RealNumbers \psi(\eps_1,\eps_2,\eps_3)d\xi(\eps_1)$ by definition of $\pi_n$, and on the other hand, using \eqref{eq:BoundDerivativePsi}, we can bound the second term by
	\begin{eqnarray*}
		\left|\int_{\RealNumbers^2} \int_{\eps_1'}^{\eps_1}(\eps_1-t) \partial_1\psi(t,\eps_2,\eps_3) dt d\pi_n(\eps_1,\eps_1')\right| &\leq & \frac 12\int_{\RealNumbers^2} |\eps_1-\eps_1'|^2 \|\partial_1\psi\|_\infty d\pi_n(\eps_1,\eps_1') \\
		&\leq & \frac{\kappa_{z,\StatisticalTemperature}}2 \int_{\RealNumbers^2} |\eps_1-\eps_1'|^2 d\pi_n(\eps_1,\eps_1') \\
		& = & \frac{\kappa_{z,\StatisticalTemperature}}2 W_2(\xi^n,\xi)^2.
	\end{eqnarray*}
	Finally, the term \eqref{eq:TermToBound} can be bounded by
	\begin{align*}
		\left|\int_{\RealNumbers^3} \psi d\xi^nd\xi d\xi - \int_{\RealNumbers^3} \psi d\xi d\xi d\xi\right| &\leq  \int_{\RealNumbers^2}\left| \int_\RealNumbers \psi(\eps_1,\cdot,\cdot) d\xi^n(\eps_1) - \int_\RealNumbers \psi(\eps_1,\cdot,\cdot) d\xi(\eps_1) \right|d\xi d\xi \\
		&=  \int_{\RealNumbers^2}\left|\int_{\RealNumbers^2} \int_{\eps_1'}^{\eps_1}(\eps_1-t) \partial_1\psi(t,\eps_2,\eps_3) dt d\pi_n(\eps_1,\eps_1')\right| d\xi(\eps_2) d\xi(\eps_3) \\
		 &\leq \int_{\RealNumbers^2} \frac{\kappa_{z,\StatisticalTemperature}}2 W_2(\xi^n,\xi)^2 d\xi(\eps_2) d\xi(\eps_3) =  \frac{\kappa_{z,\StatisticalTemperature}}2 W_2(\xi^n,\xi)^2 \underset{n\to\infty}{\longrightarrow}0.
	\end{align*}
This shows that the map $\IPTmap$ is weakly continuous. It remains to prove that the map $\HybridizationSpace^{\rm IPT}\ni\Hybridization\mapsto\SelfEnergy\in\SelfEnergySpace^{\rm IPT}$ defined by \eqref{eq:IPTintermediate}-\eqref{eq:SelfEnergyIPT} coincides on the set of discrete probability measures with finite support with the map defined in Proposition \ref{prop:IPTFiniteDimension}.

Let $\Hybridization\in\HybridizationSpace^{\mathrm{IPT}}$ with a Nevanlinna-Riesz measure of the form $\nu=\sum_{k=1}^Ka_k\delta_{\eps_k}$, where $a_k>0$, $\sum a_k = \HybridizationMass$ and $\eps_1<\dots<\eps_K$. It follows that the rational function $(z-\Hybridization(z))^{-1}$ is of the form $\sum_{k=1}^{K+1} \frac{{a}_k'}{z-{\eps}_k'}$ (see the proof of Lemma \ref{lem:IncreaseNumberOfPoles}). This means by \eqref{eq:IPTintermediate} that $\xi=\sum_{k=1}^{K+1}{a}_k'\delta_{{\eps}_k'}$, and the ${\eps}_k'$'s are the zeros of the rational function $z-\Hybridization(z)$, so that the residues are given by ${a}_k'=(1-\Hybridization'({\eps}_k'))^{-1}$. The self-energy $\SelfEnergy$ given by \eqref{eq:SelfEnergyIPT} then reads for all $z\in\UpperHalfPlane$,
\begin{align*}
    \SelfEnergy(z) & = U^2\int_{\RealNumbers^3} \frac{1+e^{-\StatisticalTemperature(\eps_1+\eps_2+\eps_3)}}{(1+e^{-\StatisticalTemperature\eps_1})(1+e^{-\StatisticalTemperature\eps_2})(1+e^{-\StatisticalTemperature\eps_3})} \frac{d\xi(\eps_1)d\xi(\eps_2)d\xi(\eps_3)}{z-(\eps_1+\eps_2+\eps_3)} \\
    & = U^2\sum_{k_1,k_2,k_3 = 1}^{K+1} \frac{1+e^{-\StatisticalTemperature(\eps'_{k_1}+\eps'_{k_2}+\eps_{k_3})}}{(1+e^{-\StatisticalTemperature\eps'_{k_1}})(1+e^{-\StatisticalTemperature\eps'_{k_2}})(1+e^{-\StatisticalTemperature\eps'_{k_3}})} \frac{a_{k_1}'a_{k_2}'a_{k_3}'}{z-(\eps_{k_1}'+\eps_{k_2}'+\eps_{k_3}')} \\
    & = U^2\sum_{k_1,k_2,k_3 = 1}^{K+1} \frac{a_{k_1,k_2,k_3}'}{z-(\eps_{k_1}'+\eps_{k_2}'+\eps_{k_3}')},
\end{align*}
where $a_{k_1,k_2,k_3}'$ is given by \eqref{eq:FormulaResiduesSelfEnergy}. This complies with the result stated in Proposition \ref{prop:IPTFiniteDimension}.

Finally, since the set of discrete probability measures with finite support is dense in the set of probability measures $\ProbabilityMeasures$ for the weak topology and since $\IPTmap$ is weakly continuous, $\IPTmap$ is the only continuous extension of the \acrshort{ipt} map defined in Proposition \ref{prop:IPTFiniteDimension} for a finite-dimensional bath.

\subsection{Continuity of the IPT-DMFT map}

The following result is central to prove the existence of a fixed point to the \acrshort{dmft} equations.

\begin{theorem}[Continuity of the IPT-DMFT map]
\label{thm:ContinuityDMFTmap}
    The IPT-DMFT map $F^{\rm DMFT}$ is weakly continuous on $\ProbabilityMeasures$. More precisely, the following stronger results holds true: if  $(\nu^n)_{n \in \Integers}$ converges weakly to $\nu$, then  $$ 
W_2\left(\DMFTmap(\bm{\nu^n}),\DMFTmap(\bm{\nu})\right)\underset{n\to\infty}{\longrightarrow}0, 
    $$
    where $W_2$ is the Wasserstein 2-distance.
\end{theorem}

We have proven in the previous section the continuity of the map $\IPTmap$.
In order to prove the continuity of $\SCmap$, we need to adapt Lemma \ref{lem:ContinuityCentralMap} to equation \eqref{eq:BathUpdateIPT}. Let $\mu\in\PositiveFiniteMeasures$. Applying Theorem \ref{thm:MomentsAsymptoticExpansion} to the measure $\mu$, and using the definition \eqref{eq:BathUpdateIPT} of the hybridization function $\Hybridization$ associated to $\SCmap(\mu)$, we obtain
\begin{eqnarray}
	\Hybridization(iy) & = & \frac{1}{iy} W\left( \Identity - \frac{\NIHamiltonian_\perp}{iy} - \frac{1}{iy} U^2 \int_{\RealNumbers} \frac{d\mu(\varepsilon)}{iy-\varepsilon}  \right)^{-1} W^\dagger \nonumber \\
	& = & \frac{1}{iy} W  \left( \Identity - \frac{\NIHamiltonian_\perp}{iy} - \frac{1}{iy}\left( \frac{U^2\mu(\RealNumbers)}{iy} +o\left(\frac{1}{y}\right)\right) \right)^{-1} W^\dagger \nonumber \\
	& = & W\left[\frac{1}{iy} + \frac{\NIHamiltonian_\perp}{(iy)^2} + \frac{(\NIHamiltonian_\perp)^2+ U^2 \mu(\RealNumbers)}{(iy)^3}\right]W^\dagger + o\left(\frac{1}{y^3}\right). \nonumber
\end{eqnarray}

It follows that, when $y\to +\infty$, we have the expansion
\begin{equation}
\label{eq:ExpansionSCmap}
	\int_{\RealNumbers} \frac{d\SCmap(\mu)(\varepsilon)}{iy-\varepsilon}   =  \frac{1}{\HybridizationMass}\Hybridization(iy) = \frac{1}{iy} + \frac{s^1}{(iy)^2}+\frac{s^2(\bm{\mu})}{(iy)^3} + o\left(\frac{1}{y^3}\right),
\end{equation}
with $s_p^1$ and $s_p^2(\mu)$ given by 
\begin{equation}
\label{eq:FirstMomentHybridization}
	s^1 =  \frac{W\NIHamiltonian_\perp W^\dagger}{\HybridizationMass},
\end{equation}
\begin{equation}
\label{eq:SecondMomentHybridization}
	s^2(\mu) = \frac{W\left((\NIHamiltonian_\perp)^2+ U^2 \mu(\RealNumbers)\right)W^\dagger}{\HybridizationMass}.
\end{equation}
In view of Theorem \ref{thm:MomentsAsymptoticExpansion}, this implies that $\nu:=\SCmap(\mu)$ has finite moments of orders 1 and 2, respectively given by $m_1(\nu)=s^1$ and $m_2(\nu)=s^2(\mu)$. The arguments in the proof of Lemma \ref{lem:ContinuityCentralMap} can then be used to show that the following result holds true.

\begin{proposition}
\label{prop:ContinuitySCmap}
	$\SCmap : \PositiveFiniteMeasures \to \ProbabilityMeasures$ is weakly continuous. More precisely, if $(\mu^n)_{n \in \Integers}$ converges weakly   to $\mu$ in $\PositiveFiniteMeasures$, then $W_2(\SCmap(\mu^n),\SCmap(\mu))\underset{n\to\infty}{\longrightarrow}0$.
\end{proposition}

We are now in position to complete the proof of Theorem \ref{thm:ContinuityDMFTmap}.
Let $(\nu^n)_n\subset \ProbabilityMeasures$ converging weakly to $\nu$ in $\ProbabilityMeasures$. Denoting by $\mu^n:=\IPTmap(\nu^n)$ and $\mu:=\IPTmap(\nu)$, we have shown in the proof of Proposition \ref{prop:DefinitionIPTmap} that $\mu^n$ converges weakly to $\mu$ in $\PositiveFiniteMeasures$, see Section \ref{sec:ProofIPTmap}. Proposition \ref{prop:ContinuitySCmap} then shows that  $W_2(\DMFTmap(\nu^n),\DMFTmap(\nu)) = W_2(\SCmap(\mu^n),\SCmap(\mu))\underset{n\to\infty}{\longrightarrow}0$. In particular, $\DMFTmap(\nu^n)$ converges weakly to $\DMFTmap(\nu)$.

\subsection{Proof of Theorem \ref{thm:Existence}: existence of a fixed point}

The existence of a fixed point to the IPT-DMFT map, that is a solution to the IPT-DMFT equations, is a consequence of the following fixed point theorem \cite{shapiro_fixed-point_2016}.

\begin{theorem}[Schauder-Singbal]
\label{thm:SchauderSingbal}
Let $E$ be a locally convex Hausdorff linear topological space, $C$ a nonempty closed convex subset of $E$, and $F$ a continuous map from $C$ into itself, such that $F(C)$ is contained in a compact subset of $C$. Then $F$ has a fixed point.
\end{theorem}

Let us consider the vector space $E = \FiniteMeasures$ endowed with the Kantorovitch-Rubinstein norm~$\KantorovitchRubinsteinNorm{\cdot}$. Recall that the latter is defined as 
	\begin{equation*}
		\KantorovitchRubinsteinNorm{\mu} := \sup \left\{ \int_\RealNumbers f d\mu \ ; \ f\in \Lip_1, \|f\|_\infty \leq 1 \right\},
	\end{equation*}
	where $\Lip_1$ is the set of continuous functions on $\RealNumbers$ with Lipschitz constant less than or equal to 1.

Let us then set $C := \ProbabilityMeasures= \{\mu\in E \ | \ \mu \geq 0, \ \int_\RealNumbers d\mu= 1 \}$.
Since $E$ is a normed vector space on $\RealNumbers$, it is a locally convex Hausdorff linear topological space, and $C$ is obviously a non-empty convex subset of $E$. 

Besides, on the set of finite positive measures, weak convergence is equivalent to convergence for the Kantorovitch-Rubinstein norm \cite[Chapter 8.3]{bogacev_measure_2007}. We can thus work with the weak topology on $C$.

The fact that $C$ is weakly closed means that $\ProbabilityMeasures$ is a weakly closed subset of $\FiniteMeasures$. This result can be found in \cite[Section 3.2]{bogacev_weak_2018}. Moreover, we already proved in Theorem~\ref{thm:ContinuityDMFTmap} that $\DMFTmap : C \to C$ was weakly continuous.

\medskip

To apply Schauder-Singbal's theorem to our setting, it thus remains to show that $\DMFTmap(C)$ is relatively compact for the weak topology. This is in fact a consequence of Prokhorov's Theorem~\cite[Theorem 2.3.4]{bogacev_weak_2018}.

Indeed, let $\nu\in \DMFTmap(C)$ and $\nu_0\in\ProbabilityMeasures$, $\mu=\IPTmap(\nu_0)\in\PositiveFiniteMeasures$ such that $\nu=\DMFTmap(\nu_0)=\SCmap(\mu)$. As we have seen in the proof of Proposition \ref{prop:ContinuitySCmap}, $\nu$ has finite moments of order 1 and 2, given respectively by $m_1(\nu)=s^1$ and $m_2(\nu)=s^2(\mu)$, where $s^1$ is defined by \eqref{eq:FirstMomentHybridization} and $s^2(\mu)$ by \eqref{eq:SecondMomentHybridization}.
The inequality \eqref{eq:BoundSelfEnergyMass} states that  the mass of the measure $\mu$ is bounded by 1. Hence, $m_2(\nu)=s^2(\mu)$ is bounded independently on $\nu$: 
\begin{equation}
\label{eq:BoundSecondMomentHybridization}
	m_2(\nu) = s^2(\mu)\leq c :=   \frac{W\left((\NIHamiltonian_\perp)^2+ U^2 \right)W^\dagger}{\HybridizationMass}.
\end{equation}
This allows us to show that $\DMFTmap(C)$ is tight. For $\eta>0$, take $K=[-a,a]$, with $a\geq 1$ large enough so that $c/a^2\leq \eta$. Then for $\nu\in \DMFTmap(C)$,    \eqref{eq:BoundSecondMomentHybridization} holds and
\begin{equation*}
	\nu(\RealNumbers\setminus K) = \int_{\RealNumbers\setminus K}\frac{\eps^2}{\eps^2}d\nu(\eps)\leq \frac{1}{a^2} \int_{\RealNumbers\setminus K}\eps^2d\nu(\eps) \leq \frac{1}{a^2}\int_\RealNumbers \eps^2d\nu(\eps) = \frac{m_2(\nu)}{a^2} \leq \frac{c}{a^2} \leq \eta.
\end{equation*}
Hence $\DMFTmap(C)$ is tight. By Prokhorov's Theorem, it is thus weakly relatively compact.

This concludes the proof of our main result.

\subsection{Proof of Proposition \ref{prop:RegularizationDMFTLoop}}

	Let $\nu^0\in\ProbabilityMeasures$ and $k\in 2\Integers$, and assume that $\nu^0$ has finite $k$-th moment, i.e. $\int_\RealNumbers |\eps|^kd\nu^0(\eps)<\infty$. By Theorem \ref{thm:MomentsAsymptoticExpansion}, the following expansion holds, with $m_k(\nu^0)\geq 0$:
	\begin{equation}
	\label{eq:ExpansionMomentOrderK}
		\int_{\RealNumbers}\frac{d\nu^0(\eps)}{iy-\varepsilon} = \frac{1}{iy}+\dots+\frac{m_k(\nu^0)}{(iy)^{k+1}}+o\left(\frac{1}{y^{k+1}}\right).
	\end{equation}
	Then define $\displaystyle\Hybridization^0(z)=\HybridizationMass \PickFunction{\nu^0}$ and $\xi$ by \eqref{eq:IPTintermediate}:
	\begin{equation*}
		\PickFunction{\xi} = \frac{1}{z-\Hybridization^0(z)}.
	\end{equation*}
	This function can be asymptotically expanded to order $k+3$ using \eqref{eq:ExpansionMomentOrderK}.
	$$
	\begin{array}{rl}		\displaystyle\int_{\RealNumbers} \frac{d\xi(\varepsilon)}{iy-\varepsilon} = & \dfrac{1}{iy-\HybridizationMass\left(\frac{1}{iy}+\dots+\frac{m_k(\nu^0)}{(iy)^{k+1}}+o\left(\frac{1}{y^{k+1}}\right)\right)} \\
		= & \displaystyle \frac{1}{iy} \left( 1-\frac{\HybridizationMass}{iy} \left(\frac{1}{iy}+\dots+\frac{m_k(\nu^0)}{(iy)^{k+1}}+o\left(\frac{1}{(iy)^{k+1}}\right)\right)\right)^{-1} \\
		= & \displaystyle\frac{1}{iy} + \dots + \frac{m_{k+2}(\xi)}{(iy)^{k+3}} + o\left(\frac{1}{y^{k+3}}\right).
	\end{array}
	$$
	By Theorem \ref{thm:MomentsAsymptoticExpansion}, $\xi$ has finite moments up to order $k+2$, which is even, denoted by $m_0(\xi)=1,\dots,m_{k+2}(\xi)$. Now let $\mu: = \IPTmap(\nu^0)$, so that $\mu\in\PositiveFiniteMeasures$ is given by \eqref{eq:IPTintermediate4} in the statement of Proposition \ref{prop:DefinitionIPTmap}. In particular, there exists $C_k \in \RealNumbers_+$ such that
	$$ 
		\begin{array}{rl}
			\displaystyle\int_\RealNumbers |\eps|^{k+2}d\mu(\eps) = &\displaystyle \int_{\RealNumbers^3} |\eps_1+\eps_2+\eps_3|^{k+2} \frac{1+e^{-\StatisticalTemperature(\eps_1+\eps_2+\eps_3)}}{(1+e^{-\StatisticalTemperature\eps_1})(1+e^{-\StatisticalTemperature\eps_2})(1+e^{-\StatisticalTemperature\eps_3})} d\xi(\eps_1)d\xi(\eps_2)d\xi(\eps_3) \\
			\leq &\displaystyle \int_{\RealNumbers^3} |\eps_1+\eps_2+\eps_3|^{k+2} d\xi(\eps_1)d\xi(\eps_2)d\xi(\eps_3) \\
			\leq & \displaystyle \int_{\RealNumbers^3}C_k\sum_{\underset{i_1+i_2+i_3=k+2}{1\leq i_1,i_2,i_3\leq k+2}}  |\eps_1|^{i_1}|\eps_2|^{i_2}|\eps_3|^{i_3} d\xi(\eps_1)d\xi(\eps_2)d\xi(\eps_3) \\
			= &\displaystyle C_k \sum_{\underset{i_1+i_2+i_3=k+2}{1\leq i_1,i_2,i_3\leq k+2}} \int_\RealNumbers  |\eps_1|^{i_1} d\xi(\eps_1) \int_\RealNumbers  |\eps_2|^{i_2} d\xi(\eps_2) \int_\RealNumbers  |\eps_3|^{i_3} d\xi(\eps_3) < \infty,
		\end{array}
	$$
	since for $l\leq k+2$, $\int_\RealNumbers |\eps|^ld\xi(\eps) <\infty$. Thus $\mu$ has finite moments of order less than or equal to $k+2$.
	Finally, let $\nu:=\SCmap(\mu)=\DMFTmap(\nu^0)$. Equations \eqref{eq:HybridizationFormulaSCmap} and \eqref{eq:NevanlinnaMeasureHybridizationSCmap} read 
	\begin{equation*}
		\Hybridization(z)=\HybridizationMass\PickFunction{\nu} = W \left( z - \NIHamiltonian_\perp - \SelfEnergy(z) \right)^{-1}  W^\dagger,
	\end{equation*}
	with $\SelfEnergy(z) =U^2\PickFunction{\mu}$. By Theorem \ref{thm:MomentsAsymptoticExpansion}, we can expand $\int_{\RealNumbers} \frac{d\mu(\varepsilon)}{iy-\varepsilon}$ as $y$ goes to $+\infty$ and get
	\begin{eqnarray*}
		\Hybridization(iy) &=&W\left(iy-\NIHamiltonian_\perp-U^2\left(\frac{m_0(\mu)}{iy}+\dots+\frac{m_{k+2}(\mu)}{(iy)^{k+3}}+o\left(\frac{1}{y^{k+3}}\right)\right)\right)^{-1}W^\dagger \\
		&=& \HybridizationMass \left(\frac{1}{iy}+\dots+\frac{m_{k+4}(\nu)}{(iy)^{k+5}}\right)+o\left(\frac{1}{y^{k+5}}\right).
	\end{eqnarray*}
	This means that 
	\begin{equation*}
		\int_{\RealNumbers} \frac{d\nu(\varepsilon)}{iy-\varepsilon} = \frac{1}{iy}+\dots+\frac{m_{k+4}(\nu)}{(iy)^{k+5}}+o\left(\frac{1}{y^{k+5}}\right),
	\end{equation*}
	which proves, by Theorem \ref{thm:MomentsAsymptoticExpansion}, that $\nu$ has finite moments up to order $k+4$.

\section*{Acknowledgements} This project has received funding from the Simons Targeted Grant Award No. 896630 and from the European Research Council (ERC) under the European Union's Horizon 2020 research and innovation programme (grant agreement EMC2 No 810367). The authors thank Michel Ferrero, David Gontier and Mathias Dus for useful discussions. Part of this work was done during the IPAM program {\it Advancing quantum mechanics with mathematics and statistics}.

\appendix
\appendixpage
\section{Uniqueness theorem for an interpolation problem}

The Nevanlinna-Pick \acrfull{acp}, which we will also call \emph{interpolation problem}, has been studied in the beginning of the 20th century independently by Nevanlinna (\cite{nevanlinna_uber_1919}, 1919) and Pick (\cite{pick_uber_1915}, 1915). The results presented in Section \ref{sec:IntroductionNevanlinnaPickInterpolation} are gathered in the book \cite{walsh_interpolation_1935}. Many other references address this problem, such as \cite{garnett_bounded_2006}, \cite{abrahamse_pick_1979}, \cite{sarason_generalized_1967}. In Section \ref{sec:IntroductionNevanlinnaPickInterpolation}, we set up the problem and then give some results that are important to our analysis. Other important results on the \acrshort{acp} and characterizations of the solutions are detailed in the references. For example, we will not discuss the question of extremal solutions (\cite{nicolau_nevanlinna-pick_2015}, \cite{stray_interpolating_1991}), nor the use of Blaschke products (\cite{walsh_interpolation_1935}, \cite{garnett_bounded_2006}, \cite{nicolau_interpolating_1994}). The approaches of Nevanlinna and Pick are different, we chose here to focus on Nevanlinna's approach.

\subsection{Introduction and some general properties}
\label{sec:IntroductionNevanlinnaPickInterpolation}

The Nevanlinna-Pick \acrfull{acp} can be stated as follows. Let $(z_n)_{n\in I}$ be a sequence of distinct points in the Poincaré upper-half-plane $\UpperHalfPlane$ and let $(w_n)_{n\in I}$ be a sequence in $\overline{\UpperHalfPlane}$. We want to answer the following question.

\begin{question}
\label{question:InterpolationProblemUpperHalfPlane}
Is there an analytic function $f:\UpperHalfPlane\rightarrow\overline{\UpperHalfPlane}$ interpolating the given values at the prescribed points? In other words, we look for a Pick function $f$ such that
\begin{equation}
	\label{eq:InterpolationProblem}
	\forall n \in I,\ f(z_n)=w_n,
\end{equation}
where $I$ is a (at most) countable set.
\end{question}

Without loss of generality, we can assume that $I=\{1,2,\ldots\}$. We denote by $\InterpolationProblem{\UpperHalfPlane}(z_n,w_n)_{I}$ the set of solutions to this problem (we will omit the dependence on $I$ unless when needed). Both sets $\UpperHalfPlane$ and $\overline{\UpperHalfPlane}$ are invariant under the action of the subset $\mathcal{T}$ of affine transformations of $\overline{\UpperHalfPlane}$ given by $\mathcal{T} = \{ \tau : \overline{\UpperHalfPlane} \to \overline{\UpperHalfPlane}, z \mapsto az+b,\ b \in \RealNumbers,\ a > 0 \}$, and we have 

\begin{equation}
\label{eq:AffineTransformationInterpolationProblem}
\forall \tau_1,\tau_2 \in \mathcal{T},\quad \InterpolationProblem{\UpperHalfPlane}(\tau_1(z_n),\tau_2(w_n))= \tau_2 \circ \InterpolationProblem{\UpperHalfPlane}(z_n,w_n) \circ \tau_1^{-1}.
\end{equation}

Moreover, question \ref{question:InterpolationProblemUpperHalfPlane} can equivalently be stated in the unit disc instead of the upper-half-plane:

\begin{question}
\label{question:ProblemInterpolationInTheDisc}
Let $(z_n)$ and $(w_n)$ be sequences in the unit disc $\UnitDisc=\{z\in\ComplexNumbers,|z|<1\}$ and the closed disc $\ClosedUnitDisc$ respectively. Is there an analytic function $f:\UnitDisc\rightarrow\ClosedUnitDisc$ such that $\forall n \in I, \quad f(z_n)=w_n$ ?
\end{question}

The reason for the equivalence between both formulations is simply the upper-half-plane can be mapped to the unit disc through the Cayley transform, which is biholomorphic between these sets. The Cayley transform $\CayleyTransform$ and its reciprocal $\CayleyTransform^{-1}$ are given by:

\begin{equation*}
	\begin{array}{rclcrcl}
		\CayleyTransform :  \ComplexNumbers\backslash\{-i\} & \longrightarrow & \ComplexNumbers\backslash\{1\} & \quad & \CayleyTransform^{-1} : \ComplexNumbers\backslash\{1\} & \longrightarrow & \ComplexNumbers\backslash\{-i\} \\
		 z & \mapsto & \dfrac{z-i}{z+i} & \quad & z & \mapsto & i\dfrac{1+z}{1-z}.
	\end{array}
\end{equation*}

As the Cayley transform $\CayleyTransform$ is biholomorphic from $\UpperHalfPlane$ to $\UnitDisc$ and maps the real line $\RealNumbers$ to the unit circle deprived of the point 1, we have the equivalence of the following statements, with $F : \UpperHalfPlane \to \overline{\UpperHalfPlane}$, $(z_n)\subset \UpperHalfPlane$, $(w_n)\subset \overline{\UpperHalfPlane}$, $f=\CayleyTransform\circ F \circ \CayleyTransform^{-1} : \UnitDisc\to\overline{\UnitDisc}$, $\tilde{z_n}=\CayleyTransform(z_n)\in\UnitDisc$ and $\tilde{w_n	}=\CayleyTransform(w_n)\in\overline{\UnitDisc}$.
\begin{equation*}
	\forall n\geq 1 , \ F(z_n) = w_n \iff \forall n\geq 1 ,\ f(\tilde{z_n}) = \tilde{w_n}.
\end{equation*}

In other words, denoting by $\InterpolationProblem{\UnitDisc}(z_n,w_n)$ the set of solutions to Question \ref{question:ProblemInterpolationInTheDisc}, we have
\begin{equation}
\label{eq:RelationInterpolationProblemsUpperHalfPlaneUnitDisc}
\InterpolationProblem{\UnitDisc}(\CayleyTransform(z_n),\CayleyTransform(w_n))=\CayleyTransform \circ \InterpolationProblem{\UpperHalfPlane}(z_n,w_n)\circ \CayleyTransform^{-1}.
\end{equation}

As a matter of fact, we can answer Question \ref{question:ProblemInterpolationInTheDisc} when the sequence $(w_n)$ is constant and of modulus 1, as a direct consequence of the maximum modulus principle.

\begin{lemma}
Given $C \in \overline{\UnitDisc} \setminus \UnitDisc$, $\InterpolationProblem{\UnitDisc}(z_n,C)= \{ C \}$.
\end{lemma}

Combined with equation \eqref{eq:RelationInterpolationProblemsUpperHalfPlaneUnitDisc}, we have a first statement about the set of solutions $\InterpolationProblem{\UpperHalfPlane}$ to the interpolation problem \eqref{eq:InterpolationProblem}.

\begin{corollary}
\label{cor:UniqueSolutionRealConstant}
Given $C \in \RealNumbers$, $\InterpolationProblem{\UpperHalfPlane}(z_n,C)=\{ C \}$.
\end{corollary}

The next definition introduces the functions $b_a$ (the notation suggests Blaschke products, see \cite{nicolau_nevanlinna-pick_2015} and \cite{garnett_bounded_2006}), which will be useful for our analysis of $\InterpolationProblem{\UnitDisc}(z_n,w_n)$.

\begin{definition}
Let $a \in \UnitDisc$. Define for $z\in\UnitDisc$, $b_{a}(z)=\frac{z-a}{1-\overline{a}z} \in \UnitDisc$.
\end{definition}

These functions are in fact biholomorphic from $\UnitDisc$ to itself and $b_a$ only vanishes in $a\in\UnitDisc$. This zero is of order 1. The function $b_a$ can actually be extended to the closed disc $\ClosedUnitDisc$, and we will also denote this extension $b_a$. Considering our interpolation problem, we notice the following. If $f$ is a solution to the interpolation problem stated in question \ref{question:ProblemInterpolationInTheDisc}, we have $f(z_1)=w_1$, with $|w_1|\leq 1$. If the modulus of $w_1$ is 1, then $f$ is constant because of the maximum modulus principle. Suppose this is not the case and define the function $f^{(1)}$ by

\begin{equation}
\label{eq:IteratedFunction}
f^{(1)}(z) = \frac{b_{w_1}(f(z))}{b_{z_1}(z)}.
\end{equation}

$f^{(1)}$ is then well-defined on $\UnitDisc$, because the only zero of $b_{z_1}$ is $z_1$ and is of order 1, and $z_1$ is also a zero of order at least 1 of $b_{w_1}\circ f$. One can check that $f^{(1)}$ takes values in $\ClosedUnitDisc$, so that $f^{(1)}$ is analytic from $\UnitDisc$ to $\ClosedUnitDisc$ if and only if $f=b_{w_1}^{-1}\circ (b_{z_1}f^{(1)})$ is analytic from $\UnitDisc$ to $\ClosedUnitDisc$.

Now suppose $\Cardinal{I}\geq 2$ and define for all $n \in I \setminus \{ 1 \}$, $w_n^{(1)} := \frac{b_{w_1}(w_n)}{b_{z_1}(z_n)}$. We notice that $f^{(1)}(z_n) = \frac{b_{w_1}(f(z_n))}{b_{z_1}(z_n)} = \frac{b_{w_1}(w_n)}{b_{z_1}(z_n)} = w_n^{(1)}$. This means that  $f^{(1)}$ is the solution to the Nevanlinna-Pick interpolation problem

\begin{equation}
\label{eq:IteratedInterpolationProblem}
g(z_n) =  w_n^{(1)}, \ \forall n\in I \setminus \{ 1 \}.
\end{equation}

We have proven the following lemma, which is the main element of the Schur interpolation algorithm \cite{gull_continuous-time_2011}.

\begin{lemma}[Schur iteration]
\label{lemma:SchurIteration}
Assume  $\Cardinal{I} \geq 2$ and $w_1 \in \UnitDisc$. We then have the following equivalence:
\begin{equation*}
f \in \InterpolationProblem{\UnitDisc}(z_n,w_n)_I \iff f^{(1)}: z \mapsto \frac{b_{w_1}(f(z))}{b_{z_1}(z)} \in \InterpolationProblem{\UnitDisc}(z_n,w_n^{(1)})_{I\setminus \{1 \}}.
\end{equation*}
In other words,
\begin{equation*}
\InterpolationProblem{\UnitDisc}(z_n,w_n)_I=b_{w_1}^{-1}\circ \left( b_{z_1} \cdot \InterpolationProblem{\UnitDisc}(z_n,w_n^{(1)})_{I\setminus \{1 \}} \right).
\end{equation*}
\end{lemma}

Both previous lemmas give the intuition about the following theorem, which can be found in any reference dealing with the issue of Nevanlinna-Pick interpolation, e.g. \cite{walsh_interpolation_1935}, \cite{garnett_bounded_2006}, \cite{nicolau_interpolating_1994}. Refinements of this result and characterizations of the solutions are detailed in the references given at the beginning of this section. 

\begin{theorem}
\label{thm:UnicityNevanlinnaPoints}
Let $(z_n)_{n\geq 1}$ and $(w_n)_{n\geq 1}$ be sequences in the unit disc $\UnitDisc$ and the closed disc $\ClosedUnitDisc$ respectively. Define by induction, for $1\leq l$ and $k>l$,
\begin{equation}
\label{eq:NevanlinnaPoints}
w_k^{(l)} := \dfrac{w_k^{(l-1)}-w_l^{(l-1)}}{1-\overline{w_l^{(l-1)}}w_k^{(l-1)}}\dfrac{1-\overline{z_l}z_k}{z_k-z_l} = \dfrac{b_{w_l^{(l-1)}}(w_k^{(l-1)})}{b_{z_l}(z_k)},
\end{equation}
with $w_k^{(0)}=w_k$ for $k\geq 1$.

There are three distinct cases.
\begin{enumerate}
\item If there exists $k\geq 1$ such that $|w_k^{(k-1)}|>1$, then there is no solution to the interpolation problem~(\ref{eq:InterpolationProblem}).
\item Else, if there exists $K\geq 1$ such that for all $1\leq k < K$, $|w_k^{(k-1)}|<1$, $|w_K^{(K-1)}|=1$ and for all $l\geq K$, $w_l^{(K-1)} = w_K^{(K-1)}$, then there exists a unique solution to the interpolation problem~(\ref{eq:InterpolationProblem}).
\item Else, we have for all $k\geq 1$, $|w_k^{(k-1)}|<1$ and there is either 1 or infinitely many solutions to the interpolation problem.
\end{enumerate}	 
\end{theorem}

\subsection{A uniqueness result for \acrshort{acp}s with a rational solution}

In our mathematical framework, we need the following uniqueness theorem (Theorem \ref{thm:UniquenessNevanlinnaPickInterpolation}).
In this section, we are tackling the uniqueness of the interpolation problem \eqref{eq:InterpolationProblem}, in the case where we have already found one solution of some specific form. We assume that we have found a solution $f$ such that its Nevanlinna-Riesz measure in the integral representation \eqref{eq:NevanlinnaRepresentation} is a finite sum of weighted Dirac measures, that is $f$ is a rational function. This means that we have $f(z)= \alpha z +C+ \sum_{k=1}^K\frac{a_k}{z-\eps_k}$, with $\alpha\geq 0$, $C\in\RealNumbers$, $K\in\Integers$, the $a_k$'s are negative numbers and the $\eps_k$'s are distinct real numbers. The following theorem states that it is then the only solution to the interpolation problem \eqref{eq:InterpolationProblem}.

\begin{theorem}
\label{thm:UniquenessNevanlinnaPickInterpolation} Let $f: \UpperHalfPlane \to \overline{\UpperHalfPlane}$ be a Pick function such that its Nevanlinna-Riesz measure $\mu$ is a finite sum of Dirac measures: there exist $K \in \Integers$, $a_1,\dots,a_K < 0$ and distinct real numbers $\eps_1,\dots,\eps_K$ such that 
\begin{equation*}
\mu=\sum_{k=1}^K a_k \delta_{\eps_k}.
\end{equation*} 
Then, if $f$ is a solution to an interpolation problem $\InterpolationProblem{\ComplexNumbers}(z_n,w_n)_I$ with $\Cardinal{I} \geq K+2$, this problem has no other solution.
\end{theorem}

\begin{proof}
We prove this result by strong induction on $K$, for $K \in \Integers$. If $K=0$, the expression of $f$ reads $f(z)=\alpha z+C$, with $\alpha\geq 0$ and $C\in\RealNumbers$. Now suppose $f \in \InterpolationProblem{\UpperHalfPlane}(z_n,w_n)$. Then for all $n \in I$, we have $\alpha z_n +C= w_n$, so that 
\begin{equation*}
\InterpolationProblem{\UpperHalfPlane}(z_n,w_n)=\InterpolationProblem{\UpperHalfPlane}(z_n,\alpha z_n+C).
\end{equation*}

If $\alpha=0$, then $\InterpolationProblem{\UpperHalfPlane}(z_n,C)=\{ C \}$ by corollary \ref{cor:UniqueSolutionRealConstant}. Else, we extend continuously $f$ to $\overline{\UpperHalfPlane}$ into the affine transformation $\widetilde{f} \in \mathcal{T}$ and we have
\begin{equation*}
\InterpolationProblem{\UpperHalfPlane}(z_n,w_n)= \widetilde{f} \circ \CayleyTransform^{-1} \circ \InterpolationProblem{\UnitDisc}(\CayleyTransform(z_n),\CayleyTransform(z_n)) \circ \CayleyTransform.
\end{equation*}

Denoting by $\widetilde{z}_n=\widetilde{w}_n=\CayleyTransform(z_n)$, we find that for all $n \in I \setminus \{1\}, \widetilde{w}^{(1)}_n=1$. Hence by Theorem \ref{thm:UnicityNevanlinnaPoints}, $\InterpolationProblem{\UnitDisc}(\CayleyTransform(z_n),\CayleyTransform(z_n))= \{ z\mapsto z\}$ and $\InterpolationProblem{\UpperHalfPlane}(z_n,w_n)=\{ f \}$.

Now assume the result holds for $l \leq K$ and take $f$ such that its Nevanlinna-Riesz measure is a sum of $K+1$ Dirac measures. Now consider an \acrshort{acp} $\InterpolationProblem{\ComplexNumbers}(z_n,w_n)_I$  with $\Cardinal{I}\geq K+3$ to which $f$ is a solution, namely $f(z_n)=w_n$ and there exist $\alpha\geq 0$, $C\in\RealNumbers$, $a_1,\dots,a_{K+1} < 0$ and real numbers $\eps_1<\dots<\eps_{K+1}$ , such that for all $z \in \UpperHalfPlane$,
\begin{equation*}
f(z)= \alpha z + C + \sum_{k=1}^{K+1}\frac{a_k}{z-\eps_k}.
\end{equation*}

We start by making use of the property  \eqref{eq:AffineTransformationInterpolationProblem} to simplify the problem. It is possible to chose two affine transformations $\tau_{z_1}$ and $\check{\tau}_{z_1}$ in $\mathcal{T}$ such that, setting $\widetilde{f}=\check{\tau}_{z_1} \circ f\circ \tau_{z_1}^{-1}$, we have
\begin{equation*}
    \widetilde{f}(z)= \widetilde{C} + \widetilde{\alpha} z + \sum_{k=1}^{K+1}\frac{\widetilde{a}_k}{z-\widetilde{\eps}_k},
\end{equation*}
where the parameters $\widetilde{C}$, $\widetilde{\alpha}$,  $\widetilde{a}_k$, and $\widetilde{\eps}_k$ satisfy the same properties as their counterparts without the tildes, and such that 
\begin{equation}
\label{eq:AssumptionsOnF}
\tau_{z_1}(z_1)=i, \quad \Re(\widetilde{f}(\widetilde{z}_1))=0 \quad \text{ and }\quad \sum_{k=1}^{K+1}\frac{-\widetilde{a}_k}{1+\widetilde{\eps}_k^2}=1.
\end{equation} 
Using \eqref{eq:AffineTransformationInterpolationProblem}, we have
\begin{equation*}
\InterpolationProblem{\UpperHalfPlane}(z_n,w_n)= \check{\tau}_{z_1}^{-1}\circ \InterpolationProblem{\UpperHalfPlane}(\underbrace{\tau_{z_1}(z_n)}_{\widetilde{z}_n},\underbrace{(\check{\tau}_{z_1} \circ f\circ \tau_{z_1}^{-1})}_{\widetilde{f}}(\tau_{z_1}(z_n)))\circ \tau_{z_1}.
\end{equation*}

For the remaining of the proof, we will then focus on $\InterpolationProblem{\UpperHalfPlane}(\widetilde{z}_n,\widetilde{f}(\widetilde{z}_n))$. We will omit the tildes for the sake of simplicity, and assume that $z_1=i$ and that \eqref{eq:AssumptionsOnF} holds. With that being said, notice that $\Im(f(z_1))=\Im(f(i))=\alpha + \sum_{k=1}^{K+1} \frac{-a_k}{1+\eps_k^2} > 0$, hence $\CayleyTransform(f(z_1)) \in \UnitDisc$. Therefore, using Lemma~\ref{lemma:SchurIteration} equation~\eqref{eq:RelationInterpolationProblemsUpperHalfPlaneUnitDisc}, we have
\begin{equation*}
\InterpolationProblem{\UnitDisc}(\CayleyTransform(z_n),\CayleyTransform(f(z_n)))_I=b_{\CayleyTransform(f(z_1))}^{-1}\circ \left( b_{z_1} \cdot \InterpolationProblem{\UnitDisc}(\CayleyTransform(z_n),\CayleyTransform(f(z_n))^{(1)})_{I\setminus \{1 \}} \right).
\end{equation*}
and
\begin{equation*}
\InterpolationProblem{\UnitDisc}(\CayleyTransform(z_n),\CayleyTransform(f(z_n))^{(1)})_{I\setminus \{1 \}}= \CayleyTransform \circ \InterpolationProblem{\UpperHalfPlane}\left(z_n,\CayleyTransform^{-1}([\CayleyTransform(f(z_n))]^{(1)})\right)_{I\setminus \{1 \}} \circ \CayleyTransform^{-1}.
\end{equation*}

We now compute $g(z) = \CayleyTransform^{-1}\circ [\CayleyTransform\circ f]^{(1)}(z)$, for $z \in \UpperHalfPlane \setminus \{i\}$. Since $\CayleyTransform(z_1) = \CayleyTransform(i)=0$ and $z\neq i$, we have $b_{\CayleyTransform(z_1)}(\CayleyTransform(z))=\CayleyTransform(z)\neq 0$ and $g(z)$ is well defined. The computation reads
\begin{equation*}
    g(z) = i\frac{(f(z_1)+i)\frac{\overline{f(z_1)}-f(z)}{z+i}+\overline{(f(z_1)+i)}\frac{f(z)-f(z_1)}{z-i}}{(f(z_1)+i)\frac{\overline{f(z_1)}-f(z)}{z+i}-\overline{(f(z_1)+i)}\frac{f(z)-f(z_1)}{z-i}}.
\end{equation*}
We used the fact that since $\overline{\CayleyTransform(f(z_1))}\CayleyTransform(f(z)) \in \UnitDisc$, it is not equal to 1 and that~$\frac{1}{2i}(f(z)+i)\vert(f(z_1)+i)\vert^2 \neq 0$.
Now, realize that, since $z_1=i$, we have
\begin{equation*}
\frac{f(z)-f(z_1)}{z-i}= \alpha + \sum_{k=1}^{K+1}\frac{1}{\eps_k-i}\frac{a_k}{z-\eps_k},
\end{equation*}
and similarly
\begin{equation*}
\frac{\overline{f(z_1)}-f(z)}{z+i} = -\alpha - \sum_{k=1}^{K+1}\frac{1}{\eps_k+i}\frac{a_k}{z-\eps_k},
\end{equation*}
so that
\begin{equation}\label{eq:denominatorNumeratorOfg}
g(z)=-\frac{\alpha(\Im(f(z_1))+1) + \sum_{k=1}^{K+1}\Im\left(\frac{f(z_1)+i}{\eps_k+i}\right)\frac{a_k}{z-\eps_k}}{\alpha\Re(f(z_1)) + \sum_{k=1}^{K+1}\Re\left(\frac{f(z_1)+i}{\eps_k+i}\right)\frac{a_k}{z-\eps_k}},
\end{equation}
which holds for $z=i$ as well. As we set $\Re(f(z_1))=0$, we have
\begin{equation*}
\Re\left(\frac{f(z_1)+i}{\eps_k+i}\right)=\frac{\Im(f(z_1))+1}{1+\eps_k^2} > 0 \quad \text{ and } \quad \Im\left(\frac{f(z_1)+i}{\eps_k+i}\right)=\eps_k\frac{\Im(f(z_1))+1}{1+\eps_k^2}.
\end{equation*}

Multiplying the numerator and denominator in \eqref{eq:denominatorNumeratorOfg} by $\frac{1}{\Im(f(z_1))+1}\prod_{k=1}^{K+1}(z-\eps_k)$, we end up with $\displaystyle g(z)=\frac{P(z)}{Q(z)}$, with
\begin{equation*}
P(z):= \alpha\prod_{k=1}^{K+1}(z-\eps_k) + \sum_{k=1}^{K+1}\eps_k \frac{a_k}{1+\eps_k^2} \prod_{l=1,l\neq k}^{K+1}(z-\eps_l),
\end{equation*}
and
\begin{equation*}
Q(z):= \sum_{k=1}^{K+1}\frac{-a_k}{1+\eps_k^2} \prod_{l=1,l \neq k}^{K+1} (z-\eps_l).
\end{equation*}

We have $Q(\eps_k)=\frac{-a_k}{1+\eps_k^2}\prod_{l=1,l \neq k}^{K+1}(\eps_k-\eps_l) \neq 0$ by hypothesis on the $\eps_k$'s, so that 
\begin{equation*}
Q(z)=0 \iff \sum_{k=1}^{K+1}\frac{\frac{-a_k}{1+\eps_k^2}}{z-\eps_k}=0,
\end{equation*}
hence $Q$ admits exactly $K$ distinct roots $\eps'_k \in (\eps_k,\eps_{k+1}) \subset \RealNumbers$ by the intermediate value theorem (as in the proof of Lemma \ref{lem:IncreaseNumberOfPoles}), and it is unitary due to \eqref{eq:AssumptionsOnF}. The partial fraction decomposition of $g$ finally reads
\begin{equation*}
g(z)=\alpha' z + C'+ \sum_{k=1}^K\frac{a'_k}{z-\eps'_k},
\end{equation*}
where
$$
\begin{array}{rcl}
\alpha'&=&\alpha, \\
C'&=&\displaystyle \underset{ y \to \infty}{\lim} g(iy)-i\alpha'y=(1-\alpha)\sum_{k=1}^{L+1}\eps_k\frac{a_k}{1+\eps_k^2} \in \RealNumbers, \\
a'_k&=& \underset{y \to 0^+}{\lim} iyg(\eps'_k+iy) \\
& =& \left(\prod_{l=1,l\neq k}^K \underbrace{\left(\frac{\eps'_k-\eps_l}{\eps'_k-\eps'_l}\right)}_{>0} \right)\underbrace{(\eps'_k-\eps_k)(\eps'_k-\eps_{L+1})}_{<0}\underbrace{(\alpha + 1)}_{> 0} < 0.
\end{array}
$$

This computation shows that the Nevanlinna-Riesz measure  of $g$ is a sum of $K$ Dirac measures as described in the statement of the theorem. It is a solution to the \acrlong{acp} $$\InterpolationProblem{\UpperHalfPlane}\left(z_n,\CayleyTransform^{-1}([\CayleyTransform(f(z_n))]^{(1)}])\right)_{I\setminus \{1 \}},$$ where $\Cardinal{I\setminus \{1\}} \geq  K+2$ by assumption. By the induction hypothesis, $g$ is the only solution to this \acrshort{acp}, and therefore $f$ is the only solution to the \acrshort{acp} $\InterpolationProblem{\UpperHalfPlane}(w_n,z_n)_I$.
\end{proof}

\section{Paramagnetic \acrshort{ipt}-\acrshort{dmft} equations}
\label{appendix:SpinIndependenceIPTDMFT}
In this appendix, we detail the spin independence of the paramagnetic \acrshort{ipt}-\acrshort{dmft} equations.

As detailed in Remark \ref{rmk:AIMCommutationRelations}, the Hamiltonian $\AIHamiltonian$ commutes with the total spin operator $\AISpinOperator$. More precisely, $\OneParticleSpace[AI]$ is $\AISpinOperator$-invariant, and in the decomposition 
\begin{align}
\OneParticleSpace[AI]=\OneParticleSpace[\uparrow]\oplus\OneParticleSpace[\downarrow], \quad \OneParticleSpace[\sigma]=\Span \left( \VacuumState \otimes \cdots \otimes \VacuumState \otimes \underset{m\text{-th}}{\vert \sigma \rangle} \otimes \VacuumState \otimes \cdots \otimes \VacuumState , m \in \HubbardVertices\sqcup \IntSubSet{1}{\BathDimension}\right)
\end{align}
the total spin operators reads $\AISpinOperator=\Identity_{\uparrow}\oplus(-\Identity_{\downarrow})$:
\begin{equation*}
\AISpinOperator=\left( \begin{matrix}
\Identity & 0 \\
0&-\Identity
\end{matrix} \right).
\end{equation*}
Writing in this decomposition $\NIHamiltonian=\NIHamiltonian_{\uparrow}\oplus\NIHamiltonian_{\downarrow}$, we then have
\begin{equation*}
\GreensFunction^0(z)=(z-\NIHamiltonian_\uparrow)^{-1}\oplus(z-\NIHamiltonian_\downarrow)^{-1}.
\end{equation*}
Since no magnetic field is included in the model, $\NIHamiltonian_\uparrow$ and $\NIHamiltonian_\downarrow$ act the same way on their respective domains: denoting by $\FlipOperator \in \LinearOperator(\OneParticleSpace[AI])$ the \emph{spin flip} isomorphism defined by linearity with $\forall m \in \HubbardVertices \sqcup \IntSubSet{1}{\BathDimension},\forall \sigma \in \{\uparrow,\downarrow\}$
\begin{equation*}
\FlipOperator\left(\VacuumState \otimes \cdots \otimes \VacuumState \otimes \underset{m\text{-th}}{\vert \sigma \rangle} \otimes \VacuumState \otimes \cdots \otimes \VacuumState \right)= \VacuumState \otimes \cdots \otimes \VacuumState \otimes \underset{m\text{-th}}{\vert \bar{\sigma} \rangle} \otimes \VacuumState \otimes \cdots \otimes \VacuumState,
\end{equation*}
we have 
\begin{equation*}
\NIHamiltonian_{\uparrow}=\FlipOperator_{\uparrow,\downarrow} \NIHamiltonian_{\downarrow} \FlipOperator_{\downarrow,\uparrow}.
\end{equation*}
Now in the impurity-spin decomposition 
\begin{align}
\OneParticleSpace&=\OneParticleSpace[\uparrow,\mathrm{imp}]\oplus\OneParticleSpace[\uparrow,\mathrm{bath}]\oplus\OneParticleSpace[\downarrow,\mathrm{imp}]\oplus\OneParticleSpace[\downarrow,\mathrm{bath}] \\
\OneParticleSpace[\sigma,imp]&= \Span \left( \VacuumState \otimes \cdots \otimes \VacuumState \otimes \underset{i\text{-th}}{\vert \sigma \rangle} \otimes \VacuumState \otimes \cdots \otimes \VacuumState , i \in \HubbardVertices \right) \\
\OneParticleSpace[\sigma,bath]&= \Span \left( \VacuumState \otimes \cdots \otimes \VacuumState \otimes \underset{k\text{-th}}{\vert \sigma \rangle} \otimes \VacuumState \otimes \cdots \otimes \VacuumState , k \in \IntSubSet{1}{\BathDimension} \right)
\end{align}
we have $\AISpinOperator=\Identity\oplus \Identity \oplus (-\Identity)\oplus (-\Identity)$ and $\ImpNumberOperator=\Identity\oplus 0 \oplus \Identity \oplus 0$ :
\begin{equation*}
\AISpinOperator=\left( \begin{matrix}
\Identity &0&0&0 \\
0&\Identity &0&0 \\
0&0&-\Identity &0 \\
0&0&0&-\Identity
\end{matrix} \right)
,\quad \ImpNumberOperator=\left( \begin{matrix}
\Identity &0&0&0\\
0&0&0&0 \\
0&0&\Identity&0\\
0&0&0&0
\end{matrix}\right)
\end{equation*}
and we have for the impurity orthogonal restriction of the non-interacting Green's function $\GreensFunction_\mathrm{imp}^0$:
\begin{equation*}
\GreensFunction_\mathrm{imp}^0(z)=\left(z-\NIHamiltonian_\mathrm{\uparrow,\mathrm{imp}}-\Hybridization_\uparrow(z)\right) \oplus \left(z-\NIHamiltonian_\mathrm{\downarrow,\mathrm{imp}}-\Hybridization_\downarrow(z)\right)
\end{equation*}
Indeed, as it is the case for the non-interacting Hamiltonian and the non-interacting Green's function, $\Hybridization_{\uparrow}$ acts similarly as $\Hybridization_{\downarrow}$ on their respective domain, and are related by conjugation with the spin flip operator
\begin{equation*}
\Hybridization_{\uparrow}=\FlipOperator_{\uparrow,\downarrow}\Hybridization_{\downarrow}\FlipOperator_{\downarrow,\uparrow}
\end{equation*}

With that being said, in the \acrshort{ipt} approximation outlined in this document, we have for $ n \in \RelativeIntegers$,
\begin{equation*}
\MatsubaraSelfEnergyFourierSeries_{\mathrm{imp},n}=\MatsubaraSelfEnergyFourierSeries_{\uparrow,\mathrm{imp},n} \oplus \MatsubaraSelfEnergyFourierSeries_{\downarrow,\mathrm{imp},n},
\end{equation*}
with for all $\sigma \in \{\uparrow,\downarrow \},$ 
\begin{equation*}
\MatsubaraSelfEnergyFourierSeries_{\sigma,\mathrm{imp},n}=\OnSiteRepulsion[]^2\int_{0}^\StatisticalTemperature e^{i\MatsubaraFrequency_n \tau} \left(\frac{1}{\StatisticalTemperature}\sum_{n' \in \RelativeIntegers}\left(i\MatsubaraFrequency_{n'} - \Hybridization_{\sigma}(i\MatsubaraFrequency_{n'})^{-1} \right)\right)^3 d \tau .
\end{equation*}
And indeed, we have the conjugation relation 
\begin{equation*}
\MatsubaraSelfEnergyFourierSeries_{\uparrow,\mathrm{imp},n} =\FlipOperator_{\uparrow,\downarrow}\MatsubaraSelfEnergyFourierSeries_{\downarrow,\mathrm{imp},n}\FlipOperator_{\downarrow,\uparrow}.
\end{equation*}

These last equations show that the impurity-spin orthogonal restrictions of the Matsubara self-energy Fourier coefficients are copies one of another. Since we assume that the partition $\DMFTPartition$ is in singletons, it follows that $\Cardinal{\HubbardVertices}=1$ and $\dim(\OneParticleSpace[\sigma,\mathrm{imp}])=1$, so that using the bases $\OneParticleBasis[\sigma,\mathrm{imp}]$ given by
\begin{equation*}
\forall \sigma \in \{\uparrow,\downarrow\}, \quad \OneParticleBasis[\sigma,\mathrm{imp}]=\{ \vert \sigma \rangle \otimes \VacuumState \otimes \cdots \otimes \VacuumState \},
\end{equation*}
we have for all $n \in \RelativeIntegers$,
\begin{equation*}
 \quad \Matrix_{\OneParticleBasis[\uparrow,\mathrm{imp}]}(\MatsubaraSelfEnergyFourierSeries_{\uparrow,\mathrm{imp},n})=\Matrix_{\OneParticleBasis[\downarrow,\mathrm{imp}]}(\MatsubaraSelfEnergyFourierSeries_{\downarrow,\mathrm{imp},n}) := \SelfEnergy_n.
\end{equation*}
Finally, we have for $ z \in \UpperHalfPlane$,
\begin{equation*}
 \Matrix_{\OneParticleBasis[\uparrow,\mathrm{imp}]}(\Hybridization_{\uparrow}(z))=\Matrix_{\OneParticleBasis[\downarrow,\mathrm{imp}]}(\Hybridization_\downarrow(z))=\sum_{k=1}^\BathDimension \frac{\BathCoupling[k] \BathCoupling[k]^\dagger}{z-\BathEnergy[k]},
\end{equation*}
so that we indifferently refer to $\Hybridization(z)$ by abuse of notation.

\clearpage

\bibliography{citations}

\begin{thebibliography}{10}

\bibitem{abrahamse_pick_1979}
Marine~B. Abrahamse.
\newblock The {Pick} interpolation theorem for finitely connected domains.
\newblock {\em Michigan Mathematical Journal}, 26(2):195--203, 1979.

\bibitem{akhiezer_classical_2020}
Naum~I. Akhiezer.
\newblock {\em The {Classical} {Moment} {Problem} and {Some} {Related}
  {Questions} in {Analysis}}.
\newblock Society for Industrial and Applied Mathematics, Philadelphia, PA,
  2020.

\bibitem{ambrosio_gradient_2005}
Luigi Ambrosio, Nicola Gigli, and Giuseppe Savaré.
\newblock {\em Gradient flows: in metric spaces and in the space of probability
  measures}.
\newblock Springer Science \& Business Media, 2005.

\bibitem{anderson_localized_1961}
Philip~W. Anderson.
\newblock Localized {Magnetic} {States} in {Metals}.
\newblock {\em Physical Review}, 124(1):41--53, 1961.

\bibitem{arsenault_benchmark_2012}
Louis-François Arsenault, Patrick Sémon, and André-Marie~S. Tremblay.
\newblock Benchmark of a modified iterated perturbation theory approach on the
  fcc lattice at strong coupling.
\newblock {\em Physical Review B}, 86(8):085133, 2012.

\bibitem{bach_generalized_1994}
Volker Bach, Elliott~H. Lieb, and Jan~Philip Solovej.
\newblock Generalized {Hartree}-{Fock} theory and the {Hubbard} model.
\newblock {\em Journal of Statistical Physics}, 76(1):3--89, 1994.

\bibitem{baym_determination_1961}
Gordon Baym and N.~David Mermin.
\newblock Determination of {Thermodynamic} {Green}'s {Functions}.
\newblock {\em Journal of Mathematical Physics}, 2(2):232--234, 1961.

\bibitem{bednorz_possible_1986}
Johannes~G. Bednorz and Karl~A. Müller.
\newblock Possible {highTc} superconductivity in the {Ba}-{La}-{Cu}-{O} system.
\newblock {\em Zeitschrift für Physik B Condensed Matter}, 64(2):189--193,
  1986.

\bibitem{bogacev_measure_2007}
Vladimir~I. Bogačev.
\newblock {\em Measure theory}.
\newblock Springer, Berlin, 2007.

\bibitem{bogacev_weak_2018}
Vladimir~I. Bogačev.
\newblock {\em Weak convergence of measures}.
\newblock Number volume 234 in Mathematical surveys and monographs. American
  Mathematical Society, Providence, Rhode Island, 2018.

\bibitem{bratteli_operator_1987}
Ola Bratteli and Derek~William Robinson.
\newblock {\em Operator algebras and quantum statistical mechanics 1. {C}*-and
  {W}*-{Algebras}. {Symmetry} {Groups}. {Decomposition} of {States}}.
\newblock Theoretical and {Mathematical} {Physics}. Springer Berlin,
  Heidelberg, 2 edition, 1987.

\bibitem{bratteli_operator_1997}
Ola Bratteli and Derek~William Robinson.
\newblock {\em Operator algebras and quantum statistical mechanics 2.
  {Equilibrium} {States}. {Models} in {Quantum} {Statistical} {Mechanics}}.
\newblock Theoretical and {Mathematical} {Physics}. Springer Berlin Heidelberg,
  2 edition, 1997.

\bibitem{cances_mathematical_2023}
Eric Cancès, Fabian~M. Faulstich, Alfred Kirsch, Eloïse Letournel, and
  Antoine Levitt.
\newblock Some mathematical insights on {Density} {Matrix} {Embedding}
  {Theory}, 2023.
\newblock arXiv:2305.16472 [cond-mat, physics:math-ph].

\bibitem{cances_mathematical_2016}
Eric Cancès, David Gontier, and Gabriel Stoltz.
\newblock A mathematical analysis of the {GW0} method for computing electronic
  excited energies of molecules.
\newblock {\em Reviews in Mathematical Physics}, 28(04):1650008, 2016.

\bibitem{eckhardt_continued_2022}
Jonathan Eckhardt.
\newblock Continued fraction expansions of {Herglotz} {Nevanlinna} functions
  and generalized indefinite strings of {Stieltjes} type.
\newblock {\em Bulletin of the London Mathematical Society}, 54(2):737--759,
  2022.

\bibitem{faulstich_pure_2022}
Fabian~M. Faulstich, Raehyun Kim, Zhi-Hao Cui, Zaiwen Wen, Garnet Kin-Lic~Chan,
  and Lin Lin.
\newblock Pure {State} v-{Representability} of {Density} {Matrix} {Embedding}
  {Theory}.
\newblock {\em Journal of Chemical Theory and Computation}, 18(2):851--864,
  2022.

\bibitem{fei_analytical_2021}
Jiani Fei, Chia-Nan Yeh, Dominika Zgid, and Emanuel Gull.
\newblock Analytical continuation of matrix-valued functions: {Carathéodory}
  formalism.
\newblock {\em Physical Review B}, 104(16):165111, 2021.

\bibitem{fertitta_energy-weighted_2019}
Edoardo Fertitta and George~H. Booth.
\newblock Energy-weighted density matrix embedding of open correlated chemical
  fragments.
\newblock {\em The Journal of Chemical Physics}, 151(1):014115, 2019.

\bibitem{feynman_theory_1963}
Richard~P Feynman and Frank~L Vernon.
\newblock The theory of a general quantum system interacting with a linear
  dissipative system.
\newblock {\em Annals of Physics}, 24:118--173, 1963.

\bibitem{garnett_bounded_2006}
John Garnett.
\newblock {\em Bounded {Analytic} {Functions}}.
\newblock Springer Science \& Business Media, 2006.

\bibitem{georges_fermions_2018}
Antoine Georges.
\newblock Fermions en interaction : {Introduction} à la théorie du champ
  moyen dynamique, 2018.

\bibitem{georges_dynamical_1996}
Antoine Georges, Gabriel Kotliar, Werner Krauth, and Marcelo~J. Rozenberg.
\newblock Dynamical mean-field theory of strongly correlated fermion systems
  and the limit of infinite dimensions.
\newblock {\em Reviews of Modern Physics}, 68(1):13--125, 1996.

\bibitem{gesztesy_matrixvalued_2000}
Fritz Gesztesy and Eduard Tsekanovskii.
\newblock On {Matrix}–{Valued} {Herglotz} {Functions}.
\newblock {\em Mathematische Nachrichten}, 218(1):61--138, 2000.

\bibitem{gontier_contributions_2015}
David Gontier.
\newblock {\em Contributions mathématiques aux calculs de structures
  électroniques}.
\newblock Thèse de doctorat, Paris Est, 2015.

\bibitem{gull_continuous-time_2011}
Emanuel Gull, Andrew~J. Millis, Alexander~I. Lichtenstein, Alexey~N. Rubtsov,
  Matthias Troyer, and Philipp Werner.
\newblock Continuous-time {Monte} {Carlo} methods for quantum impurity models.
\newblock {\em Reviews of Modern Physics}, 83(2):349--404, 2011.

\bibitem{gutzwiller_effect_1963}
Martin~C. Gutzwiller.
\newblock Effect of {Correlation} on the {Ferromagnetism} of {Transition}
  {Metals}.
\newblock {\em Physical Review Letters}, 10(5):159--162, 1963.

\bibitem{held_mott-hubbard_2001}
Karsten Held, Georg Keller, Volker Eyert, Dieter Vollhardt, and Vladimir~I
  Anisimov.
\newblock Mott-{Hubbard} {Metal}-{Insulator} {Transition} in {Paramagnetic}
  {V}$_2${O}$_3$: {An} {LDA}+{DMFT} ({QMC}) {Study}.
\newblock {\em Physical review letters}, 86(23):5345, 2001.

\bibitem{helffer_equation_2005}
Bernard Helffer and Johannes Sjöstrand.
\newblock Equation de {Schrödinger} avec champ magnétique et équation de
  {Harper}.
\newblock In {\em Schrödinger {Operators}: {Proceedings} of the {Nordic}
  {Summer} {School} in {Mathematics} {Held} at {Sandbjerg} {Slot},
  {Sønderborg}, {Denmark}, {August} 1–12, 1988}, pages 118--197. Springer,
  2005.

\bibitem{hubbard_electron_1963}
John Hubbard and Brian~H. Flowers.
\newblock Electron correlations in narrow energy bands.
\newblock {\em Proceedings of the Royal Society of London. Series A.
  Mathematical and Physical Sciences}, 276(1365):238--257, 1963.

\bibitem{kajueter_new_1996}
Henrik Kajueter and Gabriel Kotliar.
\newblock New {Iterative} {Perturbation} {Scheme} for {Lattice} {Models} with
  {Arbitrary} {Filling}.
\newblock {\em Physical Review Letters}, 77(1):131--134, 1996.

\bibitem{kanamori_electron_1963}
Junjiro Kanamori.
\newblock Electron {Correlation} and {Ferromagnetism} of {Transition} {Metals}.
\newblock {\em Progress of Theoretical Physics}, 30(3):275--289, 1963.

\bibitem{knizia_density_2012}
Gerald Knizia and Garnet Kin-Lic Chan.
\newblock Density {Matrix} {Embedding}: {A} {Simple} {Alternative} to
  {Dynamical} {Mean}-{Field} {Theory}.
\newblock {\em Physical Review Letters}, 109(18):186404, 2012.

\bibitem{kondo_resistance_1964}
Jun Kondo.
\newblock Resistance {Minimum} in {Dilute} {Magnetic} {Alloys}.
\newblock {\em Progress of Theoretical Physics}, 32(1):37--49, 1964.

\bibitem{kallen_definition_1952}
Gunnar Källén.
\newblock On the {Definition} of the {Renormalization} {Constants} in {Quantum}
  {Electrodynamics}.
\newblock {\em Helvetica Physica Acta}, 1952.

\bibitem{lechermann_rotationally_2007}
Frank Lechermann, Antoine Georges, Gabriel Kotliar, and Olivier Parcollet.
\newblock Rotationally invariant slave-boson formalism and momentum dependence
  of the quasiparticle weight.
\newblock {\em Physical Review B}, 76(15):155102, 2007.

\bibitem{lehmann_uber_1954}
Harry Lehmann.
\newblock Über {Eigenschaften} von {Ausbreitungsfunktionen} und
  {Renormierungskonstanten} quantisierter {Felder}.
\newblock {\em Il Nuovo Cimento}, 11(4):16, 1954.

\bibitem{lieb_hubbard_2004}
Elliott~H. Lieb.
\newblock The {Hubbard} model: {Some} {Rigorous} {Results} and {Open}
  {Problems}.
\newblock In Bruno Nachtergaele, Jan~Philip Solovej, and Jakob Yngvason,
  editors, {\em Condensed {Matter} {Physics} and {Exactly} {Soluble} {Models}:
  {Selecta} of {Elliott} {H}. {Lieb}}, pages 59--77. Springer, Berlin,
  Heidelberg, 2004.

\bibitem{lieb_absence_1968}
Elliott~H. Lieb and Fa-Yueh Wu.
\newblock Absence of {Mott} {Transition} in an {Exact} {Solution} of the
  {Short}-{Range}, {One}-{Band} {Model} in {One} {Dimension}.
\newblock {\em Physical Review Letters}, 20(25):1445--1448, 1968.

\bibitem{lieb_one-dimensional_2003}
Elliott~H. Lieb and Fa-Yueh Wu.
\newblock The one-dimensional {Hubbard} model: a reminiscence.
\newblock {\em Physica A: Statistical Mechanics and its Applications},
  321(1-2):1--27, 2003.

\bibitem{lin_sparsity_2020}
Lin Lin and Michael Lindsey.
\newblock Sparsity {Pattern} of the {Self}-energy for {Classical} and {Quantum}
  {Impurity} {Problems}.
\newblock {\em Annales Henri Poincaré}, 21(7):2219--2257, 2020.

\bibitem{lindsey_quantum_2019}
Michael Lindsey.
\newblock {\em The {Quantum} {Many}-{Body} {Problem}: {Methods} and
  {Analysis}}.
\newblock PhD thesis, University of California, Berkeley, 2019.

\bibitem{ma_quantum_2021}
He~Ma, Nan Sheng, Marco Govoni, and Giulia Galli.
\newblock Quantum {Embedding} {Theory} for {Strongly} {Correlated} {States} in
  {Materials}.
\newblock {\em Journal of Chemical Theory and Computation}, 17(4):2116--2125,
  2021.

\bibitem{martin_interacting_2016}
Richard~M. Martin, Lucia Reining, and David~M. Ceperley.
\newblock {\em Interacting {Electrons}: {Theory} and {Computational}
  {Approaches}}.
\newblock Cambridge University Press, Cambridge, 2016.

\bibitem{metzner_correlated_1989}
Walter Metzner and Dieter Vollhardt.
\newblock Correlated {Lattice} {Fermions} in d=$\infty$ {Dimensions}.
\newblock {\em Physical Review Letters}, 62(3):324--327, 1989.

\bibitem{nevanlinna_uber_1919}
Rolf Nevanlinna.
\newblock Uber beschrankte {Funktionen}, die in gegeben {Punkten}
  vorgeschrieben {Werte} annehmen.
\newblock {\em Ann. Acad. Sci. Fenn. Ser. A 1 Mat. Dissertationes}, 1919.

\bibitem{nicolau_interpolating_1994}
Artur Nicolau.
\newblock Interpolating {Blaschke} products solving {Pick}-{Nevanlinna}
  problems.
\newblock {\em Journal d'Analyse Mathematique}, 62(1):199--224, 1994.

\bibitem{nicolau_nevanlinna-pick_2015}
Artur Nicolau.
\newblock The {Nevanlinna}-{Pick} {Interpolation} {Problem}.
\newblock {\em Universitat Autonoma de Barcelona}, 2015.

\bibitem{pariser_semiempirical_1953}
Rudolph Pariser and Robert~G. Parr.
\newblock A {Semi}‐{Empirical} {Theory} of the {Electronic} {Spectra} and
  {Electronic} {Structure} of {Complex} {Unsaturated} {Molecules}. {II}.
\newblock {\em The Journal of Chemical Physics}, 21(5):767--776, 1953.

\bibitem{pick_uber_1915}
Georg Pick.
\newblock Über die {Beschränkungen} analytischer {Funktionen}, welche durch
  vorgegebene {Funktionswerte} bewirkt werden.
\newblock {\em Mathematische Annalen}, 77(1):7--23, 1915.

\bibitem{pople_electron_1953}
John~A. Pople.
\newblock Electron interaction in unsaturated hydrocarbons.
\newblock {\em Transactions of the Faraday Society}, 49(0):1375--1385, 1953.

\bibitem{rozenberg_mott-hubbard_1994}
Marcello~J. Rozenberg, Gabriel Kotliar, and X.~Y. Zhang.
\newblock Mott-{Hubbard} transition in infinite dimensions. {II}.
\newblock {\em Physical Review B}, 49(15):10181--10193, 1994.

\bibitem{rubtsov_continuous-time_2005}
Alexei~N. Rubtsov, V.~V. Savkin, and Alexander~I. Lichtenstein.
\newblock Continuous-time quantum {Monte} {Carlo} method for fermions.
\newblock {\em Physical Review B}, 72(3):035122, 2005.

\bibitem{santambrogio_optimal_2015}
Filippo Santambrogio.
\newblock {\em Optimal {Transport} for {Applied} {Mathematicians}: {Calculus}
  of {Variations}, {PDEs}, and {Modeling}}, volume~87 of {\em Progress in
  {Nonlinear} {Differential} {Equations} and {Their} {Applications}}.
\newblock Springer International Publishing, Cham, 2015.

\bibitem{sarachik_resistivity_1968}
Myriam~P. Sarachik.
\newblock Resistivity of {Some} 5d {Elements} and {Alloys} {Containing} {Iron}.
\newblock {\em Physical Review}, 170(3):679--682, 1968.

\bibitem{sarason_generalized_1967}
Donald Sarason.
\newblock Generalized interpolation in {H}$^\infty$.
\newblock {\em Transactions of the American Mathematical Society},
  127(2):179--203, 1967.

\bibitem{sato_gicar_2024}
Ryosuke Sato.
\newblock {GICAR} {Algebras} and {Dynamics} on {Determinantal} {Point}
  {Processes}: {Discrete} {Orthogonal} {Polynomial} {Ensemble} {Case}.
\newblock {\em Communications in Mathematical Physics}, 405(5):115, 2024.

\bibitem{shapiro_fixed-point_2016}
Joel~H. Shapiro.
\newblock {\em A {Fixed}-{Point} {Farrago}}.
\newblock Universitext. Springer International Publishing, Cham, 2016.

\bibitem{stray_interpolating_1991}
Arne Stray.
\newblock Interpolating {Sequences} and the {Nevanlinna} {Pick} {Problem}.
\newblock {\em Publicacions Matemàtiques}, 35(2):507--516, 1991.

\bibitem{sun_extended_2002}
Ping Sun and Gabriel Kotliar.
\newblock Extended dynamical mean-field theory and {GW} method.
\newblock {\em Physical Review B}, 66(8):085120, 2002.

\bibitem{tan_doping_2023}
Yuting Tan, Pak Ki~Henry Tsang, Vladimir Dobrosavljević, and Louk Rademaker.
\newblock Doping a {Wigner}-{Mott} insulator: {Exotic} charge orders in
  transition metal dichalcogenide moiré heterobilayers.
\newblock {\em Physical Review Research}, 5(4):043190, 2023.

\bibitem{titchmarsh_introduction_1948}
Eduard~C. Titchmarsh.
\newblock {\em Introduction to the {Theory} of {Fourier} {Integrals}}.
\newblock Clarendon Press, Oxford, 2 edition, 1948.

\bibitem{villani_optimal_2009}
Cédric Villani.
\newblock {\em Optimal transport. {Old} and {New}}, volume 338 of {\em
  Grundlehren der mathematischen {Wissenschaften}}.
\newblock Springer Berlin Heidelberg, 1 edition, 2009.

\bibitem{walsh_interpolation_1935}
Joseph~L. Walsh.
\newblock {\em Interpolation and {Approximation} by {Rational} {Functions} in
  the {Complex} {Domain}}.
\newblock American Mathematical Soc., 1935.

\bibitem{white_density_1992}
Steven~R. White.
\newblock Density matrix formulation for quantum renormalization groups.
\newblock {\em Physical Review Letters}, 69(19):2863--2866, 1992.

\bibitem{wick_properties_1954}
Gian-Carlo Wick.
\newblock Properties of {Bethe}-{Salpeter} {Wave} {Functions}.
\newblock {\em Physical Review}, 96(4):1124--1134, 1954.

\bibitem{wu_enhancing_2020}
Xiaojie Wu, Michael Lindsey, Tiangang Zhou, Yu~Tong, and Lin Lin.
\newblock Enhancing robustness and efficiency of density matrix embedding
  theory via semidefinite programming and local correlation potential fitting.
\newblock {\em Phys. Rev. B}, 102(8):085123, 2020.

\bibitem{yu_high-temperature_2019}
Yijun Yu, Liguo Ma, Peng Cai, Ruidan Zhong, Cun Ye, Jian Shen, G.~D. Gu,
  Xian~Hui Chen, and Yuanbo Zhang.
\newblock High-temperature superconductivity in monolayer
  {Bi}$_2${Sr}$_2${CaCu}$_2${O}$_{8+\delta}$.
\newblock {\em Nature}, 575(7781):156--163, 2019.

\bibitem{zhang_mott_1993}
X.~Y. Zhang, Marcello~J. Rozenberg, and Gabriel Kotliar.
\newblock Mott transition in the \textit{d} =$\infty$ {Hubbard} model at zero
  temperature.
\newblock {\em Physical Review Letters}, 70(11):1666--1669, 1993.

\bibitem{zhou_unraveling_2020}
Jianqiang~Sky Zhou, Lucia Reining, Alessandro Nicolaou, Azzedine Bendounan,
  Kari Ruotsalainen, Marco Vanzini, J.~J. Kas, J.~J. Rehr, Matthias Muntwiler,
  Vladimir~N. Strocov, Fausto Sirotti, and Matteo Gatti.
\newblock Unraveling intrinsic correlation effects with angle-resolved
  photoemission spectroscopy.
\newblock {\em Proceedings of the National Academy of Sciences},
  117(46):28596--28602, 2020.

\end{thebibliography}
\end{document}